\documentclass[a4paper]{article}

\usepackage{graf_GR}
\usepackage{authblk}
\usepackage{tikz,tkz-tab}

\usepackage{slashed}
\usepackage{relsize,bbm}

\renewcommand{\Re}{\mathrm{Re}}

\newcommand{\alt}{\widetilde{\al}}
\newcommand{\RW}{\mathfrak{R}}

\newcommand{\Ka}{\mathfrak{K}}
\newcommand{\Einfty}{\overset{{\scriptscriptstyle\infty}}{\mathrm{E}}}
\newcommand{\ellmode}{\ell}

\newcommand{\Psic}{\widetilde{\Psi}}

\title{Linear Stability of Schwarzschild-Anti-de Sitter spacetimes II: \\ Logarithmic decay of solutions to the Teukolsky system}

\author[1]{Olivier Graf\thanks{olivier.graf@univ-grenoble-alpes.fr}}
\author[2,3]{Gustav Holzegel\thanks{gholzegel@uni-muenster.de}}

\affil[1]{\small Univ.~Grenoble~Alpes, CNRS, IF, 38000 Grenoble, France \vskip.2pc \ }
\affil[2]{\small Universit\"at M\"unster,
Mathematisches~Institut, Einsteinstrasse~62,~48149~M\"unster,~Bundesrepublik~Deutschland \vskip.2pc \ }
\affil[3]{\small Imperial College London,
Department of Mathematics,
South~Kensington~Campus,~London~SW7~2AZ,~United~Kingdom}
\begin{document}
\maketitle
\begin{abstract}
We prove boundedness and inverse logarithmic decay in time of solutions to the Teukolsky equations on Schwarzschild-Anti-de Sitter backgrounds with standard boundary conditions originating from fixing the conformal class of the non-linear metric on the boundary. The proofs rely on (1) a physical space transformation theory between the Teukolsky equations and the Regge-Wheeler equations on Schwarzschild-Anti de Sitter backgrounds and (2) novel energy and Carleman estimates handling the coupling of the two Teukolsky equations through the boundary conditions thereby generalising earlier work of \cite{Hol.Smu13} for the covariant wave equation. Specifically, we also produce purely physical space Carleman estimates.
%
As shown in our companion paper~\cite{Gra.Hol24b}, the results obtained here are sharp. Finally, the results of the present paper form a central ingredient in our proof of the full linear stability of the Schwarzschild-Anti-de Sitter family under gravitational perturbations presented in~\cite{Gra.Hol}.
\end{abstract}

\hypersetup{linkcolor=black}
\setcounter{tocdepth}{2}
\tableofcontents

\hypersetup{linkcolor=MidnightBlue}

\section{Introduction}\label{sec:intro}
The Teukolsky wave equations on Schwarzschild and Kerr spacetimes, as well as on their de Sitter (dS) and Anti-de Sitter (AdS) counterparts, are the fundamental equations governing the dynamics of linear perturbations of black hole solutions of the vacuum Einstein equations. The equations were originally derived by Teukolsky~\cite{Teu72} in the physics literature and form one of the landmarks of the ``golden age" of black hole physics. In the asymptotically flat case, there is now an extensive mathematical literature on the long time behaviour of general solutions to the Teukolsky equations: The works of ~\cite{Daf.Hol.Rod19,Daf.Hol.Rod19a, Ma20} developed a robust understanding of the boundedness and decay properties of solutions in the Schwarzschild and slowly rotating Kerr case, which forms a key ingredient in the recent proofs of the nonlinear stability of these spacetimes~\cite{Daf.Hol.Rod.Tay21, Gio.Kla.Sze22}. More recently, a treatment of the full-subextremal range $|a|<M$ was provided in \cite{Shl.Tei23}. Detailed asymptotics on the long time behaviour have been proven in \cite{Ma.Zha23,Mil23}.

In the Kerr-dS and Kerr-AdS case, however, only mode stability results have been obtained for the Teukolsky equations, see~\cite{Cas.Tei21} and~\cite{Gra.Hol23}. Nevertheless, for slowly rotating Kerr-dS black holes, non-linear stability has been established  \cite{Hin.Vas18,Fan21,Fan22}. The reason is that the aforementioned papers employ the wave gauge and hence do not require any results about the Teukolsky equations.\footnote{The stronger robustness towards the choice of gauge in the Kerr-dS case may be viewed a consequence of the exponential decay of linear perturbations.} The fate of non-linear perturbations of Kerr-AdS black holes is a major open problem discussed more in our companion paper \cite{Gra.Hol}, to which we also refer the reader for more detailed background and introduction.

In this paper, we shall obtain definite decay results for the evolution of the Teukolsky equations on the Schwarzschild-AdS spacetime. The introduction proceeds directly with the definition of the Schwarzschild-AdS metric and the system of Teukolsky equations in Sections~\ref{sec:defSadS} and \ref{sec:defTeukintro}. The transformation theory for the Teukolsky equations, which is at the algebraic heart of the paper, is reviewed in Section \ref{sec:chandraintro}.
Defining the norms in Section \ref{sec:normsintro} will then allow us to give a precise formulation of our main theorem in Section~\ref{sec:upperboundintro} as well as a brief overview of the proof in Section~\ref{sec:overviewupperintro}. The remainder of the paper, Sections \ref{sec:RWbound}--\ref{sec:RWtoTeuk}, is then concerned with a detailed proof of the main theorem.

\subsection{The Schwarzschild-AdS background}\label{sec:defSadS}
Let $k \geq 0$, $M\geq 0$ be fixed parameters. The family of Schwarzschild-Anti-de Sitter metrics 
%
\begin{align}
  \label{eq:gSadSintro}
  g_{\mathrm{SAdS}} & = -\left(1 + k^2 r^2 -\frac{2M}{r}\right) (\d t^\ast)^2 + \frac{4M}{r(1+k^2 r^2)} dt^\ast \d r +\frac{1 + k^2 r^2 +\frac{2M}{r}}{(1+k^2 r^2)^2} \d r^2 +r^2\le(\d\varth^2+\sin^2\varth\d\varphi^2\ri),
\end{align}
on $\mathcal{M}:=(-\infty,\infty)_{t^\ast} \times [r_+,\infty)_r \times \SSS^2$ (with $r=r_+$ the largest real zero of $\Delta = r^2 + k^2 r^4 - 2Mr $) is the unique spherically symmetric solution of the Einstein equations with negative cosmological constant, $R_{\mu \nu} = -3k^2 g_{\mu \nu}$. The set $\mathcal{H}^+:=\mathcal{M} \cap \{ r= r_+\}$ defines a null boundary of $\mathcal{M}$, called the \emph{future event horizon}. See the Penrose diagram in Figure~\ref{fig:penroseSadS} below for a depiction of the geometry. Its perhaps most distinguishing feature is the existence of a timelike conformal boundary at infinity.

Defining the \emph{radial tortoise coordinate} $r^\ast$ and the time coordinate $t$ by
\begin{align} \label{tortoise}
  \frac{\d r^\ast}{\d r} & := \frac{r^2}{\De}, & r^\ast\le(r=+\infty\ri) & = \frac{\pi}{2}, &  t & := t^\ast - r^\ast + \frac{1}{k}\arctan(k r)
\end{align}
one can express the metric in $(t,r,\varth, \varphi)$ coordinates to obtain the more familiar Schwarzschildean form
\begin{align*}
  g_{\mathrm{SAdS}} =  -\left(1+ k^2r^2 -\frac{2M}{r} \right)\d t^2  +\left(1+k^2r^2-\frac{2M}{r}\right)^{-1} \d r^2  +r^2\le(\d\varth^2+\sin^2\varth\d\varphi^2\ri),
\end{align*}
which is well-defined on the interior of $\mathcal{M}$. We will use both the $(t,r,\theta,\phi)$ and the regular $(t^\ast, r, \theta,\phi)$ coordinate systems as well as the (trivially related by a rescaling of $r$) coordinate systems $(t,r^\ast,\theta,\phi)$ and $(t^\ast, r^\ast, \theta,\phi)$.
\begin{figure}[h!]
  \centering
  \includegraphics[height=6cm]{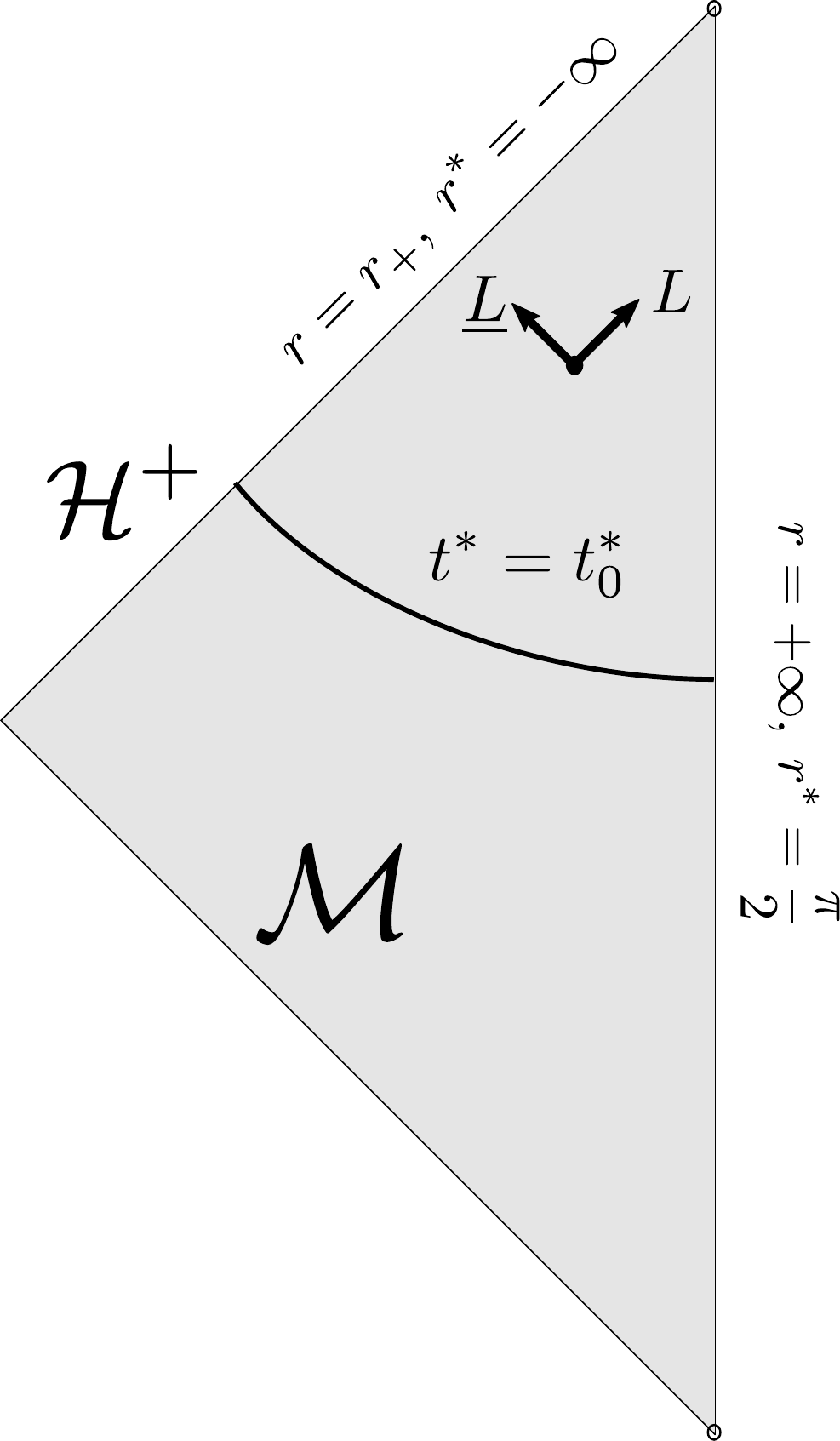}
  \caption{Penrose diagram of the Schwarzschild-AdS spacetime}
  \label{fig:penroseSadS}
\end{figure}

We define the pair of future directed null vectorfields $L,\Lb$, expressed in the respective coordinates by
\begin{subequations}\label{eq:defLLb}
  \begin{align}
    L & := \pr_t |_{(t,r,\varth,\varphi)}+\frac{\De}{r^2}\pr_r |_{(t,r,\varth,\varphi)}= \frac{\De_+}{\De_0}\pr_{t^\ast} |_{(t^\ast,r,\varth,\varphi)}+ \frac{\De}{r^2}\pr_r|_{(t^\ast,r,\varth,\varphi)}   \, , \\ 
  \Lb & :=  \pr_t |_{(t,r,\varth,\varphi)}+\frac{\De}{r^2}\pr_r |_{(t,r,\varth,\varphi)}=  \frac{\De_-}{\De_0}\pr_{t^\ast} |_{(t^\ast,r,\varth,\varphi)}- \frac{\De}{r^2}\pr_r|_{(t^\ast,r,\varth,\varphi)} 
  \, ,   
  \end{align}
 \end{subequations}
  where, here and in the sequel, we abbreviate 
\begin{align} \label{wdef}
  \De_\pm & := r^2+k^2r^4\pm 2Mr, & \De_0 & := r^2+k^2r^4, & w=\frac{\Delta}{r^4} .
\end{align}
Note that $\Delta_-=\Delta$ and that the vectorfields $\Delta^{-1} \underline{L}$ and $L$ extend smoothly to the event horizon $\mathcal{H}^+$. 

\subsection{The Teukolsky equations}\label{sec:defTeukintro}

The Teukolsky equations on Schwarzschild-Anti-de Sitter spacetime $(\mathcal{M}, g_{\mathrm{SAdS}})$ take the following form, which is taken from \cite{Kha83}. Alternatively, the reader can consult our companion paper \cite{Gra.Hol}) for a derivation of (\ref{eq:Teukoriginal}) from the linearised Einstein equations in double null gauge.
\begin{subequations}\label{eq:Teukoriginal}
  \begin{align}
    \begin{aligned}
      0 & = \Box_{g_{\mathrm{SAdS}}}\al^{[+2]} + \frac{2}{r^2}\frac{\d\De}{\d r}\pr_r\al^{[+2]}  + \frac{4}{r^2}\le(\frac{r^2}{2\De}\frac{\d\De}{\d r}-2r\ri)\pr_t\al^{[+2]} + \frac{4}{r^2}i\frac{\cos\varth}{\sin^2\varth}\pr_\varphi\al^{[+2]}\\
      & \quad + \frac{2}{r^2}\le(1+15k^2r^2-2\cot^2\varth\ri)\al^{[+2]},
    \end{aligned}
  \end{align}
  and
  \begin{align} 
      \begin{aligned}
      0 & = \Box_{g_{\mathrm{SAdS}}}\al^{[-2]} - \frac{2}{r^2}\frac{\d\De}{\d r}\pr_r\al^{[-2]}  - \frac{4}{r^2}\le(\frac{r^2}{2\De}\frac{\d\De}{\d r}-2r\ri)\pr_t\al^{[-2]} - \frac{4}{r^2}i\frac{\cos\varth}{\sin^2\varth}\pr_\varphi\al^{[-2]}\\
      & \quad + \frac{2}{r^2}\le(-1+3k^2r^2-2\cot^2\varth\ri)\al^{[-2]},
    \end{aligned}
  \end{align}
    where $\Box_g$ is the d'Alembertian operator associated to the metric $g=g_{\mathrm{SAdS}}$,
  \begin{align*}
  \Box_{g_{\mathrm{SAdS}}} & = -\frac{r^2}{\De}\pr_t^2 + \frac{1}{r^2}\pr_r\le(\De\pr_r\ri) + \frac{1}{r^2\sin\varth}\pr_\varth\le(\sin\varth\pr_\varth\ri) + \frac{1}{r^2\sin^2\varth}\pr_\varphi^2,
  \end{align*}
\end{subequations}
and $\al^{[+2]}$, $\al^{[-2]}$ are complex-valued spin-weighted functions of weight $+2$ and $-2$ respectively.\footnote{See \cite{Gra.Hol23} for a definition. For practical purposes the unfamiliar reader may think of $\al^{[+2]}$, $\al^{[-2]}$ simply as complex scalars.} In view of~\eqref{eq:Teukoriginal}, we define $\LL^{[\pm2]}$ to be the following angular operators
\begin{align*}
  -\LL^{[\pm2]} & := \frac{1}{\sin\varth}\pr_\varth\le(\sin\varth\,\pr_\varth\ri) + \frac{1}{\sin^2\varth}\pr_\varphi^2 + 2(\pm2)i\frac{\cos\varth}{\sin^2\varth}\pr_\varphi - 4\cot^2\varth -4.
\end{align*}
\begin{remark} \label{rem:eigenbasis}
As is well-known, the operators $\LL^{[\pm2]}$ admit a Hilbert basis of eigenfunctions $\le(e^{\pm im\varphi}S_{m\ellmode}(\varth)\ri)_{\ellmode\geq 2,|m|\leq\ellmode}$ with eigenvalues $\ellmode(\ellmode+1)$. See for instance Section 6.1 of \cite{Daf.Hol.Rod19a} and references therein.
\end{remark}

We now briefly discuss the well-posedness theory of the system (\ref{eq:Teukoriginal}) on $\mathcal{M}$ to introduce the class of solutions we would like to study. The solutions will be regular at the horizon and infinity in the following sense:
\begin{definition} \label{def:regular}
A smooth spin-weighted function $\phi$ on $\mathcal{M} \setminus \mathcal{H}^+$ is called regular at the future event horizon if
  \begin{align}\label{eq:defreghorgeneral}
    \begin{aligned}
  \sup_{\mathcal{M} \setminus \mathcal{H}^+ \cap \left[t_1^\ast,t_2^\ast\right] \cap  \{r \leq 3M\}} |    L^p\le(\De^{-1}\Lb\ri)^q\phi | < \infty \ \ \ \textrm{ holds for all $t^\ast_1,t^\ast_2 \in \mathbb{R}$ and all $p,q\in\mathbb{N}$. }
    \end{aligned}
  \end{align}
Similarly, a smooth spin weighted function $\phi$ on $\mathcal{M} \setminus \mathcal{H}^+$ is called regular at infinity if
  \begin{align}\label{eq:defreginfgeneral}
    \begin{aligned}
  \sup_{\mathcal{M}  \cap \left[t_1^\ast,t_2^\ast\right] \cap  \{r \geq 8M\}} |    (r^2 (L-\underline{L}))^p(L+\underline{L})^q \phi| < \infty \ \ \ \textrm{ holds for all $t^\ast_1,t^\ast_2 \in \mathbb{R}$ and all $p,q\in\mathbb{N}$. }
    \end{aligned}
  \end{align}
\end{definition}
It is not hard to see that with these definitions, the quantities inside the norm in (\ref{eq:defreghorgeneral}) and (\ref{eq:defreginfgeneral}) extend continuously for all $p,q \in \mathbb{N}$ to the future event horizon and to the conformal boundary respectively.

\begin{definition} \label{def:alparel}
We will call a solution $(\al^{\pm2})$ of (\ref{eq:Teukoriginal}) on $\mathcal{M} \setminus \mathcal{H}^+$ \underline{future regular} if defining 
\begin{align} \label{refo}
  \widetilde{\al}^{[+2]} & := \De^2r^{-3}\al^{[+2]}, & \widetilde{\al}^{[-2]} & := r^{-3}\al^{[-2]},
\end{align}
then $\alt^{[+2]}$ and $\De^{-2}\alt^{[-2]}$ are regular at the future event horizon $\mathcal{H}^+$ and  $\alt^{[+2]}$ and $\alt^{[-2]}$ are regular at infinity.
\end{definition}

The presence of the conformal boundary at infinity requires to impose boundary conditions on the solution: 

\begin{definition}\label{def:sysTeuk}
We will say that a solution $(\al^{[\pm2]})$ of (\ref{eq:Teukoriginal}) on $\mathcal{M} \setminus \mathcal{H}^+$ satisfies \underline{conformal Teukolsky Anti-de Sitter} \underline{boundary conditions at infinity} if (with the $\ast$ denoting complex conjugation)
\begin{subequations}\label{eq:TeukBC}
\begin{align}
  \widetilde{\al}^{[+2]} - \left(\widetilde{\al}^{[-2]}\right)^\ast & \xrightarrow{r\to+\infty} 0,\label{eq:TeukBCDirichlet} \\
  r^2\pr_r\widetilde{\al}^{[+2]} + r^2\pr_r\big(\widetilde{\al}^{[-2]}\big)^\ast  & \xrightarrow{r\to+\infty} 0.\label{eq:TeukBCNeumann}
\end{align}
\end{subequations}
Finally, we will say that a smooth solution $(\al^{[\pm2]})$ of (\ref{eq:Teukoriginal}) on $\mathcal{M} \setminus \mathcal{H}^+$ is a \underline{solution of the Teukolsky problem} on $\mathcal{M}$ if $(\al^{[\pm2]})$ is future regular and satisfies the boundary conditions~\eqref{eq:TeukBC}.
\end{definition}
See~our \cite{Gra.Hol} and also \cite{Hol.Luk.Smu.War20} for a more detailed discussion and derivation of these boundary conditions. Here it suffices to say that these boundary conditions can be derived from the requirement of metric perturbations preserving the conformal class of the Anti-de Sitter metric at infinity.

\begin{remark}
In the well-posedness theory, solutions to the Teukolsky problem are constructed to the future of a fixed constant $t^\ast$-slice, $\Sigma_{t^\ast_0}$, on which smooth initial data are prescribed. One can of course replace $\mathcal{M}$ by $\mathcal{M} \cap J^+ \left(\Sigma_{t^\ast=t^\ast_0}\right)$ in all of the above definitions to accommodate solutions defined only on $\mathcal{M} \cap J^+ \left(\Sigma_{t^\ast=t^\ast_0}\right)$.
\end{remark}

Finally, to specify initial data on $\Sigma_{t^\ast_0}$ one prescribes the functions ($\al^{[\pm2]}$, $\partial_t \al^{[\pm2]}$) as smooth spin-weighted functions from $\Sigma_{t^\ast_0}$ to $\mathbb{C}$. The equations \eqref{eq:Teukoriginal} will then determine all derivatives of the prospective solution on $\Sigma_{t^\ast_0}$ in terms of the data ($\al^{[\pm2]}$, $\partial_t \al^{[\pm2]}$), as $\Sigma_{t^\ast_0}$ is non-characteristic. One hence may check whether the solution is future regular \emph{on} $\Sigma_{t^\ast_0}$ in the sense of Definition \ref{def:regular} and whether the boundary conditions  (\ref{eq:TeukBC}) are satisfied  \emph{on} $\Sigma_{t^\ast_0}$. If both is the case, we call the prescribed data ($\al^{[\pm2]}$, $\partial_t \al^{[\pm2]}$) admissible. The well-posedness theorem for (\ref{eq:Teukoriginal}) can then be stated as follows; see Theorem 1.4 of \cite{Gra.Hol23}.
\begin{theorem}[Well-posedness]\label{thm:IBVP}
  Let $t^\ast_0\in\RRR$. Given smooth admissible initial data ($\al^{[\pm2]}$, $\partial_t \al^{[\pm2]}$) on $\Sigma_{t^\ast_0}$, there exists a unique solution $(\al^{[\pm2]})$ to the Teukolsky problem on $\mathcal{M} \cap J^+ \left(\Sigma_{t^\ast=t^\ast_0}\right)$ assuming the given data.
\end{theorem}

It is the global in time behaviour of solutions arising from Theorem \ref{thm:IBVP} that we wish to study in this paper. 

\subsection{Chandrasekhar transformations and the Regge-Wheeler equations}\label{sec:chandraintro}
We review here the generalisation of the physical space transformation theory introduced for the Teukolsky equation in the asymptotically flat case in \cite{Daf.Hol.Rod19}. 
\begin{definition}\label{def:Chandra}
Recall from (\ref{wdef}) the weight $w=\frac{\Delta}{r^4}$. Given a smooth solution $\al^{[\pm 2]}$ of the Teukolsky problem, the \emph{Chandrasekhar transformations} of $\widetilde{\al}^{[\pm2]}$ are the functions $\psi^{[\pm2]},\Psi^{[\pm2]}$ given by
  \begin{subequations}\label{eq:Chandra}
    \begin{align}
      \psi^{[+2]} & := w^{-1}\Lb\widetilde{\al}^{[+2]}, & \Psi^{[+2]} & := w^{-1}\Lb\psi^{[+2]},\label{eq:Chandrap2}\\
      \psi^{[-2]} & := w^{-1}L\widetilde{\al}^{[-2]}, & \Psi^{[-2]} & := w^{-1}L\psi^{[-2]}.\label{eq:Chandram2}
    \end{align}
  \end{subequations}
\end{definition}
It is straightforward to verify that if $(\al^{[\pm2]})$ is future regular, then $\Psi^{[+2]},\Psi^{[-2]}$ (and $\Psi^D,\Psi^R$ defined below) are also regular at the horizon and infinity in the sense of Definition \ref{def:regular}.
Moreover, as one readily checks, the transformations~\eqref{eq:Chandra} map solutions $\al^{[\pm 2]}$ of the Teukolsky equations~\eqref{eq:Teukoriginal} to quantities $\Psi^{[\pm 2]}$, which satisfy the Regge-Wheeler equations
\begin{align}\label{eq:RW}
  \RW^{[\pm2]} \Psi^{[\pm 2]}:= -L\Lb\Psi^{[\pm2]}- \frac{\De}{r^4}\le(\LL^{[\pm2]}-\frac{6M}{r}\ri)  \Psi^{[\pm 2]}= 0.
\end{align}
Remarkably, as in the Schwarzschild case \cite{Daf.Hol.Rod19}, the Regge-Wheeler equation (\ref{eq:RW}) does not couple to the quantities $\widetilde{\al}^{[\pm2]}$ and $\psi^{[\pm2]}$. However, unlike in the asymptotically flat case, the boundary values for the $\Psi^{[\pm 2]}$ are now coupled to one another \emph{and} to the lower order Teukolsky boundary values as follows:
\begin{align}\label{bcintro}
  \Psi^{[+2]}-(\Psi^{[-2]})^\ast  & \xrightarrow{r\to+\infty} 0, & r^2 \partial_r \Psi^{[+2]} + r^2 \partial_r (\Psi^{[-2]})^\ast & \xrightarrow{r\to+\infty} -6M \widetilde{\al}^{[+2]} -6M (\widetilde{\al}^{[-2]})^\ast. 
\end{align}
See Proposition \ref{prop:ChandraAux} below for a summary of the wave equations and boundary conditions satisfied by the $\widetilde{\al}^{[\pm 2]}$, $\psi^{[\pm 2]}$ and $\Psi^{[\pm 2]}$. 
Using~\eqref{bcintro}, (\ref{eq:RW}) we can show that $\Psi^{[\pm 2]}$ also satisfies higher order boundary conditions
\begin{subequations}\label{eq:RWBC}
  \begin{align}
    \Psi^{[+2]}-(\Psi^{[-2]})^\ast & \xrightarrow{r\to+\infty} 0,\label{eq:RWBCD}\\
    LL\Psi^{[+2]} + \Lb\Lb(\Psi^{[-2]})^\ast + \frac{1}{6M}\LL\le(\LL-2\ri)\le(L\Psi^{[+2]} - \Lb(\Psi^{[-2]})^\ast\ri) & \xrightarrow{r\to+\infty} 0,\label{eq:RWBCN}                    
  \end{align}
\end{subequations}
which eliminate the original $\widetilde{\al}^{[\pm 2]}$ from the equations.\footnote{Here and in the following $\LL := \LL^{[+2]}$. Note that $\big(\LL^{[-2]}\Psi\big)^\ast = \LL^{[+2]}\Psi^\ast $.} Moreover, if we define 
\begin{align}\label{eq:defPsiDR}
    \begin{aligned}
    \Psi^D & := \Psi^{[+2]} - \big(\Psi^{[-2]}\big)^\ast, \\
    \Psi^R & := \le(\Psi^{[+2]} + \big(\Psi^{[-2]}\big)^\ast\ri) +12M \left(\LL(\LL-2)\right)^{-1} \pr_t\le(\Psi^{[+2]} - \big(\Psi^{[-2]}\big)^\ast\ri),
    \end{aligned}
\end{align}
the boundary conditions~\eqref{eq:RWBC} for $\Psi^{[+2]}$ and $\Psi^{[-2]}$ decouple as
\begin{subequations}\label{eq:RWBCdeco}
  \begin{align}
    \Psi^D & \xrightarrow{r\to +\infty} 0, \label{eq:RWBCdecoDirichlet}\\
    2\pr_{t}^2\Psi^R + \frac{\LL(\LL-2)}{6M}\pr_{r^\ast}\Psi^R + k^2\LL\Psi^R & \xrightarrow{r\to+\infty} 0. \label{eq:RWBCdecoRobin}
  \end{align}
\end{subequations}
The first condition~\eqref{eq:RWBCdecoDirichlet} is obviously a Dirichlet boundary condition, while the second condition~\eqref{eq:RWBCdecoRobin} can be interpreted as a ``Robin'' condition for each time and angular frequency.

In conclusion, one may therefore study the following \emph{decoupled Regge-Wheeler problem} for $\Psi^{D}$ (resp. $\Psi^R$), independently from the Teukolsky problem of Definition~\ref{def:sysTeuk}, i.e.~study the following class of solutions:
\begin{align}
  \label{sys:RWdeco}
  \begin{aligned}
    \text{$\Psi^{D}$ (resp. $\Psi^R$) satisfies the Regge-Wheeler equations~\eqref{eq:RW} on $\mathcal{M}$ (or $\mathcal{M} \cap \{t^\ast \geq t^\ast_0\}$)},\\
    \text{$\Psi^{D}$ (resp. $\Psi^R$) is regular at the horizon and infinity in the sense of Definition \ref{def:regular}},\\
    \text{$\Psi^{D}$ (resp. $\Psi^R$) satisfies the boundary condition~\eqref{eq:RWBCdecoDirichlet} (resp.~\eqref{eq:RWBCdecoRobin})}.
  \end{aligned}
\end{align}

\begin{remark}\label{rem:TSid}
  When $\al^{[+2]},\al^{[-2]}$ are the Teukolsky quantities associated to a solution to the linearised gravity equations, $\al^{[+2]},\al^{[-2]}$ satisfy additional, coupled, so-called \emph{Teukolsky-Starobinsky identities}. They will not play a role in this paper. 
\end{remark}

\begin{remark}
Unlike in the asymptotically flat case, the Chandrasekhar transformations~\eqref{eq:Chandra} are actually injective on the space of solutions to the Teukolsky problem. See~\cite{Gra.Hol24b} for further discussion.
\end{remark}

\subsection{Norms and Energies}\label{sec:normsintro}
Given a spin-$\pm2$-weighted function $f$ on $\mathcal{M}$, we define $f_{m\ellmode}$ to be the coefficients of $f$ on the angular Hilbert basis $\le(e^{\pm im\varphi}S_{m\ellmode}(\varth)\ri)_{\ellmode\geq 2,|m|\leq\ellmode}$ of $\LL^{[\pm2]}$, see Remark \ref{rem:eigenbasis}.

For all $\tilde{t}^\ast \in\mathbb{R}$ we write $\Si_{\tilde{t}^\ast} := \{t^\ast=\tilde{t}^\ast\}$ for the spacelike slices of constant $t^\ast=\tilde{t}^\ast$. For all $\tilde{r} \in[r_+,+\infty]$, we write $S_{\tilde{t}^\ast, \tilde{r}} := \{t^\ast=\tilde{t}^\ast,~r=\tilde{r}\}$ for the spheres at constant $t^\ast=\tilde{t}^\ast$ and $r=\tilde{r}$. 

We now define the spin-weighted Sobolev norms on the spheres $S_{{t}^\ast, {r}}$ and the slices $\Si_{t^\ast}$. For later purposes in the paper, it is actually most convenient to define them in (angular) frequency space. We define for $s\in \mathbb{R}$
\begin{align*}
  \norm{f}^2_{H^s(S_{t^\ast,r})} :=  \sum_{\ellmode\geq 2} \sum_{|m|\leq\ellmode} \ellmode^{2s}|f_{m\ellmode}(t^\ast,r)|^2 \, .
\end{align*}
Note that for $s =n \in \mathbb{N}$, the above norm is equivalent to the standard physical space Sobolev norm for spin-$\pm2$-weighted functions on the spheres $S^2_{t^\ast,r}$ (denoted ${}^{[2]}H^n(\sin \vartheta dr d\vartheta d\varphi)$  in~\cite[Section 2.2]{Daf.Hol.Rod19a}).\footnote{Note the measure is that of the \emph{unit} sphere, not the geometrically induced one. The difference is a factor of $r^{2}$.}


Recalling the definitions (\ref{tortoise}) and (\ref{wdef}) we also define 
\begin{align*}
  \norm{f}_{L^2(\Si_{t^\ast})}^2 & := \sum_{\ellmode\geq 2} \sum_{|m|\leq\ellmode} \int_{r_+}^{+\infty} |f_{m\ellmode}(t^\ast,r)|^2\,\frac{\d r}{r^2} = \sum_{\ellmode\geq 2} \sum_{|m|\leq\ellmode} \int_{-\infty}^{\frac{\pi}{2}} |f_{m\ellmode}(t^\ast,r^\ast)|^2\,w\d r^\ast \nonumber \ , \\
  \norm{f}^2_{H^1(\Si_{t^\ast})} & :=\sum_{\ellmode\geq 2} \sum_{|m|\leq\ellmode}\int_{r_+}^{+\infty}\le(\big|(r^2\pr_r)f_{m\ellmode}(t^\ast,r)\big|^2+ \ellmode^2\big|f_{m\ellmode}(t^\ast,r)\big|^2\ri)\, \frac{\d r}{r^2} \\
                                 & = \sum_{\ellmode\geq 2} \sum_{|m|\leq\ellmode}\int_{-\infty}^{\frac{\pi}{2}}\le(w^{-1}\big|\pr_{r^\ast}f_{m\ellmode}(t^\ast,r^\ast)\big|^2+ w\ellmode^2\big|f_{m\ellmode}(t^\ast,r^\ast)\big|^2\ri)\, \d r^\ast.
\end{align*}
These norms are again easily shown to be equivalent to the standard physical space norms for spin-$\pm2$-weighted functions ${}^{[2]}L^2(\Si_{t^\ast}, \frac{1}{r^2} \sin \vartheta dr d\vartheta d\varphi)$ and ${}^{[2]}H^1(\Si_{t^\ast}, \frac{1}{r^2} \sin \vartheta dr d\vartheta d\varphi)$ respectively. See~\cite[Section 2.2]{Daf.Hol.Rod19a}.

We next define the following energy-type norm for the (Teukolsky) quantities $(\alt^{[+2]},\alt^{[-2]})$ on $\Si_{t^\ast}$ :
\begin{align} \label{Tenergy}
  \begin{aligned}
    \mathrm{E}^{\mathfrak{T}}[\alt](t^\ast) & := \big\Vert\pr_{t^\ast}\alt^{[+2]}\big\Vert_{L^2(\Si_{t^\ast})}^2 + \big\Vert w^{-2}\pr_{t^\ast}\alt^{[-2]}\big\Vert_{L^2(\Si_{t^\ast})}^2+\big\Vert\alt^{[+2]}\big\Vert_{H^{1}(\Si_{t^\ast})}^2 + \big\Vert w^{-2}\alt^{[-2]}\big\Vert_{H^{1}(\Si_{t^\ast})}^2 \, . 
  \end{aligned}
\end{align}
We define the following energy-type norm with boundary terms for the (Regge-Wheeler) quantities $\Psi$ on $\Si_{t^\ast}$:
\begin{align} \label{Renergy}
  \begin{aligned}
    \mathrm{E}^{\mathfrak{R}}[\Psi](t^\ast) & := \norm{\pr_{t^\ast}\Psi}^2_{L^2(\Si_{t^\ast})} + \norm{\Psi}^2_{H^1(\Si_{t^\ast})} + \norm{\pr_{t^\ast}\Psi}^2_{H^{-2}(S_{t^\ast,\infty})} + \norm{\Psi}^2_{H^{-1}(S_{t^\ast,\infty})}.
  \end{aligned}
\end{align}
Given any energy $\mathrm{E}$ (such as $\mathrm{E}^{\mathfrak{T}}[\alt]$ or $\mathrm{E}^{\mathfrak{R}}[\Psi]$ above), we define for $n\geq 1$ the higher order commuted energies\footnote{Note that our definition of the norms and energies allow commutation with the operator $\mathcal{L}$ raised to half-integer powers. We could alternatively (and more cumbersome notationally)  commute with the spin-weighted angular momentum vectorfields.}
\begin{align*}
  \mathrm{E}^n[\Phi]  :=  \sum_{i\leq n-1} \sum_{k_0+k_1=i} \mathrm{E}\le[\LL^{k_0/2}\pr_t^{k_1}\Phi\ri] && \text{and} &&  \overline{\mathbb{E}}^n[\Phi]  := \sum_{i\leq n-1} \sum_{k_0+k_1+k_2=i} \mathrm{E}\le[ \LL^{k_0/2}\pr_t^{k_1} (w^{-1}\underline{L})^{k_2}\Phi\ri] .
\end{align*}
Note that $\overline{\mathrm{E}}^n[\Phi]$ controls higher order derivatives transversal to the horizon and hence all regular derivatives. 

Finally, for all $n\geq 3$, using directly the above convention we define the following combined energy
\begin{align}
  \begin{aligned} \label{combiE}
    \mathrm{E}^{\mathfrak{T},\mathfrak{R},n}[\alt] & := \mathrm{E}^{\mathfrak{T},n}[\alt] + \mathrm{E}^{\mathfrak{R},n-2}\le[\Psi^D\ri] + \mathrm{E}^{\mathfrak{R},n-2}\le[\Psi^R\ri] + \mathrm{E}^{\mathfrak{R},n-2}\le[\LL^{-1}(\LL-2)^{-1}\pr_t\Psi^D\ri],
  \end{aligned}
\end{align}
where we recall that $\Psi^D,\Psi^R$ are combinations~\eqref{eq:defPsiDR} of the Chandrasekhar transformations of $\alt^{[\pm2]}$ defined by~\eqref{eq:Chandra}, see Section~\ref{sec:chandraintro}. Similarly, the energy $\overline{\mathbb{E}}^{\mathfrak{T},\mathfrak{R},n}[\alt]$ is defined replacing $\mathrm{E}$ by $\overline{\mathbb{E}}$ everywhere in (\ref{combiE}).

We finally define the energy associated with an individual angular mode. Given any energy $E$ (such as $\mathrm{E}^{\mathfrak{T}}[\alt]$ or $\mathrm{E}^{\mathfrak{R}}[\Psi]$ or the higher order energies introduced above) we set
\begin{align} \label{individualmodeE}
\mathrm{E}_{m \ell} [\Phi] := \mathrm{E}[\Phi_{m \ell} e^{\pm i m \varphi} S_{m \ell}(\vartheta)] .
\end{align}

For convenience, we recapitulate these and also the energies introduced later in the paper in Appendix~\ref{sec:energynorms}.
\begin{remark}\label{rem:combinednorm}
Some of the terms appearing in the energy (\ref{combiE}) are redundant, in fact one may easily verify
  \begin{align}\label{est:remcombinednorm}
    \mathrm{E}^{\mathfrak{T},\mathfrak{R},3}[\alt](t^\ast) & \simeq_{M,k} \mathrm{E}^{\mathfrak{T},3}[\alt](t^\ast) + \norm{\LL^{-2}\pr_t\Psi^R}^2_{H^{-2}(S_{t^\ast,\infty})} + \mathrm{E}^{\mathfrak{R}}\le[\LL^{-1}(\LL-2)^{-1}\pr_t\Psi^D\ri](t^\ast) 
  \end{align}
  and obvious generalisations replacing $3$ by $n>3$. The last two terms in~\eqref{est:remcombinednorm} cannot be controlled by the first and must necessary be included in the energy  for $\alt$ accounting for the specific definition of $\mathrm{E}^{\mathfrak{T},\mathfrak{R},n}$. The penultimate term in~\eqref{est:remcombinednorm} will arise from the boundary terms in the energy identity for the decoupled Regge-Wheeler quantity $\Psi^R$ (see Section~\ref{sec:boundednessintro}) and the last term from relating $\Psi^R$ to $\Psi^D$ and $\alt$ via~\eqref{eq:defPsiDR}. Note also that if at least one angular derivative is applied, the last two terms in~\eqref{est:remcombinednorm}  \underline{can} be absorbed by the first:
\begin{align} \label{equi}
\mathrm{E}^{\mathfrak{T},\mathfrak{R},n}[\LL^{m/2}\alt](t^\ast) \lesssim \mathrm{E}^{\mathfrak{T},n+m}[\alt](t^\ast)  \ \ \ \textrm{for $m >1$ and $n\geq 3$.} 
\end{align}
\end{remark}

\subsection{The main theorem: Boundedness and logarithmic decay}\label{sec:upperboundintro}
The following theorem is the main result of the paper. To state it, we agree on the following standard convention:
\begin{align*}
  A \les_{a,b} B \Leftrightarrow \exists C > 0,~\text{depending only on $a,b$, such that} ~ A \leq C B, 
\end{align*}
and ``$C = C(a,b)>0$'', as a short-hand for ``$C$ is a positive constant depending only on $a,b$''.

\begin{theorem}[Main theorem]\label{thm:main1}
  Let $\al^{[\pm2]}$ be solutions of the Teukolsky problem arising from Theorem \ref{thm:IBVP}. Then the $\alt^{[\pm2]}$ defined by (\ref{refo}) satisfy the following estimates.
  
  \begin{itemize}
\item {\bf Boundedness in time:}
  \begin{align} \label{mtheob}
    \begin{aligned}
      \mathrm{E}^{\mathfrak{T},\mathfrak{R},n}[\alt](t^\ast) & \les_{M,k} \mathrm{E}^{\mathfrak{T},\mathfrak{R},n}[\alt](t^\ast_0)  
    \end{aligned}
  \end{align}
 hold for all $t^\ast\geq t^\ast_0$ and for all $n>2$.
  \item {\bf Inverse logarithmic decay in time:}
  \begin{align}\label{est:main1log}
    \begin{aligned}
      \mathrm{E}^{\mathfrak{T},\mathfrak{R},n}[\alt](t^\ast) & \les_{M,k} \frac{\mathrm{E}^{\mathfrak{T},\mathfrak{R},n}[\LL^{m/2}\alt](t^\ast_0)}{\le(\log(t^\ast-t^\ast_0)\ri)^{2m}}  
    \end{aligned}
  \end{align}
  holds for all $t^\ast>t^\ast_0$, and for all $n > 2$ and any $m \in\mathbb{N}$.
  \item {\bf Exponential decay of each fixed angular mode} with the rate degenerating exponentially in $\ell$:
   \begin{align}\label{est:main1expo}
    \begin{aligned}
      \mathrm{E}_{m\ellmode}^{\mathfrak{T},\mathfrak{R},n}[\alt](t^\ast) & \leq \mathrm{E}_{m\ellmode}^{\mathfrak{T},\mathfrak{R},n}[\alt](t^\ast_0) \exp\le(e^{-C\ellmode}\le(t^\ast_0-t^\ast\ri)\ri),
    \end{aligned}
  \end{align}
holds  for all $t^\ast>t^\ast_0$, for all $\ellmode\geq 2$, $|m|\leq \ellmode$, and all $n>2$, and with $C=C(M,k)>0$.
\end{itemize}
 \end{theorem}
 
We remark that by standard techniques involving the commuted redshift effect \cite{Daf.Rod08} one can easily strengthen the estimates in the theorem to include the higher order \emph{non-degenerate} energies $\overline{\mathbb{E}}^{\mathfrak{T},\mathfrak{R},n}[\alt]$ introduced after (\ref{combiE}). Since this is standard, we state it without proof as the following corollary. 
\begin{corollary}
Under the assumptions of Theorem \ref{thm:main1}, the estimates (\ref{mtheob}) and (\ref{est:main1log}) remain true if one replaces all energies $\mathrm{E}^{\mathfrak{T},\mathfrak{R},n}[\alt]$ by the non-degenerate energies $\overline{\mathbb{E}}^{\mathfrak{T},\mathfrak{R},n}[\alt]$. 
\end{corollary}

From (\ref{est:main1log}) and (\ref{equi}) we immediately deduce the following statement:
\begin{corollary}
 Under the assumptions of Theorem \ref{thm:main1} we have for all $t^\ast>t^\ast_0$ and all $n_0 > n > 2$:
  \begin{align}\label{est:main1logbis}
    \begin{aligned}
      \mathrm{E}^{\mathfrak{T},n}[\alt](t^\ast) & \les_{M,k} \frac{\mathrm{E}^{\mathfrak{T},n_0}[\alt](t^\ast_0)}{\le(\log(t^\ast-t^\ast_0)\ri)^{2(n_0-n)}}.
    \end{aligned}
  \end{align}
  \end{corollary}
From standard Sobolev embeddings, we infer moreover the following pointwise decay: 
  \begin{corollary}
   Under the assumptions of Theorem \ref{thm:main1} we have  for all $t^\ast>t^\ast_0$, $r\in[r_+,+\infty)$, $\om\in\SSS^2$ and all $n_0>3$ the estimates
  \begin{align}\label{est:pointwiseremark}
    \big|\alt^{[+2]}(t^\ast,r,\om)\big|^2 + \big|w^{-2}\alt^{[-2]}(t^\ast,r,\om)\big|^2 & \les_{M,k} \frac{\mathrm{E}^{\mathfrak{T},n_0}[\alt](t^\ast_0)}{\le(\log(t^\ast-t^\ast_0)\ri)^{2(n_0-3)}}.
  \end{align}
\end{corollary}

\begin{remark}
  Inverse logarithmic decay as exhibited in Theorem \ref{thm:main1} is a general and robust feature of the solutions to the wave-type equations for which energy boundedness holds and for which at least part of the energy can leave the spacetime. It was first obtained on product spacetimes $\mathbb{R}_t\times\NN$ in~\cite{Bur98}, under the assumption that the manifold $\NN$ is diffeomorphic to $\mathbb{R}^d$, $d\geq 2$, with finitely many obstacles, and exactly $\mathbb{R}^d$ outside a compact set. The result of~\cite{Bur98} was generalised in~\cite{Car.Vod02,Car.Vod04} and in~\cite{Rod.Tao15} to more general asymptotically conical Riemannian manifolds $\NN$ and in~\cite{Mos16} to the scalar wave equation on a general class of stationary asymptotically flat spacetimes. For Kerr-Anti-de Sitter spacetimes, the analogous result for the wave equation with Dirichlet conditions was obtained in~\cite{Hol.Smu13}. While in the asymptotically conical/flat manifolds/spacetimes, the energy  escapes through infinity, in Kerr-AdS, the energy leaks through the event horizon of the black hole.  
 \end{remark}

\subsection{Overview of the proof}\label{sec:overviewupperintro}
The proof of Theorem~\ref{thm:main1} has essentially four steps. The first one, already described in detail Section \ref{sec:chandraintro}, is purely algebraic. It consists in deriving conservative wave equations, the Regge-Wheeler equations~\eqref{eq:RW}, from the Teukolsky equations via the transformations of Definition \ref{def:Chandra}. As we have seen, a notable difference with the asymptotically flat case is that the relevant Regge-Wheeler quantities $\Psi^D,\Psi^R$ are non-trivial combinations of the original Chandrasekhar transformed quantities (involving higher order operators, recall (\ref{eq:defPsiDR})), a feature necessitated by the desire to decouple also the boundary conditions.  

The remaining steps are analytic and discussed in the remainder of this section. The overall idea is to exploit that the Regge-Wheeler equations are conservative wave equations and that we can hence adapt the techniques for the standard wave equation.  In Step 2, discussed in Section~\ref{sec:boundednessintro},  one establishes boundedness and coercivity of the energy of the Regge-Wheeler quantities $\Psi^D,\Psi^R$. In Step 3, discussed in Section~\ref{sec:logdecayintro}, one uses the energy boundedness, to prove integral Carleman-type estimates. In the final Step 4, discussed in Section~\ref{sec:SolvingChandraintro}, one revisits the Chandrasekhar transformations to deduce energy boundedness and integral estimates for the original Teukolsky quantities, from which logarithmic decay follows by an interpolation argument.

\subsubsection{Energy boundedness of $\Psi^D$ and $\Psi^R$}\label{sec:boundednessintro}
The Dirichlet and the Robin boundary conditions~\eqref{eq:RWBCdecoDirichlet} and~\eqref{eq:RWBCdecoRobin} for $\Psi^{D}$ and $\Psi^R$ admit (despite first appearance!) a good structure to establish an energy conservation law for $\Psi^{D}$ and $\Psi^R$. In fact, for $\Psi=\Psi^D$ and $\Psi=\Psi^R$, we prove a conservation law at the mode level for the following energy\footnote{Taking the sum on angular modes, one recovers the energy (\ref{Renergy}) introduced earlier.} 
\begin{align}\label{eq:introenergy}
  \begin{aligned}
  & \quad\quad\quad\quad \half\int_{\Si_{t^\ast}}\le(\frac{\De_-\De_+}{\De_0}|\pr_{t^\ast}\Psi_{m\ellmode}|^2 + |\pr_{r^\ast}\Psi_{m\ellmode}|^2 + V|\Psi_{m\ellmode}|^2\ri)\,\d r^\ast \\
    + & \lim_{r\to+\infty} \frac{6M}{(\ellmode-1)\ellmode(\ellmode+1)(\ellmode+2)}\le(|\pr_{t^\ast}\Psi_{m\ellmode}(t^\ast,r)|^2 + \half k^2\ellmode(\ellmode+1)|\Psi_{m\ellmode}(t^\ast,r)|^2\ri),
  \end{aligned}
\end{align}
where we recall that $\Si_{t^\ast}$ is the constant $t^\ast$ slice, and where $V$ is the Regge-Wheeler potential, given by
\begin{align*}
  V(r) & := w(r)\le(\ellmode(\ellmode+1)-\frac{6M}{r}\ri).
\end{align*}
Remarkably, the conserved energy~\eqref{eq:introenergy} contains, besides the familiar $H^1$-terms on the constant time slice $\Sigma_{t^\ast}$, an additional term on the sphere at infinity. In the Dirichlet case $\Psi^D$, this term vanishes and the conserved energy is the familiar spatial integral quantity.

It is \emph{a priori} not clear whether the conserved quantity~\eqref{eq:introenergy} is coercive. Indeed, the potential $V$ is not positive between its two roots $r_+$ and $r_c:=\frac{6M}{\ellmode(\ellmode+1)}$ (and positive otherwise). While in the asymptotically flat case we have $r_+=2M$, and therefore $r_+>r_c$ for all $\ellmode\geq 2$, in the AdS case, $r_+$ can be made arbitrarily small for fixed mass $M$ (in particular, smaller than $r_c$), provided that the cosmological constant $k$ is sufficiently large. See Remark~\ref{rem:Neumann} below. To resolve this problem and to obtain the coercivity of~\eqref{eq:introenergy}, we show that the quantity
\begin{align*}
  \half\int_{\Si_{t^\ast}}\le(|\pr_{r^\ast}\Phi|^2 + V|\Phi|^2\ri)\,\d r^\ast + \frac{3Mk^2}{(\ellmode-1)(\ellmode+2)}\lim_{r\to+\infty}|\Phi|^2
\end{align*}
is coercive. This uses a precise Hardy inequality, which crucially exploits the boundary term at infinity. Once established, the conservation law can be combined with a standard redshift estimate to establish uniform boundedness of $\Psi^{D}$ and $\Psi^R$. See already Theorem \ref{thm:mainRW1a}. 

\subsubsection{Carleman estimates and logarithmic decay for $\Psi^D$ and $\Psi^R$} \label{sec:logdecayintro}
The next step consists in proving a Carleman-type estimate for each fixed angular frequency mode $m\ellmode$ of the Regge-Wheeler quantities $\Psi^D$ and $\Psi^R$. The estimate is based on the use of exponential multipliers and bounds the energy integrated in time by the initial energy multiplied by a constant which \emph{grows exponentially in the angular momentum number $\ellmode$}:
\begin{align}\label{est:introcarleman}
  \begin{aligned}
    \int_{t^\ast_0}^{t^\ast}\mathrm{E}_{m\ellmode}[\Psi](t^{\ast,'})\,\d t^{\ast,'} \leq \exp\le(C\ellmode^{\mathfrak{p}}\ri)\mathrm{E}_{m\ellmode}[\Psi](t^\ast_0),
  \end{aligned}
\end{align}
where $C=C(M,k)$ is a constant, $\mathfrak{p}>0$ and where we have denoted by $\mathrm{E}_{m\ellmode}[\Psi]$ the relevant energy associated with $\Psi_{m\ellmode}$. See already Proposition \ref{prop:integralestimatesv} and Theorem \ref{thm:mainRW1b} below. 

The proof of~\eqref{est:introcarleman} is similar in spirit to \cite{Hol.Smu13} with an important difference: We are able to provide a \underline{purely physical space} proof of~\eqref{est:introcarleman} with $\mathfrak{p}=2$ (\emph{i.e.} using only angular decomposition but no Fourier transform in time). This relies on a careful analysis of the interaction of the Carleman multipliers with the Regge-Wheeler potential. Moreover, contrary to the proof in~\cite{Hol.Smu13} which relied on the Dirichlet boundary condition, our proof does not rely on the specific boundary conditions satisfied by the Regge-Wheeler quantity.

Combining~\eqref{est:introcarleman} with the energy boundedness, one infers the exponential decay for each angular frequency, with decay rate exponentially decreasing in $\ellmode$:
\begin{align}\label{est:expodecayintronotsharp}
  \mathrm{E}_{m\ellmode}[\Psi](t^{\ast}) \les \exp\le(e^{-C\ellmode^{\mathfrak{p}}}(t^\ast-t^\ast_0)\ri)\mathrm{E}_{m\ellmode}[\Psi](t^\ast_0).
\end{align}
Summing~\eqref{est:expodecayintronotsharp} in $m\ellmode$, interpolating it with the energy boundedness for $\LL^{1/2}\Psi$, we obtain
\begin{align*}
  \mathrm{E}[\Psi](t^\ast) \les \frac{\mathrm{E}[\LL^{1/2}\Psi](t^\ast_0)}{\le(\log(t^\ast-t^\ast_0)\ri)^{\frac{2}{\mathfrak{p}}}}.
\end{align*}
The exponent $\mathfrak{p}=2$ does not provide the optimal decay rate $\le(\log(t^\ast-t^\ast_0)\ri)^{-2}$ stated in Theorem~\ref{thm:main1}. Thus, we have to go back to~\eqref{est:introcarleman}, this time Fourier transforming in time as in~\cite{Hol.Smu13}, splitting $\Psi$ between its low and high time-frequencies. The splitting allows to apply Plancherel estimates, controlling $\pr_{t^\ast}\Psi$ by $\Psi$ (or $\Psi$ by $\pr_{t^\ast}\Psi$) in the estimates, which cannot be done with the physical space vectorfield method. 

There are however some benefits in first proving the non-optimal rate in physical space: Recall that in \cite{Hol.Smu13} one needed, even in the Schwarzschild case, to cut-off the solution in time to justify the use of the Fourier transform, which added a rather technical layer to the proof. Here we can avoid the future cut-off since we already established a slightly weaker exponential decay by physical space methods. More importantly perhaps, we believe that establishing a weaker exponential decay by physical space methods can be useful for applications to non-linear problems where taking the Fourier transform is often a source of technical complications. 

\subsubsection{Estimating $\alt^{[+2]}$ and $\alt^{[-2]}$}\label{sec:SolvingChandraintro}
For the final step, we note that by inverting~\eqref{eq:defPsiDR}, one can infer from the uniform boundedness and Carleman integral estimates for $\Psi^D$ and $\Psi^R$, appropriate estimates for $\Psi^{[\pm2]}$, see Section \ref{sec:proofthmTeukbdddecay}. Given estimates on $\Psi^{[\pm2]}$, uniform boundedness and integral estimates for $\alt^{[\pm2]}$ can  be derived by interpreting the transformations \eqref{eq:Chandra} as transport equations in the null directions (see Figure~\ref{fig:integrationRWtoTeuk}): Specifically, one obtains first bounds for $\psi^{[-2]}$ by integrating from data in the outgoing direction (using the bounds on $\Psi^{[-2]}$) and then -- using the corresponding boundary condition relating $\psi^{[-2]}$ to $\psi^{[+2]}$ -- bounds for $\psi^{[+2]}$ by integrating from the boundary in the ingoing direction (using now the bounds on $\Psi^{[+2]}$).  Finally, one obtains bounds for $\widetilde{\al}^{[-2]}$ by integrating from data in the outgoing direction (using the bounds just established on $\psi^{[-2]}$) and then -- now using the corresponding boundary condition relating $\widetilde{\al}^{[-2]}$  to $\widetilde{\al}^{[+2]}$  -- bounds for $\widetilde{\al}^{[+2]}$  by integrating from the boundary in the ingoing direction (using the bounds just established on $\psi^{[+2]}$).

\begin{figure}[h!]
  \centering
  \includegraphics[height=6cm]{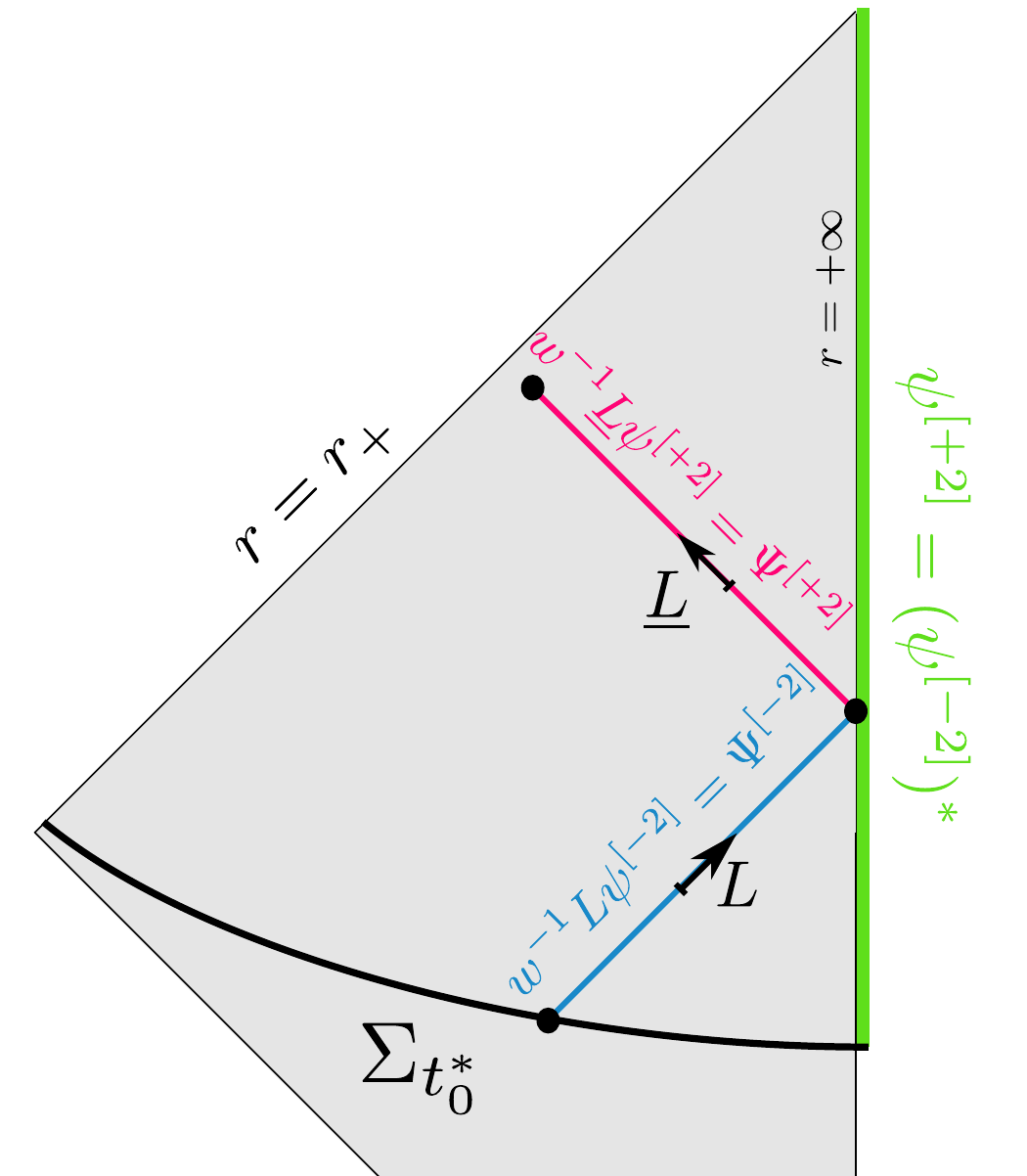}
  \caption{Integrating the Chandrasekhar transformations}
  \label{fig:integrationRWtoTeuk}
\end{figure}

While this procedure, carried out in Section \ref{sec:prelimtrans}, encounters an obvious loss of derivatives, this loss can finally be recovered through elliptic and inhomogeneous Morawetz estimates similar to the arguments carried out in \cite{Daf.Hol.Rod19, Daf.Hol.Rod.Tay21} but now with a careful treatment of the coupling of the boundary terms in these estimates. The details are collected in Sections~\ref{sec:energyestpsipm2}--\ref{sec:redshiftpsim2}. In summary, these estimates yield exponential decay for each angular frequency~\eqref{est:main1expo} as concluded in Section \ref{sec:fipo}, from which inverse logarithmic decay for the whole solution~\eqref{est:main1log} follows by an interpolation argument presented in Section \ref{sec:proofexpodecay}.

\subsection{Relation with the other papers of the series}\label{sec:outlook}
We comment on the links between the present paper and the two other papers~\cite{Gra.Hol,Gra.Hol24b} of the series.

In~\cite{Gra.Hol24b}, we show that the decay estimates of Theorem~\ref{thm:main1} are sharp by constructing a sequence of quasimode solutions to the Teukolsky problem of Definition~\ref{def:sysTeuk}. Quasimode solutions are the exact solutions to the Teukolsky problem associated with approximate solutions which are real modes, \emph{i.e.} oscillating (without decay) in time. The construction is inspired by \cite{Hol.Smu14} for the covariant wave equation but has to overcome the problem of the first order term in the Teukolsky equation. 


In our companion paper \cite{Gra.Hol}, we consider the full system of gravitational perturbations, i.e.~the linearised Einstein equations, on Schwarzschild-AdS in a double null gauge. In this gauge, the Teukolsky quantities are part of the dynamical variables and satisfy the Teukolsky equation, which allows us to apply directly the boundedness and decay results for solutions to the Teukolsky problem of Theorem~\ref{thm:main1}. Integrating in a second step the so-called \emph{null structure equations} in an appropriate order, we obtain the boundedness and decay of all metric and Ricci coefficients of solutions to linearised gravity in double null gauge, up to residual pure gauge solutions and linearised Kerr-AdS solutions. This overall strategy follows closely the one from~\cite{Daf.Hol.Rod19} although the presence of the conformal boundary requires novel arguments to understand the asymptotics. 

\subsection{Acknowledgements} 
G.H.~acknowledges support by the Alexander von Humboldt Foundation in the framework of the Alexander von Humboldt Professorship endowed by the Federal Ministry of Education and Research as well as ERC Consolidator Grant 772249. Both authors acknowledge funding through Germany’s Excellence Strategy EXC 2044 390685587, Mathematics M\"unster: Dynamics-Geometry-Structure.

\section{Energy estimates for the Regge-Wheeler equations}\label{sec:RWbound}
Let us define the following energy quantities expressed in regular $(t^\ast,r^\ast,\varth,\varphi)$ coordinates. They should be thought of arising from the timelike isometry $T=\partial_{t^\ast}$.
\begin{align}\label{eq:defenergiesPsiDR}
  \begin{aligned}
    \widetilde{\mathrm{E}}_{m\ellmode}\le[\Psi\ri](t^\ast) & := \half \int_{-\infty}^{\frac{\pi}{2}} \le(\frac{\De_-\De_+}{\De_0^2}|\pr_{t^\ast}\Psi_{m\ellmode}|^2 + |\pr_{r^\ast}\Psi_{m\ellmode}|^2 + \frac{\De}{r^4}\ellmode(\ellmode+1)|\Psi_{m\ellmode}|^2\ri) \,\d r^\ast,\\
    \widetilde{\mathrm{E}}^\HH_{m\ellmode}[\Psi_{m\ellmode}](t^\ast_2;t^\ast_1) & := \lim_{r\to r_+}\int_{t^\ast_1}^{t^\ast_2} |\pr_{t^\ast}\Psi_{m\ellmode}|^2 \,\d t^{\ast}, \\
    \Einfty_{m\ellmode}[\Psi](t^\ast) & := \lim_{r\to+\infty}\frac{6M}{\ellmode(\ellmode+1)\le(\ellmode(\ellmode+1)-2\ri)} \le[|\pr_{t^\ast}\Psi_{m\ellmode}|^2 + k^2\frac{\ellmode(\ellmode+1)}{2}|\Psi_{m\ellmode}|^2\ri].
  \end{aligned}
\end{align}
For any energy $E_{m \ell}$ (such as the ones defined in (\ref{eq:defenergiesPsiDR})) we agree on the convention that $E$ without the angular subscripts denotes  the summed energies that is we define
\begin{align*}
  \mathrm{E}[\Phi] & := \sum_{\ellmode\geq 2}\sum_{|m|\leq\ellmode} \mathrm{E}_{m\ellmode}[\Phi] .
\end{align*}
Note this is consistent with (\ref{individualmodeE}) and the energies defined earlier in (\ref{Tenergy}) and (\ref{Renergy}). 


The following proposition is the main result of this section. Its proof is carried out in Section~\ref{sec:proofpropenergybounded}.
\begin{proposition}[Energy boundedness for $\Psi^D,\Psi^R$]\label{prop:energybounded}
  For all $t_{1}^\ast,t^\ast_2\in\mathbb{R}$ with $t^\ast_1\leq t^\ast_2$, and all solutions $\Psi^D,\Psi^R$ to the  decoupled Regge-Wheeler problem~\eqref{sys:RWdeco} on $
 \mathcal{M} \cap \{t^\ast \geq t^\ast_1\}$, we have
    \begin{align}\label{est:energyPsiDRmain}
      \begin{aligned}
        \widetilde{\mathrm{E}}\le[\Psi\ri](t^\ast_2) + \widetilde{\mathrm{E}}^\HH[\Psi](t^\ast_2;t^\ast_1) + \Einfty[\Psi](t^\ast_2) & \les_{M,k} \widetilde{\mathrm{E}}\le[\Psi\ri](t^\ast_1) + \Einfty[\Psi](t^\ast_1),
      \end{aligned}
    \end{align}
    for $\Psi=\Psi^D$ and for $\Psi=\Psi^R$.
\end{proposition}

\begin{remark}
Note that if $\alt^{[\pm2]}$ are solutions of the Teukolsky problem arising from Theorem \ref{thm:IBVP}, then we obtain solutions $\Psi^D,\Psi^R$ to the  decoupled Regge-Wheeler problem~\eqref{sys:RWdeco} on $\mathcal{M} \cap \{t^\ast \geq t^\ast_0\}$, see Section \ref{sec:chandraintro}.
\end{remark}

\begin{remark}
  Each term in the energies~\eqref{eq:defenergiesPsiDR} is non-negative and Proposition~\ref{prop:energybounded} hence provides a full boundedness result for both $\Psi^D$ and $\Psi^R$. Note however that the regular derivative transversal to the event horizon $\mathcal{H}^+$ degenerates for $\widetilde{\mathrm{E}}_{m\ellmode}\le[\Psi\ri]$ as is expected from the fact that the energy arises from the vectorfield $T$.  
\end{remark}
\begin{remark}
  In the Dirichlet case $\Psi=\Psi^D$, we simply have $\Einfty[\Psi](t^\ast)=0$ for all $t^\ast\in\RRR$.
\end{remark}
\begin{remark}
Obviously, (\ref{est:energyPsiDRmain}) also holds for each individual mode, i.e.~putting a subscript $m\ell$ on all energies (recall (\ref{individualmodeE})). In fact, (\ref{est:energyPsiDRmain})  will be proven first for the individual modes and then summed. 
\end{remark}


\subsection{Energy identities}
Let us define the energy
\begin{align*}
  \mathring{\mathrm{E}}_{m\ellmode}[\Psi](t^\ast)  & := \half \int_{-\infty}^{\frac{\pi}{2}} \le(\frac{\De_-\De_+}{\De_0^2}|\pr_{t^\ast}\Psi_{m\ellmode}|^2 + |\pr_{r^\ast}\Psi_{m\ellmode}|^2 + \frac{\De}{r^2}\le(\frac{\ellmode(\ellmode+1)}{r^2}-\frac{6M}{r^3}\ri)|\Psi_{m\ellmode}|^2\ri) \,\d r^\ast.
\end{align*}

We have the following lemma.
\begin{lemma}[Energy identities]\label{lem:energyid}
  For all $t_{1}^\ast,t^\ast_2\in\mathbb{R}$ with $t^\ast_1\leq t^\ast_2$, and all smooth solutions $\Psi^D,\Psi^R$ to the  decoupled Regge-Wheeler problem~\eqref{sys:RWdeco} on $\mathcal{M} \cap \{ t \geq t_1^\ast\}$, we have
    \begin{align}\label{eq:energyPsiDRmain}
      \begin{aligned}
        \mathring{\mathrm{E}}\le[\Psi\ri](t^\ast_2) + \widetilde{\mathrm{E}}^\HH[\Psi](t^\ast_2;t^\ast_1) + \Einfty[\Psi](t^\ast_2) & = \mathring{\mathrm{E}}\le[\Psi\ri](t^\ast_1) + \Einfty[\Psi](t^\ast_1),
      \end{aligned}
    \end{align}
    for $\Psi=\Psi^D$ and for $\Psi=\Psi^R$.
\end{lemma}
\begin{proof}
  Using~\eqref{eq:defLLb} the Regge-Wheeler operator of~\eqref{eq:RW} rewrites in the $t^\ast,r^\ast$ coordinates as\footnote{Here and in the following we write $\RW^{[\pm2]}_{m\ellmode}$ to denote the projections of the Teukolsky operators $\RW^{[\pm2]}$ onto the Hilbert basis of the angular operator $\LL^{[\pm2]}$. See Remark \ref{rem:eigenbasis}.}
  \begin{align}\label{eq:pfenergid1}
    \begin{aligned}
    -\RW_{m\ellmode}\Psi_{m\ellmode} & = \le(\frac{\De_+}{\De_0}\pr_{t^\ast} + \pr_{r^\ast}\ri)\le(\frac{\De_-}{\De_0}\pr_{t^\ast} - \pr_{r^\ast}\ri) \Psi_{m\ellmode} + \frac{\De}{r^4}\le(\LL-\frac{6M}{r}\ri)\Psi_{m\ellmode} \\
    & = \frac{\De_-\De_+}{\De_0^2}\pr_{t^\ast}^2\Psi_{m\ellmode} -\pr_{r^\ast}^2\Psi_{m\ellmode} + \pr_{r^\ast}\le(\frac{\De_-}{\De_0}\pr_{t^\ast}\Psi_{m\ellmode}\ri) - \frac{\De_+}{\De_0}\pr_{r^\ast}\pr_{t^\ast}\Psi_{m\ellmode} \\
    & \quad + \frac{\De}{r^4}\le(\ellmode(\ellmode+1)-\frac{6M}{r}\ri)\Psi_{m\ellmode}.
    \end{aligned}
  \end{align}
  Multiplying~\eqref{eq:pfenergid1} by $\pr_{t^\ast}\Psi^\ast_{m\ellmode}$ and taking the real part, we get
  \begin{align}\label{eq:pfenergid2}
    \begin{aligned}
      -\Re\le(\RW_{m\ellmode}\Psi_{m\ellmode}\pr_{t^\ast}\Psi_{m\ellmode}^\ast\ri) & = \half \pr_{t^\ast}\le(\frac{\De_-\De_+}{\De_0^2} |\pr_{t^\ast}\Psi_{m\ellmode}|^2 \ri) - \pr_{r^\ast}\Re\le(\pr_{r^\ast}\Psi_{m\ellmode}\pr_{t^\ast}\Psi_{m\ellmode}^\ast\ri) + \half \pr_{t^\ast}\le(|\pr_{r^\ast}\Psi_{m\ellmode}|^2\ri) \\
      & \quad + \half\frac{\De_--\De_+}{\De_0}\pr_{r^\ast}\le(|\pr_{t^\ast}\Psi_{m\ellmode}|^2\ri) + \pr_{r^\ast}\le(\frac{\De_-}{\De_0}\ri) |\pr_{t^\ast}\Psi_{m\ellmode}|^2 \\
      & \quad + \half\frac{\De}{r^4}\le(\ellmode(\ellmode+1)-\frac{6M}{r}\ri)\pr_{t^\ast}\le(|\Psi_{m\ellmode}|^2\ri).
    \end{aligned}
  \end{align}
  Integrating~\eqref{eq:pfenergid2} on the domain $(t^\ast_1,t^\ast_2)_{t^\ast}\times(-\infty,\frac{\pi}{2})_{r^\ast}\times\SSS^2_{\varth,\varphi}$, using that
  \begin{align*}
    -\half\pr_{r^\ast}\le(\frac{\De_--\De_+}{\De_0}\ri) + \pr_{r^\ast}\le(\frac{\De_-}{\De_0}\ri) = 0,
  \end{align*}
  that
  \begin{align*}
    \half\frac{\De_--\De_+}{\De_0} & \xrightarrow{r\to+\infty} 0, & \half\frac{\De_--\De_+}{\De_0} & \xrightarrow{r\to r_+} -1,  
  \end{align*}
  and that by~\eqref{eq:defLLb} and the regularity at the horizon conditions~\eqref{eq:defreghorgeneral} for $\Psi$, we have
  \begin{align*}
    \pr_{r^\ast}\Psi & = \Lb\Psi + \half\frac{\De_{-}}{\De_0}\le(L+\Lb\ri)\Psi \xrightarrow{r\to r_+} 0,
  \end{align*}
  we obtain
  \begin{align}\label{eq:pfenergid3}
    \begin{aligned}
      & \mathring{\mathrm{E}}_{m\ellmode}\le[\Psi\ri](t^\ast_2) + \widetilde{\mathrm{E}}_{m\ellmode}^\HH[\Psi](t^\ast_2;t^\ast_1) - \lim_{r\to+\infty} \int_{t^\ast_1}^{t^\ast_2}\Re\le(\pr_{r^\ast}\Psi_{m\ellmode}\pr_{t^\ast}\Psi^\ast_{m\ellmode}\ri)\,\d t^\ast \\
      = & \; \mathring{\mathrm{E}}_{m\ellmode}\le[\Psi\ri](t^\ast_1) - \int_{t^\ast_1}^{t^\ast_2}\int_{-\infty}^{\frac{\pi}{2}}\Re\le(\RW_{m\ellmode}\Psi_{m\ellmode}\pr_{t^\ast}\Psi_{m\ellmode}^\ast\ri)\,\d t^\ast\d r^\ast.
    \end{aligned}
  \end{align}
  In the Dirichlet case $\Psi=\Psi^D$, we directly have\footnote{The boundary term would also vanish with Neumann boundary conditions (but would not be equal to $\Einfty[\Psi]$).}
  \begin{align*}
    -\lim_{r\to+\infty}\int_{t^\ast_1}^{t^\ast_2}\Re\le(\pr_{r^\ast}\Psi^D_{m\ellmode}\pr_{t^\ast}\big(\Psi^D_{m\ellmode}\big)^\ast\ri)\,\d t^\ast & = 0 = \Einfty_{m\ellmode}[\Psi^D](t^\ast_{2})-\Einfty_{m\ellmode}[\Psi^D](t^\ast_1)
  \end{align*}
  and, using that $\RW\Psi=0$,~\eqref{eq:energyPsiDRmain} follows. In the Robin case $\Psi=\Psi^R$, the boundary condition~\eqref{eq:RWBCdecoRobin} implies 
  \begin{align}\label{eq:pfenergid4}
    \begin{aligned}
      & -\lim_{r\to+\infty}\int_{t^\ast_1}^{t^\ast_2}\Re\le(\pr_{r^\ast}\Psi^R_{m\ellmode}\pr_{t^\ast}\big(\Psi^R_{m\ellmode}\big)^\ast\ri)\,\d t^\ast  \\
       = &\; \lim_{r\to+\infty}\int_{t^\ast_1}^{t^\ast_2}\Re\le(\frac{6M}{\ellmode(\ellmode+1)\le(\ellmode(\ellmode+1)-2\ri)}\le(2\pr_{t^\ast}^2\Psi^R_{m\ellmode}+k^2\ellmode(\ellmode+1)\Psi^R_{m\ellmode}\ri)\pr_{t^\ast}\big(\Psi^R_{m\ellmode}\big)^\ast\ri)\,\d t^\ast \\
      = & \; \Einfty_{m\ellmode}[\Psi^R](t^\ast_2)-\Einfty_{m\ellmode}[\Psi^R](t^\ast_1).
    \end{aligned}
  \end{align}
  Plugging~\eqref{eq:pfenergid4} in~\eqref{eq:pfenergid3}, using that $\RW\Psi=0$, this finishes the proof of~\eqref{eq:energyPsiDRmain} when $\Psi=\Psi^R$.
\end{proof}

\subsection{Hardy estimates}
We have the following lemma.
\begin{lemma}[Hardy estimate]\label{lem:Hardy}
  Let $\ellmode\geq 2$. For all smooth radial functions $\Phi:(r_+,+\infty)\to\CCC$ satisfying the regularity condition at the horizon~\eqref{eq:defreghorgeneral} and having a limit at infinity, we have
  \begin{align}
    \label{eq:Hardymain}
    \begin{aligned}
      \int_{-\infty}^{\frac{\pi}{2}} \frac{\De}{r^4} \frac{\ellmode(\ellmode+1)}{\le(1+\frac{6M}{(\ellmode(\ellmode+1)-2)r}\ri)}|\Phi|^2\,\d r^\ast & \leq \int_{-\infty}^{\frac{\pi}{2}} \le(|\pr_{r^\ast}\Phi|^2 + \frac{\De}{r^4}\le(\ellmode(\ellmode+1) - \frac{6M}{r}\ri)|\Phi|^2\ri)\,\d r^\ast \\
      & \quad + \lim_{r\to+\infty}\le(k^2\frac{6M}{\ellmode(\ellmode+1)-2} |\Phi|^2\ri).   
    \end{aligned}
  \end{align}
\end{lemma}
\begin{proof}
  Let $f:(-\infty,\frac{\pi}{2})_{r^\ast} \to \RRR$ be a regular function. We have
  \begin{align}\label{eq:mainHardyPsiRpf}
    \begin{aligned}
      0 \leq \int_{-\infty}^{\frac{\pi}{2}} f^2 \le|\pr_{r^\ast}\le(\frac{\Phi}{f}\ri)\ri|^2 \,\d r^\ast & =  \int_{-\infty}^{\frac{\pi}{2}} \le(|\pr_{r^\ast}\Phi|^2 - \frac{f'}{f} \pr_{r^\ast}\le(|\Phi|^2\ri) + \le(\frac{f'}{f}\ri)^2|\Phi|^2\ri)\,\d r^\ast \\
      & =  \int_{-\infty}^{\frac{\pi}{2}} \le(|\pr_{r^\ast}\Phi|^2 + \frac{f''}{f} |\Phi|^2\ri)\,\d r^\ast - \le[\frac{f'}{f}|\Phi|^2\ri]^{\frac{\pi}{2}}_{-\infty},
    \end{aligned}
  \end{align}
  provided that $\frac{f'}{f}$ is sufficiently regular at $r^\ast=-\infty$ and $r^\ast=\frac{\pi}{2}$. Let us take $f(r^\ast) := 1 + \frac{c}{r(r^\ast)}$ with $c:= \frac{6M}{\le(\ellmode(\ellmode+1)-2\ri)}$, such that
  \begin{align}\label{eq:HardyPsiRpff}
    f' & = -cw, & f'' & = w\frac{c}{r}\le(2 - \frac{6M}{r}\ri),  
  \end{align}
  and in particular
  \begin{align}\label{eq:limitsHardyPsiR}
    \frac{f'}{f} & \xrightarrow{r\to r_+} 0, & \frac{f'}{f} & \xrightarrow{r\to+\infty} -ck^2. 
  \end{align}
  Plugging~\eqref{eq:HardyPsiRpff} and~\eqref{eq:limitsHardyPsiR} in~\eqref{eq:mainHardyPsiRpf}, using the regularity condition~\eqref{eq:defreghorgeneral} for $\Phi$, we obtain that
  \begin{align}\label{est:mainHardyPsiRpf2}
    0 & \leq \int_{-\infty}^{\frac{\pi}{2}} \le(|\pr_{r^\ast}\Phi|^2 + w\frac{c}{r\le(1+\frac{c}{r}\ri)}\le(2 - \frac{6M}{r}\ri)|\Phi|^2\ri)\,\d r^\ast + \lim_{r\to+\infty} \le(ck^2 |\Phi|^2\ri).
  \end{align}
  Using that $6M = c(\ellmode(\ellmode+1)-2)$, we check that
  \begin{align*}
    \frac{c}{r\le(1+\frac{c}{r}\ri)}\le(2 - \frac{6M}{r}\ri) & = \le(\ellmode(\ellmode+1)-\frac{6M}{r}\ri)- \frac{\ellmode(\ellmode+1)}{1+\frac{c}{r}}.
  \end{align*}
  Plugging the above in~\eqref{est:mainHardyPsiRpf2}, we get~\eqref{eq:Hardymain} and this finishes the proof of the lemma.
\end{proof}

\subsection{Proof of boundedness of the degenerate energy}\label{sec:proofpropenergybounded}
We now prove Proposition~\ref{prop:energybounded}.
Let $\Psi=\Psi^D,\Psi^R$ be a solution to the  Regge-Wheeler problem~\eqref{sys:RWdeco}. The Hardy inequality~\eqref{eq:Hardymain} implies that on all slices $\Si_{t^\ast}$
\begin{align*}
  \int_{-\infty}^{\frac{\pi}{2}} \frac{\De}{r^4} \ellmode(\ellmode+1) |\Psi_{m\ellmode}|^2\,\d r^\ast \leq \le(1+\frac{3M}{2r_+}\ri) \int_{-\infty}^{\frac{\pi}{2}} \frac{\De}{r^4} \frac{\ellmode(\ellmode+1)}{\le(1+\frac{6M}{(\ellmode(\ellmode+1)-2)r}\ri)}|\Psi_{m\ellmode}|^2\,\d r^\ast \les_{M,k} \mathring{\mathrm{E}}_{m\ellmode}[\Psi](t^\ast) + \Einfty_{m\ellmode}[\Psi](t^\ast). 
\end{align*}
From the above and the definition of the energy quantities, we infer $\widetilde{\mathrm{E}}[\Psi](t^\ast) \les_{M,k} \mathring{\mathrm{E}}[\Psi](t^\ast) + \Einfty[\Psi](t^\ast)$. Plugging this, together with the obvious $\mathring{\mathrm{E}}[\Psi](t^\ast) \leq \widetilde{\mathrm{E}}[\Psi](t^\ast)$, in the energy identity~\eqref{eq:energyPsiDRmain} yields the desired energy estimate~\eqref{est:energyPsiDRmain} and this finishes the proof of Proposition~\ref{prop:energybounded}.   

\begin{remark}\label{rem:Neumann}
  If $\Psi$ satisfies Neumann boundary conditions $r^2\pr_r\Psi\xrightarrow{r\to+\infty}0$ instead of the Dirichlet or ``Robin'' conditions~\eqref{eq:RWBCdeco}, we still have a conservation law
  \begin{align*}
    \mathring{\mathrm{E}}\le[\Psi\ri](t^\ast_2) + \widetilde{\mathrm{E}}^\HH[\Psi](t^\ast_2;t^\ast_1) & = \mathring{\mathrm{E}}\le[\Psi\ri](t^\ast_1),
  \end{align*}
  for all $t^\ast_2\geq t^\ast_1$. However, as suggested by the Hardy estimate~\eqref{eq:Hardymain}, the energy $\mathring{\mathrm{E}}\le[\Psi\ri]$ alone -- \emph{i.e.} without $\Einfty[\Psi]$ -- \underline{is not coercive}. In fact, for $\Psi_{m\ellmode}=1$, we have
  \begin{align*}
    \mathring{\mathrm{E}}_{m\ellmode}\le[\Psi\ri] = \half\int_{-\infty}^{\frac{\pi}{2}}\frac{\De}{r^4}\le(\ellmode(\ellmode+1)-\frac{6M}{r}\ri)\,\d r^\ast = \frac{1}{2r_+}\le(\ellmode(\ellmode+1) - \frac{3M}{r_+}\ri),
  \end{align*}
  which is negative if the radius $r_+$ is sufficiently small (depending on $\ellmode$ and $M$). Recall that $r_+$ can be made arbitrarily small, provided that the cosmological constant $k$ is chosen sufficiently large. See also Section~\ref{sec:boundednessintro}.
\end{remark}

\section{Integral estimates for the Regge-Wheeler equations}\label{sec:RWhigh}

\subsection{Redshift and non-degenerate boundedness estimates}
Let $\ellmode\geq2$, $|m|\leq\ellmode$. We recall (\ref{eq:defenergiesPsiDR}) and define the following energy norms:
\begin{align}\label{eq:defenergynormsRWhigh}
  \begin{aligned}
    \overline{\mathrm{E}}_{m\ellmode}[\Psi](t^\ast) & := \widetilde{\mathrm{E}}_{m\ellmode}[\Psi](t^\ast) + \int_{-\infty}^{\frac{\pi}{2}}w^{-1}|\pr_{r^\ast}\Psi_{m\ellmode}|^2\,\d r^\ast, \\
    \overline{\mathrm{E}}^\HH_{m\ellmode}[\Psi](t^\ast_2;t^\ast_1) & := \widetilde{\mathrm{E}}^\HH_{m\ellmode}[\Psi](t^\ast_2;t^\ast_1) + \lim_{r\to r_+}\int_{t^\ast_1}^{t^\ast_2} \ellmode^2|\Psi_{m\ellmode}|^2 \,\d t^{\ast}, \\
    \overline{\mathrm{E}}^\II_{m\ellmode}[\Psi](t^\ast_2;t^\ast_1) & := \lim_{r\to+\infty}\int_{t^\ast_1}^{t^\ast_2}\le(|\pr_{t^\ast}\Psi_{m\ellmode}|^2 + |\pr_{r^\ast}\Psi_{m\ellmode}|^2 + k^2\ellmode(\ellmode+1)|\Psi_{m\ellmode}|^2\ri)\,\d t^\ast.
  \end{aligned}
\end{align}
Note that these energies have stronger weights at the horizon compared with  (\ref{eq:defenergiesPsiDR}) and that $\partial_{r^\ast}=\partial_{r^\ast}|_{(t^\ast,r^\ast,\theta,\phi)}$ in the above, so in particular $w^{-1} \partial_{r^\ast}$ extends regularly to $\mathcal{H}^+$.  With the definitions of Section~\ref{sec:normsintro} and~\eqref{eq:defenergynormsRWhigh}, we have the relation
\begin{align} \label{merel}
  \mathrm{E}^{\mathfrak{R}}_{m\ellmode}[\Psi](t^\ast) \sim_{M,k} \overline{\mathrm{E}}_{m\ellmode}[\Psi](t^\ast) + \Einfty_{m\ellmode}[\Psi](t^\ast).
\end{align}

Below we will prove the following redshift estimate.
\begin{proposition}[Redshift integral estimate]\label{prop:redshift}
  Let $\Psi \in \{\Psi^D,\Psi^R\}$ be a solution to the decoupled Regge-Wheeler problem~\eqref{sys:RWdeco} on $\mathcal{M} \cap \{t^\ast \geq t^\ast_1\}$. Then, for all $t^\ast_2\geq t^\ast_1$, we have
  \begin{align}\label{est:redshiftmain}
    \begin{aligned}
      & \mathrm{E}^{\mathfrak{R}}[\Psi](t^\ast_2) + \int_{t^\ast_1}^{t^\ast_2} \mathrm{E}^{\mathfrak{R}}[\Psi]\,\d t^\ast + \overline{\mathrm{E}}^\HH[\Psi](t^\ast_2;t^\ast_1) \les_{M,k}  \mathrm{E}^{\mathfrak{R}}[\Psi](t^\ast_1) + \int_{t^\ast_1}^{t^\ast_2}\le(\widetilde{\mathrm{E}}[\Psi](t^\ast) + \Einfty[\Psi](t^\ast)\ri)\,\d t^\ast.
    \end{aligned}
  \end{align}
\end{proposition}
\begin{remark}
  The redshift estimate of Proposition~\ref{prop:redshift} is by now classical, see~\cite{Daf.Rod08} for a general formulation for non-degenerate Killing horizons and \cite{Hol09} where the covariant wave equation on exact Kerr-AdS case is treated. We will prove (\ref{est:redshiftmain}) using the same framework as for the Carleman estimates of Proposition~\ref{prop:integralestimatesv}. See Lemma~\ref{lem:Moridentity} for the basic integral identity and Section~\ref{sec:Carlemanredshift} for the proof.
\end{remark}

As is well known, Proposition~\ref{est:redshiftmain} and Proposition~\ref{prop:energybounded} can be combined to obtain uniform boundedness by a pigeonhole argument (see~\cite{Daf.Rod08} for the original appearance of the argument, which we will not repeat here). This proves the following theorem.

\begin{theorem}[Boundedness for $\Psi^{D},\Psi^R$]\label{thm:mainRW1a}
  Let $\ellmode\geq 2$, $|m|\leq\ellmode$. Let $\Psi \in \{\Psi^D,\Psi^R\}$ be a solution to the decoupled Regge-Wheeler problem~\eqref{sys:RWdeco} on $\mathcal{M} \cap \{t^\ast \geq t^\ast_1\}$. We have for all $t^\ast_2\geq t^\ast_1$ the boundedness estimate 
    \begin{align}\label{est:boundednessPsiDR}
      \begin{aligned}
        \mathrm{E}^{\mathfrak{R}}[\Psi](t^\ast_2) + \widetilde{\mathrm{E}}^\HH[\Psi](t^\ast_2;t^\ast_1) & \les_{M,k} \mathrm{E}^{\mathfrak{R}}[\Psi](t^\ast_1). 
      \end{aligned}
    \end{align}
 \end{theorem}

\subsection{Integrated decay estimates}

The following proposition is the key, angular dependent, integrated decay estimate of the paper. 
\begin{proposition}[Carleman integral estimates]\label{prop:integralestimatesv}
  Let $\ellmode\geq 2$, $|m|\leq\ellmode$. Let $\Psi \in \{\Psi^D,\Psi^R\}$ be a solution to the  decoupled Regge-Wheeler problem~\eqref{sys:RWdeco} on $\mathcal{M} \cap \{t^\ast \geq t^\ast_1\}$. Then, for all $t^\ast_2\geq t^\ast_1$, we have
  \begin{align}
    \label{est:Carlemanmain}
    \begin{aligned}
      \int_{t^\ast_1}^{t^\ast_2} \mathrm{E}^{\mathfrak{R}}_{m\ellmode}[\Psi](t^\ast)\,\d t^\ast + \overline{\mathrm{E}}^\II_{m\ellmode}[\Psi](t^\ast_2;t^\ast_1) & \leq e^{C\ellmode^\mathfrak{p}}\mathrm{E}^{\mathfrak{R}}_{m\ellmode}[\Psi](t^\ast_1),
    \end{aligned}
  \end{align}
  where $C=C(M,k)>0$ and where $\mathfrak{p}=1$. Moreover, we establish~\eqref{est:Carlemanmain} in the weaker case $\mathfrak{p}=2$, purely based on vectorfields estimates (\emph{i.e.} without time-frequency considerations).
\end{proposition}


As a trivial consequence, combining (\ref{est:Carlemanmain}) and (\ref{est:boundednessPsiDR})  we have the following theorem:

\begin{theorem}[Integrated decay for $\Psi^{D},\Psi^R$]\label{thm:mainRW1b}
  Let $\ellmode\geq 2$, $|m|\leq\ellmode$. Let $\Psi \in \{\Psi^D,\Psi^R\}$ be a solution to the decoupled Regge-Wheeler problem~\eqref{sys:RWdeco} on $\mathcal{M} \cap \{t^\ast \geq t^\ast_1\}$. We have
  \begin{align}\label{est:pfcorCarl1}
      \begin{aligned}
        \mathrm{E}^{\mathfrak{R}}_{m\ellmode}[\Psi](t^\ast_2)+\int_{t^\ast_1}^{t^\ast_2} \mathrm{E}^{\mathfrak{R}}_{m\ellmode}[\Psi](t^\ast)\,\d t^\ast  + \overline{\mathrm{E}}^\HH_{m\ellmode}[\Psi](t^\ast_2;t^\ast_1) + \overline{\mathrm{E}}^\II_{m\ellmode}[\Psi](t^\ast_2;t^\ast_1) \leq e^{C\ellmode^{\mathfrak{p}}}\mathrm{E}^{\mathfrak{R}}_{m\ellmode}[\Psi](t^\ast_1)
      \end{aligned}
    \end{align}
    for all $t^\ast_2\geq t^\ast_1$, with $\mathfrak{p}=1$, and where $C=C(M,k)>0$.
\end{theorem}

The remainder of this section is dedicated to the proofs of Propositions~\ref{prop:redshift} and~\ref{prop:integralestimatesv}. In Section~\ref{sec:threebasicestimates} we collect the basic integral identities and estimates that will be used in the proofs. In Section~\ref{sec:Carlemanredshift} we prove the redshift estimates of Proposition~\ref{prop:redshift}. In Section~\ref{sec:Carlemanestimatesfirst} we prove the Carleman estimates of Proposition~\ref{prop:integralestimatesv} in the $\mathfrak{p}=2$ case using only vectorfield estimates (\emph{i.e.} without any time-frequency consideration). To that end, we need detailed asymptotics towards the horizon of the potentials appearing in these estimates, which are postponed to Appendix~\ref{sec:Carlemanmultipliers}. In Section~\ref{sec:Fourier} we improve the Carleman estimates of Proposition~\ref{prop:integralestimatesv} from $\mathfrak{p}=2$ to $\mathfrak{p}=1$ using time-frequency cut-off and a separated treatment of the low and the high time-frequencies.


\begin{remark}
  The proof of the Carleman estimates in the strong $\mathfrak{p}=1$ case is an adaptation of Sections 7 (low-frequency estimates) and 8 (high-frequency estimates) of~\cite{Hol.Smu13}. We realised that the basic estimate used in the high-frequency case in~\cite{Hol.Smu13} can be used to obtain a global, unique Carleman estimate \underline{for all time frequencies}, provided that we consider a weaker $\mathfrak{p}=2$ power of $\ellmode$ in~\eqref{est:Carlemanmain}. Thus, we obtain that each angular frequency decays exponentially in time before having to take the Fourier transform in time. This allows us to take the Fourier transform without having to cut-off at large times. See the cut-off argument in Section~\ref{sec:Fourier}.
\end{remark}

\begin{remark}
  In~\cite{Hol.Smu13}, only Dirichlet boundary conditions are considered (which corresponds to the $\Psi=\Psi^D$ case in the present paper). In the proof of Proposition~\ref{prop:integralestimatesv} in the $\mathfrak{p}=2$ case, we do not rely on any specific boundary condition and we show that the Carleman estimates hold true in general, provided that $\Psi$ satisfies an energy boundedness statement. Indeed, it turns out that the (integrated) Dirichlet and Neumann values of $\Psi$ enter with a good sign and can be controlled simultaneously as the other spacetime integrals. See (the signs of) the boundary terms in estimates~\eqref{est:easyCarleman} and~\eqref{est:bdyineasyCarleman}. In the $\mathfrak{p}=1$ case, the specific (Dirichlet or Robin) boundary conditions come into play in the estimates of the low time-frequencies $\Psic^\flat$ in Section~\ref{sec:Psiflat}. We show that the Robin condition in the $\Psic^\flat$ case produces a good sign on the boundary which allows to close the estimate. See Lemma~\ref{lem:Plancherelellipticboundary}. 
\end{remark}

In the next sections we assume that $\Psi=\Psi_{m\ellmode}$ and that the Regge-Wheeler equation~\eqref{eq:RW} is projected onto the eigenspaces of $\LL$ since there is no possible confusion.

\subsection{Basic integral identities}\label{sec:threebasicestimates}
In this section, we will derive a few multiplier identities to prepare for the proof of Propositions~\ref{prop:redshift} and~\ref{prop:integralestimatesv}. In order not to clutter the notation we assume a priori that all integrals appearing in the formulae of the following two lemmas are finite. In all applications of the paper, in particular Lemma \ref{lem:basicintegrals} below, this is easily verified using the regularity of $\Phi$ at the horizon and infinity and the explicit form of the multipliers chosen later. 

\begin{lemma}[Morawetz identity]\label{lem:Moridentity}
  Let $p:(r_+,+\infty)\to\RRR$ be a smooth function. For all $t^\ast_1 \leq t^\ast_2$ and for all smooth functions $\Phi$ on $[t^\ast_1,t^\ast_2]\times(r_+,+\infty)$, we have
  \begin{align}\label{est:pfCarleman1int1}
    \begin{aligned}
      0 = &\int_{t^\ast_1}^{t^\ast_2}\int_{-\infty}^{\frac{\pi}{2}} \bigg(\half\pr_{r^\ast}\le(\frac{\De_-\De_+}{\De_0^2}p\ri)|\pr_{t^\ast}\Phi|^2 + \half p'|\pr_{r^\ast}\Phi|^2 \\
      & \quad\quad +  \pr_{r^\ast}\le(\frac{\De_-}{\De_0}\ri)p\Re\le(\pr_{r^\ast}\Phi^\ast\pr_{t^\ast}\Phi\ri) - \half \le(pV\ri)'|\Phi|^2\bigg) \,\d t^\ast\d r^\ast\\
      & -\half \lim_{r\to+\infty}\int_{t_1^\ast}^{t_2^\ast}p\le(|\pr_{r^\ast}\Phi|^2 + |\pr_{t^\ast}\Phi|^2\ri)\,\d t^\ast  + \half k^2\ellmode(\ellmode+1) \lim_{r\to+\infty}\int_{t^\ast_1}^{t_2^\ast}p|\Phi|^2\,\d t^\ast  \\
      & + \half\lim_{r\to r_+}\int_{t^\ast_1}^{t^\ast_2}\le(\frac{\De_-\De_+}{\De_0^2}p|\pr_{t^\ast}\Phi|^2+p|\pr_{r^\ast}\Phi|^2 - pV |\Phi|^2\ri) \,\d t^\ast\\
      & + \half \le[\int_{-\infty}^{\frac{\pi}{2}}\le(\frac{\De_--\De_+}{\De_0}\ri)p|\pr_{r^\ast}\Phi|^2\,\d r^\ast\ri]_{t^\ast_1}^{t^\ast_2} +  \le[\int_{-\infty}^{\frac{\pi}{2}}\frac{\De_-\De_+}{\De_0^2} p \Re\le(\pr_{r^\ast}\Phi^\ast\pr_{t^\ast}\Phi\ri)\,\d r^\ast\ri]_{t^\ast_1}^{t^\ast_2} \\
      & + \int_{t^\ast_1}^{t^\ast_2}\int_{-\infty}^{\frac{\pi}{2}}p\Re\le(\RW\Phi\pr_{r^\ast}\Phi^\ast\ri)  \,\d t^\ast\d r^\ast. 
    \end{aligned}
  \end{align}
\end{lemma}
\begin{proof}
  Multiplying~\eqref{eq:pfenergid1} by $p(r)\pr_r\Phi^\ast$ and taking the real part, we get
  \begin{align}\label{eq:pfCarleman1}
    \begin{aligned}
      -p\Re\le(\RW\Phi\pr_r\Phi^\ast\ri) & = \frac{\De_-\De_+}{\De_0^2}\pr_{t^\ast}\Re\le(p\pr_{r^\ast}\Phi^\ast\pr_{t^\ast}\Phi\ri) - \half\pr_{r^\ast}\le(\frac{\De_-\De_+}{\De_0^2}p|\pr_{t^\ast}\Phi|^2\ri) + \half\pr_{r^\ast}\le(\frac{\De_-\De_+}{\De_0^2}p\ri)|\pr_{t^\ast}\Phi|^2\\
      & \quad - \half \pr_{r^\ast}\le(p|\pr_{r^\ast}\Phi|^2\ri) + \half p'|\pr_{r^\ast}\Phi|^2 \\
      & \quad + \pr_{r^\ast}\le(\frac{\De_-}{\De_0}\ri)p\Re\le(\pr_{r^\ast}\Phi^\ast\pr_{t^\ast}\Phi\ri) + \half\pr_{t^\ast}\le(\frac{\De_--\De_+}{\De_0} p |\pr_{r^\ast}\Phi|^2\ri) \\
      & \quad + \half \pr_{r^\ast}\le(pV|\Phi|^2\ri) - \half \pr_{r^\ast}\le(pV\ri)|\Phi|^2.
    \end{aligned}
  \end{align}
  Integrating~\eqref{eq:pfCarleman1} over the domain $(t_1^\ast,t_2^\ast)_{t^\ast}\times(-\infty,\frac{\pi}{2})_{r^\ast}$, we obtain~\eqref{est:pfCarleman1int1} and this finishes the proof of the lemma.
\end{proof}

\begin{lemma}[Elliptic identity]\label{lem:ellipticFourierbasic}
  Let $f:(r_+,+\infty)\to\RRR$ be a smooth function. For all $t^\ast_1 \leq t^\ast_2$ and for all smooth functions $\Phi$ on $[t^\ast_1,t^\ast_2]\times(r_+,+\infty)$, we have
  \begin{align}\label{eq:ellipticFourier}
    \begin{aligned}
      0 & = \int_{t^\ast_1}^{t^\ast_2}\int_{-\infty}^{\frac{\pi}{2}} \bigg(f |\pr_{r^\ast}\Phi|^2 + \le(fV-\half f''\ri) |\Phi|^2- f\frac{\De_-\De_+}{\De_0^2}|\pr_{t^\ast}\Phi|^2 - f\frac{\De_--\De_+}{\De_0}\Re\le(\pr_{t^\ast}\Phi^\ast\pr_{r^\ast}\Phi\ri) \bigg) \,\d t^\ast\d r^\ast  \\
      & \quad +\lim_{r\to+\infty}\int_{t^\ast_1}^{t^\ast_2}\le(-f\Re\le(\Phi^\ast\pr_{r^\ast}\Phi\ri) +\half f'|\Phi|^2\ri)\,\d t^\ast +\lim_{r\to r_+}\int_{t^\ast_1}^{t^\ast_2}\le(f\Re\le(\Phi^\ast\pr_{r^\ast}\Phi\ri) -\half f'|\Phi|^2\ri)\,\d t^\ast \\
      & \quad +\le[\int_{-\infty}^{\frac{\pi}{2}}\le(f\frac{\De_-\De_+}{\De_0^2}\Re\le(\Phi^\ast\pr_{t^\ast}\Phi\ri)+ \half f\pr_{r^\ast}\le(\frac{\De_-}{\De_0}\ri)|\Phi|^2 + f\frac{\De_--\De_+}{\De_0}\Re\le(\Phi^\ast\pr_{r^\ast}\Phi\ri)\ri)\,\d r^\ast\ri]_{t^\ast_1}^{t^\ast_2} \\
      & \quad +\int_{t^\ast_1}^{t^\ast_2}\int_{-\infty}^{\frac{\pi}{2}}f\Re\le(\Phi^\ast\RW\Phi\ri)\,\d t^\ast\d r^\ast. 
    \end{aligned}
  \end{align}
\end{lemma}
\begin{proof}
  Multiplying~\eqref{eq:pfenergid1} by $f\Phi^\ast$ and taking the real part, we get
  \begin{align}
    \label{eq:pfellipticFourier}
    \begin{aligned}
      -f\Re\le(\Phi^\ast\RW\Phi\ri) & = \pr_{t^\ast}\le(f\frac{\De_-\De_+}{\De_0^2}\Re\le(\Phi^\ast\pr_{t^\ast}\Phi\ri)\ri) - f\frac{\De_-\De_+}{\De_0^2}|\pr_{t^\ast}\Phi|^2 \\
      & \quad -\pr_{r^\ast}\le(f\Re\le(\Phi^\ast\pr_{r^\ast}\Phi\ri)\ri) + f |\pr_{r^\ast}\Phi|^2 + \half \pr_{r^\ast}\le( f' |\Phi|^2\ri) - \half f'' |\Phi|^2 \\
      & \quad +\half \pr_{t^\ast}\le(\pr_{r^\ast}\le(\frac{\De_-}{\De_0}\ri)f|\Phi|^2\ri) + \pr_{t^\ast}\le(f\frac{\De_--\De_+}{\De_0}\Re\le(\Phi^\ast\pr_{r^\ast}\Phi\ri)\ri) \\
      & \quad - f\frac{\De_--\De_+}{\De_0}\Re\le(\pr_{t^\ast}\Phi^\ast\pr_{r^\ast}\Phi\ri) + fV|\Phi|^2.
    \end{aligned}
  \end{align}
  Integrating~\eqref{eq:pfellipticFourier} over the domain $(t_1^\ast,t_2^\ast)_{t^\ast}\times(-\infty,\frac{\pi}{2})_{r^\ast}$, we obtain~\eqref{eq:ellipticFourier} completing the proof of the lemma.
\end{proof}

We now combine the two lemmata to derive the basic integral estimate we will use. 

\begin{lemma}[Integral estimates]\label{lem:basicintegrals}
  Let $p:(r_+,+\infty)\to\RRR$ be a real function, smooth on $[r_+,+\infty)$ on $(-\infty,\frac{\pi}{2}]_{r^\ast}$.\footnote{In particular, arbitrary many derivatives of $w^{-1} \partial_{r^\ast}|_{(t^\ast,r^\ast,\theta,\phi)}$ applied to $p$ extend continuously to $r=r_+$ and arbitrary many derivatives of $\partial_{r^\ast}|_{(t^\ast,r^\ast,\theta,\phi)}$ applied to $p$ extend continuously to $r=\infty$.} Let $t^\ast_2\geq t^\ast_1$ and let $\Phi$ be a smooth function on $[t^\ast_1,t^\ast_2]\times(r_+,+\infty)$, regular at the horizon~\eqref{eq:defreghorgeneral} and at infinity~\eqref{eq:defreginfgeneral} on $[t^\ast_1,t^\ast_2]$. We have
  \begin{align}\label{est:easytouseCarleman}
    \begin{aligned}
      &\bigg|\int_{t^\ast_1}^{t^\ast_2}\int_{-\infty}^{\frac{\pi}{2}} \le(p' \le(2\frac{\De_-\De_+}{\De_0^2}+\quar\le(\frac{\De_--\De_+}{\De_0}\ri)^2\ri)|\pr_{t^\ast}\Phi|^2 + p'|\pr_{r^\ast}\Phi|^2\ri)\,\d t^\ast\d r^\ast \\
      & -\int_{t^\ast_1}^{t^\ast_2}\int_{-\infty}^{\frac{\pi}{2}}\le(2 (pV)' + \half\le(-pV' - \half p'''\ri)\ri)|\Phi|^2\,\d t^\ast\d r^\ast\\
      & - \frac{3}{2}\lim_{r\to+\infty}\int_{t_1^\ast}^{t_2^\ast}p\le(|\pr_{r^\ast}\Phi|^2 + |\pr_{t^\ast}\Phi|^2\ri)\,\d t^\ast  + \half \lim_{r\to+\infty}\int_{t_1^\ast}^{t_2^\ast}p'\Re\le(\Phi^\ast\pr_{r^\ast}\Phi\ri)\,\d t^\ast \\
      & + \lim_{r\to+\infty}\int_{t^\ast_1}^{t^\ast_2} \le(-\quar p''+ \frac{3}{2}k^2\ellmode(\ellmode+1)p\ri)|\Phi|^2 \,\d t^\ast\bigg| \\
      & \les_{M,k} \le(\sup |p|+\le(\sup w^{-1}|p'|\ri)\ri) \bigg(\widetilde{\mathrm{E}}[\Phi](t^\ast_2) + \widetilde{\mathrm{E}}[\Phi](t^\ast_1) + \widetilde{\mathrm{E}}^\HH[\Phi](t^\ast_2;t^\ast_1) \\
      & \quad \quad \quad \quad \quad \quad \quad \quad \quad \quad \quad \quad \quad + \int_{t^\ast_1}^{t^\ast_2}\norm{w^{-1}\RW\Phi}_{L^2(\Si_{t^\ast})}\le(\widetilde{\mathrm{E}}[\Phi](t^\ast)\ri)^{1/2}\,\d t^\ast\bigg).
    \end{aligned}
  \end{align}
\end{lemma}
\begin{proof}
  Using that
  \begin{align*}
    -3(pV)' - \half\le(p'V-\half p'''\ri) = -\le(2(pV)' + \half \le(pV' -\half p'''\ri)\ri),
  \end{align*}
  $3$-times~\eqref{est:pfCarleman1int1} minus $\half$-times~\eqref{eq:ellipticFourier}, with $f=p'$ and where we use the regularity at the horizon~\eqref{eq:defreghorgeneral} for $\Phi$, gives
  \begin{align}\label{est:easytouseCarlemanproof}
    \begin{aligned}
      &\bigg|\int_{t^\ast_1}^{t^\ast_2}\int_{-\infty}^{\frac{\pi}{2}} \bigg(\le(\frac{3}{2}\pr_{r^\ast}\le(\frac{\De_-\De_+}{\De_0^2}p\ri)+\half p'\frac{\De_-\De_+}{\De_0^2}\ri)|\pr_{t^\ast}\Phi|^2 \\
      & \quad + \le(3\pr_{r^\ast}\le(\frac{\De_-}{\De_0}\ri)p +\half p'\frac{\De_--\De_+}{\De_0}\ri)\Re\le(\pr_{r^\ast}\Phi^\ast\pr_{t^\ast}\Phi\ri)  + p'|\pr_{r^\ast}\Phi|^2\bigg)\,\d t^\ast\d r^\ast \\
      & -\int_{t^\ast_1}^{t^\ast_2}\int_{-\infty}^{\frac{\pi}{2}}\le(2 (pV)' + \half\le(-pV' - \half p'''\ri)\ri)|\Phi|^2\,\d t^\ast\d r^\ast\\
      & - \frac{3}{2}\lim_{r\to+\infty}\int_{t_1^\ast}^{t_2^\ast}p\le(|\pr_{r^\ast}\Phi|^2 + |\pr_{t^\ast}\Phi|^2\ri)\,\d t^\ast  + \half \lim_{r\to+\infty}\int_{t_1^\ast}^{t_2^\ast}p'\Re\le(\Phi^\ast\pr_{r^\ast}\Phi\ri)\,\d t^\ast \\
      & + \lim_{r\to+\infty}\int_{t^\ast_1}^{t^\ast_2} \le(-\quar p''+ \frac{3}{2}k^2\ellmode(\ellmode+1)p\ri)|\Phi|^2 \,\d t^\ast\bigg| \\
      & \les_{M,k} (\sup |p|+\sup w^{-1}|p'|)\le(\widetilde{\mathrm{E}}[\Phi](t^\ast_1)+\widetilde{\mathrm{E}}[\Phi](t^\ast_2) + \int_{t^\ast_1}^{t^\ast_2}\norm{w^{-1}\RW\Phi}_{L^2(\Si_{t^\ast})}\le(\widetilde{\mathrm{E}}[\Phi](t^\ast)\ri)^{1/2}\,\d t^\ast\ri).
    \end{aligned}
  \end{align} 
  Let $q:(r_+,+\infty)\to\RRR$ be a function such that $q$ is smooth on $[r_+,+\infty)_r$ and $q\xrightarrow{r\to+\infty}0$, which will be determined in the sequel. Multiplying~\eqref{eq:pfenergid1} by $q(r)\pr_{t^\ast}\Phi^\ast$ and taking the real part, we get
  \begin{align}\label{eq:pfredshift1}
    \begin{aligned}
      -q\Re\le(\RW\Phi\pr_{t^\ast}\Phi^\ast\ri) & = \half \pr_{t^\ast}\le(\frac{\De_-\De_+}{\De_0^2} q |\pr_{t^\ast}\Phi|^2 \ri) - \pr_{r^\ast}\Re\le(q\pr_{r^\ast}\Phi\pr_{t^\ast}\Phi^\ast\ri) + \half \pr_{t^\ast}\le(q|\pr_{r^\ast}\Phi|^2\ri) \\
      & \quad + q'\Re\le(\pr_{r^\ast}\Phi\pr_{t^\ast}\Phi^\ast\ri) \\
      & \quad + \half\frac{\De_--\De_+}{\De_0}\pr_{r^\ast}\le(q|\pr_{t^\ast}\Phi|^2\ri) + \pr_{r^\ast}\le(\frac{\De_-}{\De_0}\ri) q |\pr_{t^\ast}\Phi|^2 \\
      & \quad - \half\frac{\De_--\De_+}{\De_0}q'|\pr_{t^\ast}\Phi|^2 + \half V\pr_{t^\ast}\le(q|\Phi|^2\ri).
    \end{aligned}
  \end{align}
  Integrating~\eqref{eq:pfredshift1} over the domain $(t_1^\ast,t_2^\ast)_{t^\ast}\times(-\infty,\frac{\pi}{2})_{r^\ast}$, using the regularity at the horizon~\eqref{eq:defreghorgeneral} for $\Phi$, we get
  \begin{align}\label{est:pfredshift1}
    \begin{aligned}
      &\bigg|\int_{t^\ast_1}^{t^\ast_2}\int_{-\infty}^{\frac{\pi}{2}} \bigg(q'\Re\le(\pr_{r^\ast}\Phi\pr_{t^\ast}\Phi^\ast\ri)- \half\frac{\De_--\De_+}{\De_0}q'|\pr_{t^\ast}\Phi|^2\bigg) \,\d t^\ast\d r^\ast \,\d t^\ast\bigg| \\
      & \les_{M,k} \le(\sup |q|\ri) \bigg(\widetilde{\mathrm{E}}[\Phi](t^\ast_2) + \widetilde{\mathrm{E}}[\Phi](t^\ast_1) + \widetilde{\mathrm{E}}^\HH[\Phi](t^\ast_2;t^\ast_1) +\int_{t^\ast_1}^{t^\ast_2}\int_{-\infty}^{\frac{\pi}{2}}|\RW\Phi||\pr_{t^\ast}\Phi|\,\d t^\ast\d r^\ast\bigg).
    \end{aligned}
  \end{align}
  Take
  \begin{align*}
    q(r^\ast) & := -\int_{r^\ast}^{\frac{\pi}{2}} \le(3\pr_{r^\ast}\le(\frac{\De_-}{\De_0}\ri)p +\half p'\frac{\De_--\De_+}{\De_0}\ri) \,\d r^{\ast,'},
  \end{align*}
  so that $q\xrightarrow{r\to+\infty}0$ and $q' = 3\pr_{r^\ast}\le(\frac{\De_-}{\De_0}\ri)p +\half p'\frac{\De_--\De_+}{\De_0}$. Note that $q$ is a smooth function on $[r_+,+\infty)_r$ and $\sup |q| \les_{M,k}(\sup |p|+\sup w^{-1}|p'|)$. Moreover, we have the identity
  \begin{align}\label{eq:derivDeltasproofCarlemanlemma}
    \begin{aligned}
      \half \frac{\De_--\De_+}{\De_0} q' + \le(\frac{3}{2}\pr_{r^\ast}\le(\frac{\De_-\De_+}{\De_0^2}p\ri)+\half p'\frac{\De_-\De_+}{\De_0^2}\ri) = \le(\quar\le(\frac{\De_--\De_+}{\De_0}\ri)^2+2\frac{\De_-\De_+}{\De_0^2}\ri)p'.
    \end{aligned}
  \end{align}  
  Combining~\eqref{est:easytouseCarlemanproof} and~\eqref{est:pfredshift1}, using~\eqref{eq:derivDeltasproofCarlemanlemma}, we obtain~\eqref{est:easytouseCarleman} and this finishes the proof of the lemma.
\end{proof}

\subsection{Proof of the redshift estimate}\label{sec:Carlemanredshift}
We now prove Proposition~\ref{prop:redshift}. It will be an immediate consequence of the following redshift lemma, which is also used repeatedly later (see Sections~\ref{sec:Carlemanestimatesfirst}, \ref{sec:Fourier} and \ref{sec:energyestpsipm2}).
To state it, we define the auxiliary degenerate energy
\begin{align} \label{est:weakenergy}
  \mathrm{E}^{w}[\Psi] := \int_{-\infty}^{\frac{\pi}{2}} w  \le(|\pr_{t^\ast}\Psi|^2 + |\pr_{r^\ast}\Psi|^2 +  \ellmode^2|\Psi|^2\ri)\,\d r^\ast    \les_{M,k} \widetilde{\mathrm{E}}[\Psi].
\end{align}

\begin{lemma}[Redshift lemma]\label{lem:redshift}
For all spacetime functions $\Psi$ regular at the horizon in the sense of~\eqref{eq:defreghorgeneral}, we have for all $t^\ast_2\geq t^\ast_1$ the estimate 
  \begin{align}\label{est:redshift3}  
    \begin{aligned}
      & \overline{\mathrm{E}}[\Psi](t^\ast_2)  + \int_{t^\ast_1}^{t^\ast_2} \overline{\mathrm{E}}[\Psi](t^\ast)\,\d t^\ast + \overline{\mathrm{E}}^\HH[\Psi](t^\ast_2;t^\ast_1) \\
      \les_{M,k} & \; \overline{\mathrm{E}}[\Psi](t^\ast_1) + \widetilde{\mathrm{E}}[\Psi](t^\ast_2)  + \widetilde{\mathrm{E}}^\HH[\Psi](t^\ast_2;t^\ast_1) + \int_{t^\ast_1}^{t^\ast_2} {\mathrm{E}}^{w}[\Psi](t^\ast)\,\d t^\ast + \int_{t^\ast_1}^{t^\ast_2}\int_{-\infty}^{\frac{\pi}{2}}w^{-1}|\RW\Psi|^2\,\d t^\ast\d r^\ast \, .
    \end{aligned}
  \end{align}
\end{lemma}

\begin{remark}
With slightly more work, one can prove the estimate (\ref{est:redshift3}) replacing $E^w[\Psi]$ on the right hand side by an energy supported strictly away from the horizon. However, the above weaker version is sufficient for our purposes and can be proven using only the multiplier identity \eqref{est:pfCarleman1int1}.
\end{remark}

\begin{proof}[Proof of Proposition~\ref{prop:redshift}]  
  Using the energy boundedness~\eqref{est:energyPsiDRmain} for solutions to the  Regge-Wheeler problem~\eqref{sys:RWdeco}, estimate~\eqref{est:redshiftmain} follows from Lemma~\ref{lem:redshift} and this proves Proposition~\ref{prop:redshift}.
\end{proof}
\begin{proof}[Proof of Lemma~\ref{lem:redshift}]
  We recall the weight $w=\frac{\Delta_-}{r^4}$ and choose $p=-w^{-1}\chi$ in~\eqref{est:pfCarleman1int1}, where $\chi$ is a cut-off function with support in $[r_+,4M]$  and equal to $1$ in $[r_+,3M]$. We now estimate all the terms appearing in the multiplier identity~\eqref{est:pfCarleman1int1}. We have
   \begin{align} \label{display1}
  \textrm{terms in the first two lines of \eqref{est:pfCarleman1int1}}
 \geq  c_{M,k} \int_{t^\ast_1}^{t^\ast_2}  \overline{\mathrm{E}}[\Psi](t^\ast)\,\d t^\ast  - C_{M,k}  \int_{t^\ast_1}^{t^\ast_2} {\mathrm{E}}^{w}[\Psi](t^\ast)\,\d t^\ast 
 \end{align}
 for constants $C_{M,k}$ and $c_{M,k}$ depending only on $M$ and $k$. The boundary terms at infinity in the third line of \eqref{est:pfCarleman1int1} vanish because the chosen $p$ vanishes for $r \geq 4M$. For the boundary terms at $t^\ast_1$ and $t^\ast_2$ appearing in line five of~\eqref{est:pfCarleman1int1} we have 
    \begin{align} \label{display2}
  \textrm{boundary-terms on constant $t^\ast$ of  \eqref{est:pfCarleman1int1}}
 \geq c_{M,k} \cdot \overline{\mathrm{E}}[\Psi](t^\ast_2) - C_{M,k} \cdot  \widetilde{\mathrm{E}}[\Psi](t^\ast_2) -C_{M,k}\cdot \overline{\mathrm{E}}[\Psi](t^\ast_1)  \, .
  \end{align}
 For the boundary term on the horizon appearing in line four  of \eqref{est:pfCarleman1int1} it is easy to see that
  \begin{align} \label{display3}
    \textrm{boundary-terms on $r=r_+$ of \eqref{est:pfCarleman1int1}} \geq c_{M,k}  \cdot \overline{\mathrm{E}}^\HH[\Psi](t^\ast_2;t^\ast_1) - C_{M,k} \widetilde{\mathrm{E}}^\HH[\Psi](t^\ast_2;t^\ast_1)- \frac{3M}{r_+} \int_{t_1^\ast}^{t_2^\ast} |\Psi|^2 |_{r=r_+} dt^\ast  . 
  \end{align}
 The following Hardy inequality (whose simple proof we leave to the reader) allows to estimate  for any $\delta>0$
  \begin{align} \label{hardy2}
  \int_{t_1^\ast}^{t_2^\ast} |\Psi|^2 |_{r=r_+} dt^\ast \leq \delta  \int_{t^\ast_1}^{t^\ast_2}  \overline{\mathrm{E}}[\Psi](t^\ast)\,\d t^\ast  + C_{\delta}  \int_{t^\ast_1}^{t^\ast_2} {\mathrm{E}}^{w}[\Psi](t^\ast)\,\d t^\ast \, .
  \end{align}
 Finally, for the last term in appearing in~\eqref{est:pfCarleman1int1} we have for any $\delta>0$
  \begin{align} \label{display4}
   \textrm{last term in \eqref{est:pfCarleman1int1}} \geq -  \delta \int_{t^\ast_1}^{t^\ast_2}\overline{\mathrm{E}}[\Psi](t^\ast) \d t^\ast- C_{M,k}\int_{t^\ast_1}^{t^\ast_2} {\mathrm{E}}^{w}[\Psi](t^\ast)\d t^\ast +\frac{C_{M,k}}{\delta} \int_{t^\ast_1}^{t^\ast_2}\int_{-\infty}^{\frac{\pi}{2}}\frac{|\RW\Psi|^2}{w} \,\d t^\ast\d r^\ast \, .
  \end{align} 
Adding the estimates (\ref{display1}), (\ref{display2}), (\ref{display3}) and (\ref{display4}) and inserting (\ref{hardy2}) into (\ref{display3}), we produce the desired (\ref{est:redshift3}) after choosing $\delta$ sufficiently small so that we can absorb the two terms proportional to $\delta$ by the good first term on the right hand side of \eqref{display1}. 
 \end{proof}

\subsection{Carleman estimates without time-frequency decompositions}\label{sec:Carlemanestimatesfirst}
In this subsection, we prove Proposition~\ref{prop:integralestimatesv} in the weaker $\mathfrak{p}=2$ case using only physical space multipliers. 

Let us for simplicity set
\begin{align*}
  \EE[\Psi](t^\ast_2;t^\ast_1) & := \le(\widetilde{\mathrm{E}}[\Psi](t^\ast_2) + \widetilde{\mathrm{E}}[\Psi](t^\ast_1) + \widetilde{\mathrm{E}}^\HH[\Psi](t^\ast_2;t^\ast_1)\ri)
\end{align*}
in this section.\footnote{The quantity $\EE[\Psi](t^\ast_2;t^\ast_1)$ will eventually be controlled using the energy boundedness~\eqref{est:energyPsiDRmain}.} We also define $p=e^f$ with $f=\frac{\Ka}{r}$, where $\Ka>0$ is a constant which will be determined. Setting $p=e^f$ with $f=\frac{\Ka}{r}$ in the integral estimate~\eqref{est:easytouseCarleman}, using that $p'=-\Ka w e^f$ and that $\le(\sup |p|\ri)+\le(\sup w^{-1}|p'|\ri) \leq (1+\Ka)e^{\frac{\Ka}{r_+}}=e^{C\Ka}$ with $C=C(M,k)>0$, we get
\begin{align}\label{est:easyCarleman}
  \begin{aligned}
    &\int_{t^\ast_1}^{t^\ast_2}\int_{-\infty}^{\frac{\pi}{2}} \le(\Ka w e^f \le(2\frac{\De_-\De_+}{\De_0^2}+\quar\le(\frac{\De_--\De_+}{\De_0}\ri)^2\ri)|\pr_{t^\ast}\Psi|^2 + \Ka w e^f|\pr_{r^\ast}\Psi|^2\ri)\,\d t^\ast\d r^\ast \\
    & +\int_{t^\ast_1}^{t^\ast_2}\int_{-\infty}^{\frac{\pi}{2}}\le(2 (pV)' + \half\le(-pV' - \half p'''\ri)\ri)|\Psi|^2\,\d t^\ast\d r^\ast\\
    & + \frac{3}{2}\lim_{r\to+\infty}\int_{t_1^\ast}^{t_2^\ast}\le(|\pr_{r^\ast}\Psi|^2 + |\pr_{t^\ast}\Psi|^2\ri)\,\d t^\ast  + \half \lim_{r\to+\infty}\int_{t_1^\ast}^{t_2^\ast}k^2\Ka\Re\le(\Psi^\ast\pr_{r^\ast}\Psi\ri)\,\d t^\ast \\
    & + \lim_{r\to+\infty}\int_{t^\ast_1}^{t^\ast_2} \le(\quar k^4\Ka^2- \frac{3}{2}k^2\ellmode(\ellmode+1)\ri)|\Psi|^2 \,\d t^\ast \\
    & \les_{M,k} e^{C\Ka}\EE[\Psi](t^\ast_2;t^\ast_1). 
  \end{aligned}
\end{align}
For the boundary terms at infinity in~\eqref{est:easyCarleman}, we have
\begin{align*}
  \begin{aligned}
  \frac{3}{2}|\pr_{r^\ast}\Psi|^2 + \half k^2\Ka\Re\le(\Psi^\ast\pr_{r^\ast}\Psi\ri) + \quar k^4\Ka^2|\Psi|^2  = \frac{1}{2}|\pr_{r^\ast}\Psi|^2 + \quar\le|2\pr_{r^\ast}\Psi  + \half k^2\Ka\Psi\ri|^2 + \frac{1}{8} k^4\Ka^2|\Psi|^2.  
  \end{aligned}
\end{align*}
Hence, provided that $\Ka$ is such that $\frac{1}{8}k^4\Ka^2 - 2 k^2\ellmode(\ellmode+1)\geq 0$, \emph{i.e.} $\Ka \geq 8k^{-1}\sqrt{\ellmode(\ellmode+1)}$, we have
\begin{align}\label{est:bdyineasyCarleman}
  \begin{aligned}
    & \half \lim_{r\to+\infty}\int_{t_1^\ast}^{t_2^\ast}\le(|\pr_{r^\ast}\Psi|^2 + |\pr_{t^\ast}\Psi|^2 + k^2\ellmode(\ellmode+1)|\Psi|^2\ri)\,\d t^\ast \\
    \leq & \; \frac{3}{2}\lim_{r\to+\infty}\int_{t_1^\ast}^{t_2^\ast}\le(|\pr_{r^\ast}\Psi|^2 + |\pr_{t^\ast}\Psi|^2\ri)\,\d t^\ast  + \half \lim_{r\to+\infty}\int_{t_1^\ast}^{t_2^\ast}k^2\Ka\Re\le(\Psi^\ast\pr_{r^\ast}\Psi\ri)\,\d t^\ast \\
    & + \lim_{r\to+\infty}\int_{t^\ast_1}^{t^\ast_2} \le(\quar k^4\Ka^2- \frac{3}{2}k^2\ellmode(\ellmode+1)\ri)|\Psi|^2 \,\d t^\ast.
  \end{aligned}
\end{align}

To estimate the potential terms, we have the following lemmas.
\begin{lemma}\label{lem:easypotentialsigns}
  Let $\ellmode\geq 2$. There exists $\varep_0(M,k)>0$ such that, for all $\Ka\geq \varep^{-1}_0\ellmode^2$, there exists $r_{\Ka,-1},r_{\Ka,+1}\in(r_+,+\infty)$ satisfying $r_+<r_{\Ka,-1}<r_{\Ka,0}<r_{\Ka,+1}$, where $r_{\Ka,0}:=r_++\Ka^{-1}r_+^2$, and we have
  \begin{align}\label{eq:defrK0rK2insec}
    \begin{aligned}
      2(pV)' -\half pV' - \quar p''' & \geq 0, && \text{if $r\leq r_{\Ka,-1}$,}\\
      2(pV)' -\half pV' - \quar p''' & \leq 0, && \text{if $r_{\Ka,-1}\leq r\leq r_{\Ka,+1}$,}\\
      2(pV)' -\half pV' - \quar p''' & \geq 0, && \text{if $r\geq r_{\Ka,+1}$.}
    \end{aligned}
  \end{align}
  Moreover, we have
  \begin{align}\label{eq:DLrK0rK2insec}
    \begin{aligned}
    \le|r_{\Ka,-1} -\le(r_+ +\Ka^{-1}r_+^2 - \Ka^{-3/2}r_+^2\sqrt{\frac{\le(\ellmode(\ellmode+1)-\half\ri)}{\le(\frac{6M}{r_+^2}-\frac{2}{r_+}\ri)}}\ri)\ri| & \les_{M,k} \Ka^{-1} (\Ka^{-1}\ellmode^2), \\
    \le|r_{\Ka,+1} -\le(r_+ +\Ka^{-1}r_+^2 + \Ka^{-3/2}r_+^2\sqrt{\frac{\le(\ellmode(\ellmode+1)-\half\ri)}{\le(\frac{6M}{r_+^2}-\frac{2}{r_+}\ri)}} \ri)\ri| & \les_{M,k} \Ka^{-1} (\Ka^{-1}\ellmode^2).
    \end{aligned}
  \end{align}
\end{lemma}
\begin{proof}
  See Appendix~\ref{sec:Carlemanmultipliers}.
\end{proof}
Define
\begin{align}\label{eq:defrKapm2insec}
  \begin{aligned}
    r_{\Ka,\pm2} & := r_+ +\Ka^{-1}r_+^2\pm \Ka^{-1}r_+^2\le(\Ka^{-1}\ellmode^2\ri)^{1/4}.
  \end{aligned}
\end{align}
Using~\eqref{eq:DLrK0rK2insec}, we have
\begin{align}\label{eq:comparaisonrKaiinsec}
  \begin{aligned}
    r_+ < r_{\Ka,-2} \leq r_{\Ka,-1} \leq r_{\Ka,0} \leq r_{\Ka,+1} \leq r_{\Ka,+2},
  \end{aligned}
\end{align}
for $\Ka^{-1}\ellmode^2$ sufficiently small depending on $M,k$. Between these radii, we have the following bounds for the potential.
\begin{lemma}\label{lem:easyboundspotential}
  Let $\ellmode\geq 2$. There exists $\varep_{1}(M,k)>0$ such that for all $\Ka \geq \varep_1^{-1}\ellmode^2$, we have
  \begin{align}
    2(pV)' -\half pV' - \quar p''' & \gtrsim_{M,k} w\ellmode^2e^f, && \text{for $r_{+} < r < r_{\Ka,-2}$,} \label{est:Br+rKa-2insec}\\
    \le|2(pV)' -\half pV' - \quar p'''\ri| & \les_{M,k} w \ellmode^2 e^f, && \text{for $r_{\Ka,-1} < r < r_{\Ka,+1}$,} \label{est:BrKa+-rKa++insec}\\
    2(pV)' -\half pV' - \quar p''' & \gtrsim_{M,k} \le(\Ka^{-1}\ellmode^2\ri)^{1/2} (\Ka w)^2 e^f \gtrsim_{M,k} w \ellmode^2 e^f, && \text{for $r_{\Ka,+2} < r < +\infty$.} \label{est:BrKa2+inftyinsec}
  \end{align}
\end{lemma}
\begin{proof}
  See Appendix~\ref{sec:Carlemanmultipliers}.
\end{proof}

With these two lemmas, we can now control the second line of~\eqref{est:easyCarleman}. This is the following lemma.
\begin{lemma}
  Let $\ellmode\geq 2$. There exists $\varep_2(M,k)>0$ such that for all $\Ka \geq \varep_3^{-1}\ellmode^2$, we have
  \begin{align}\label{est:finaleasypotentialestimate}
    \begin{aligned}
      \int_{-\infty}^{\frac{\pi}{2}} \ellmode^2we^f |\Psi|^2\,\d r^\ast & \les_{M,k} \int_{-\infty}^{\frac{\pi}{2}}\Ka w e^f|\pr_{r^\ast}\Psi|^2 \,\d r^\ast +\int_{-\infty}^{\frac{\pi}{2}}\le(2 (pV)' -\half pV' - \quar p'''\ri)|\Psi|^2\,\d r^\ast,
    \end{aligned}
  \end{align}
  for all functions $\Psi$.
\end{lemma}
\begin{proof}
  Using~\eqref{est:Br+rKa-2insec} and~\eqref{est:BrKa2+inftyinsec}, it is immediate that 
  \begin{align}
    \label{est:firstoutsiderKapm2easywell2ef}
    \begin{aligned}
    & \int_{-\infty}^{r^\ast_{\Ka,-2}} w\ellmode^2 e^f |\Psi|^2\,\d r^\ast + \int_{r^\ast_{\Ka,+2}}^{\frac{\pi}{2}} w\ellmode^2 e^f |\Psi|^2\,\d r^\ast \\
    \les_{M,k} & \; \int_{-\infty}^{r^\ast_{\Ka,-2}}\le(2(pV)' - \half pV' - \quar p'''\ri)|\Psi|^2\,\d r^\ast + \int_{r^\ast_{\Ka,+2}}^{\frac{\pi}{2}}\le(2(pV)' - \half pV' - \quar p'''\ri)|\Psi|^2\,\d r^\ast ,
    \end{aligned}
  \end{align}
 where we have used the bijection between $r$ and $r^\ast$ (\ref{tortoise}) to write $r^\ast_{\Ka, \pm 2}$ for the $r^\ast$-value corresponding to $r_{\Ka, \pm 2}$.
  Using that $w(r) = O_{r\to r_+}(r-r_+)$ and the definition~\eqref{eq:defrKapm2insec} of $r_{\Ka,\pm2}$, we have $\Ka w \sim_{M,k}1$ on $(r_{\Ka,-2},r_{\Ka,2})$. Thus, using the change of variable $\d r^\ast = w^{-1}r^{-2}\d r$, we have
  \begin{align}\label{est:errintpV'v1bis2}
    \begin{aligned}
      \int_{r^\ast_{\Ka,-2}}^{r^\ast_{\Ka,2}}w\ellmode^2e^f|\Psi|^2\,\d r^\ast & \les_{M,k}\Ka\underbrace{\le|r_{\Ka,2}-r_{\Ka,-2}\ri|}_{\les \Ka^{-1}\Ka^{-1/4}\ellmode^{1/2}}\sup_{r^\ast_{\Ka,-2}\leq r^\ast \leq r^\ast_{\Ka,2}}\le(w\ellmode^2e^f|\Psi|^2\ri) \\
      & \les_{M,k} \le(\Ka^{-1}\ellmode^2\ri)^{1/4}\sup_{r^\ast_{\Ka,-2}\leq r^\ast \leq r^\ast_{\Ka,2}}\le(w\ellmode^2e^f|\Psi|^2\ri).
    \end{aligned}
  \end{align}
  Using the fundamental theorem of calculus (with a cut-off between $r=2M$ and $r=3M$), using that $|(w\ellmode^2e^f)'|\les_{M,k}(\Ka w)w\ellmode^2e^f$ on $(r_{\Ka,-2},+\infty)$, that $\Ka w \les_{M,k}1$ on $(r_{\Ka,-2},r_{\Ka,2})$ and that $\Ka w\gtrsim_{M,k} 1$ on $(r^\ast_{\Ka,+2},+\infty)$, we get
  \begin{align}\label{est:controlsuppVPsiv1bis}
    \begin{aligned}
      \sup_{r^\ast_{\Ka,-2}\leq r^\ast \leq r^\ast_{\Ka,2}} \le(w\ellmode^2e^f|\Psi|^2\ri) & \les_{M,k} \int_{r^\ast_{\Ka,-2}}^{\frac{\pi}{2}} (\Ka w)w\ellmode^2e^f|\Psi|^2\,\d r^\ast + \underbrace{\int_{r=2M}^{r=3M} w\ellmode^2e^f|\Psi|^2\,\d r^\ast}_{\les \int_{r^\ast_{\Ka,2}}^{\frac{\pi}{2}} (\Ka w) w\ellmode^2e^f|\Psi|^2\,\d r^\ast} \\
      & \quad + \int_{r^\ast_{\Ka,-2}}^{\frac{\pi}{2}}w\ellmode^2e^f|\Psi||\pr_{r^\ast}\Psi| \, \d r^\ast\\
      & \les_{M,k} \int_{r^\ast_{\Ka,-2}}^{r^\ast_{\Ka,2}} (\Ka w)w\ellmode^2e^f|\Psi|^2\,\d r^\ast + \int_{r^\ast_{\Ka,2}}^{\frac{\pi}{2}} (\Ka w)w\ellmode^2e^f|\Psi|^2\,\d r^\ast \\
      & \quad + \int_{r^\ast_{\Ka,-2}}^{\frac{\pi}{2}}w\ellmode^2e^f|\Psi|^2\,\d r^\ast + \int_{r^\ast_{\Ka,-2}}^{\frac{\pi}{2}}w\ellmode^2e^f|\pr_{r^\ast}\Psi|^2\,\d r^\ast \\
      & \les_{M,k} \int_{r^\ast_{\Ka,-2}}^{r^\ast_{\Ka,2}} w\ellmode^2e^f|\Psi|^2\,\d r^\ast + \int_{r^\ast_{\Ka,2}}^{\frac{\pi}{2}} (\Ka w)w\ellmode^2e^f|\Psi|^2\,\d r^\ast \\
      & \quad + \le(\Ka^{-1}\ellmode^2\ri) \int_{r^\ast_{\Ka,-2}}^{\frac{\pi}{2}}\Ka w e^f|\pr_{r^\ast}\Psi|^2\,\d r^\ast.
    \end{aligned}
  \end{align}
  The first term on the RHS of~\eqref{est:controlsuppVPsiv1bis} can be absorbed on the LHS by re-using estimate~\eqref{est:errintpV'v1bis2} provided that $\Ka^{-1}\ellmode^2$ is sufficiently small depending on $M,k$. From that argument and the bound~\eqref{est:BrKa2+inftyinsec} on $2(pV)'-\half pV'-\quar p'''$, we get from~\eqref{est:controlsuppVPsiv1bis} that
  \begin{align}
    \sup_{r^\ast_{\Ka,-2}\leq r^\ast \leq r^\ast_{\Ka,2}} \le(w\ellmode^2e^f|\Psi|^2\ri) & \les_{M,k} \int_{r^\ast_{\Ka,2}}^{\frac{\pi}{2}} (\Ka w)w\ellmode^2e^f|\Psi|^2\,\d r^\ast + \le(\Ka^{-1}\ellmode^2\ri) \int_{r^\ast_{\Ka,-2}}^{\frac{\pi}{2}}\Ka w e^f|\pr_{r^\ast}\Psi|^2\,\d r^\ast \nonumber\\
                                                                                     & \les_{M,k} \le(\Ka^{-1}\ellmode^2\ri)\le(\int_{r^\ast_{\Ka,2}}^{\frac{\pi}{2}} (\Ka w)^2e^f|\Psi|^2\,\d r^\ast + \int_{r^\ast_{\Ka,-2}}^{\frac{\pi}{2}}\Ka w e^f|\pr_{r^\ast}\Psi|^2\,\d r^\ast\ri) \label{est:controlsuppVPsiv2bis}\\
                                                                                     & \les_{M,k}  \le(\Ka^{-1}\ellmode^2\ri)^{1/2} \bigg(\int_{r^\ast_{\Ka,2}}^{\frac{\pi}{2}}\le(2(pV)'-\half pV'-\quar p'''\ri)|\Psi|^2\,\d r^\ast+ \int_{r^\ast_{\Ka,-2}}^{\frac{\pi}{2}}\Ka w e^f|\pr_{r^\ast}\Psi|^2\,\d r^\ast\bigg).\nonumber
  \end{align}
  In particular, using~\eqref{est:errintpV'v1bis2}, \eqref{est:controlsuppVPsiv2bis} yields
  \begin{align}\label{est:intwell2efrKapm2}
    \int_{r^\ast_{\Ka,-2}}^{r^\ast_{\Ka,+2}}w\ellmode^2e^f|\Psi|^2\,\d r^\ast \les_{M,k} \int_{r^\ast_{\Ka,2}}^{\frac{\pi}{2}}\le(2(pV)'-\half pV'-\quar p'''\ri)|\Psi|^2\,\d r^\ast+ \int_{r^\ast_{\Ka,-2}}^{\frac{\pi}{2}}\Ka w e^f|\pr_{r^\ast}\Psi|^2\,\d r^\ast.
  \end{align}
  From~\eqref{eq:comparaisonrKaiinsec},~\eqref{est:BrKa+-rKa++insec} and~\eqref{est:controlsuppVPsiv2bis}, we get
  \begin{align}\label{est:errintpV'finalbis}
    \begin{aligned}
      -\int_{r^\ast_{\Ka,-1}}^{r^\ast_{\Ka,+1}}\le(2(pV)' - \half pV' - \quar p'''\ri)|\Psi|^2\,\d r^\ast & \les_{M,k} \int_{r^\ast_{\Ka,-2}}^{r^\ast_{\Ka,+2}}w\ellmode^2e^f|\Psi|^2\,\d r^\ast \\
      & \les_{M,k} \le(\Ka^{-1}\ellmode^2\ri)^{3/4} \bigg(\int_{r^\ast_{\Ka,2}}^{\frac{\pi}{2}}\le(2(pV)' - \half pV' - \quar p'''\ri)|\Psi|^2\,\d r^\ast \\
      & \quad \quad\quad\quad + \int_{r^\ast_{\Ka,-2}}^{\frac{\pi}{2}}\Ka w e^f|\pr_{r^\ast}\Psi|^2\,\d r^\ast\bigg).
    \end{aligned}
  \end{align}
  Hence, using~\eqref{eq:defrK0rK2insec} and for $\Ka^{-1}\ellmode^2$ sufficiently small depending only on $M,k$, an absorption argument gives
  \begin{align}\label{est:onlypositiveintheRHS}
    \begin{aligned}
      & \int_{-\infty}^{\frac{\pi}{2}}\Ka w e^f|\pr_{r^\ast}\Psi|^2 \,\d r^\ast +\int_{-\infty}^{r^\ast_{\Ka,-2}}\le(2 (pV)' -\half pV' - \quar p'''\ri)|\Psi|^2\,\d r^\ast \\
      & \quad +\int_{r^\ast_{\Ka,+2}}^{\frac{\pi}{2}}\le(2 (pV)' -\half pV' - \quar p'''\ri)|\Psi|^2\,\d r^\ast \\
      \leq & \; 2 \bigg(\int_{-\infty}^{\frac{\pi}{2}}\Ka w e^f|\pr_{r^\ast}\Psi|^2 \,\d r^\ast +\int_{-\infty}^{\frac{\pi}{2}}\le(2 (pV)' -\half pV' - \quar p'''\ri)|\Psi|^2\,\d r^\ast \bigg).
    \end{aligned}
  \end{align}
  Combining~\eqref{est:firstoutsiderKapm2easywell2ef} and~\eqref{est:intwell2efrKapm2} with~\eqref{est:onlypositiveintheRHS}, we have obtained the desired~\eqref{est:finaleasypotentialestimate} and this finishes the proof of the lemma.
\end{proof}
Now, for all $\Ka$ such that
\begin{align}\label{est:Kabound}
  \Ka \geq \max\le(8k^{-1}\sqrt{\ellmode(\ellmode+1)},\varep_{2}^{-1}\ellmode^2\ri),
\end{align}
we infer, by plugging~\eqref{est:bdyineasyCarleman} and~\eqref{est:finaleasypotentialestimate} in~\eqref{est:easyCarleman}, that
\begin{align*}
  &\int_{t^\ast_1}^{t^\ast_2}\int_{-\infty}^{\frac{\pi}{2}} \le(\Ka w e^f \le(2\frac{\De_-\De_+}{\De_0^2}+\quar\le(\frac{\De_--\De_+}{\De_0}\ri)^2\ri)|\pr_{t^\ast}\Psi|^2 + \Ka w e^f|\pr_{r^\ast}\Psi|^2\ri)\,\d t^\ast\d r^\ast \\
    & +\int_{t^\ast_1}^{t^\ast_2}\int_{-\infty}^{\frac{\pi}{2}}\ellmode^2we^f|\Psi|^2\,\d t^\ast\d r^\ast + \half \lim_{r\to+\infty}\int_{t_1^\ast}^{t_2^\ast}\le(|\pr_{r^\ast}\Psi|^2 + |\pr_{t^\ast}\Psi|^2 + k^2\ellmode(\ellmode+1)|\Psi|^2\ri)\,\d t^\ast \\
    & \les_{M,k} e^{C\Ka}\EE[\Psi](t^\ast_2;t^\ast_1).
\end{align*}
Using that $2\frac{\De_-\De_+}{\De_0^2}+\quar\le(\frac{\De_--\De_+}{\De_0}\ri)^2\gtrsim_{M,k}1$ and that $\Ka \gtrsim_{M,k} 1$ and $e^f\geq 1$, we have
\begin{align}\label{est:proofCarlemanfinalestimatebis}
  \begin{aligned}
    & \int_{t^\ast_1}^{t^\ast_2}\int_{-\infty}^{\frac{\pi}{2}}\le(w|\pr_{t^\ast}\Psi|^2 + w |\pr_{r^\ast}\Psi|^2 + w\ellmode^2|\Psi|^2\ri)\,\d t^\ast\d r^\ast \\
    & + \lim_{r\to+\infty}\int_{t^\ast_1}^{t^\ast_2}\le(|\pr_{t^\ast}\Psi|^2 + |\pr_{r^\ast}\Psi|^2 + \ellmode^2|\Psi|^2\ri)\,\d t^\ast \\
  \les_{M,k} & \; e^{C\Ka} \EE[\Psi](t^\ast_2;t^\ast_1).
  \end{aligned}
\end{align}
Using the redshift estimate from Lemma~\ref{lem:redshift}, recalling the energy norm definitions~\eqref{eq:defenergynormsRWhigh}, estimate~\eqref{est:proofCarlemanfinalestimatebis} upgrades as
\begin{align*}
  \begin{aligned}
    \int_{t^\ast_1}^{t^\ast_2}\overline{\mathrm{E}}[\Psi](t^\ast)\,\d t^\ast + \overline{\mathrm{E}}^\II[\Psi](t^\ast_2;t^\ast_1) & \les_{M,k} \overline{\mathrm{E}}[\Psi](t^\ast_1) + e^{C\Ka}\EE[\Psi](t^\ast_2;t^\ast_1),
  \end{aligned}
\end{align*}
which, using the energy boundedness~\eqref{est:energyPsiDRmain} for solutions to the  Regge-Wheeler problem~\eqref{sys:RWdeco} and recalling~\eqref{est:Kabound}, finishes the proof of Proposition~\ref{prop:integralestimatesv} in the $\mathfrak{p}=2$ case.

\subsection{Carleman estimates with time-frequency decompositions}\label{sec:Fourier}
In this subsection, we prove Proposition~\ref{prop:integralestimatesv} in the stronger $\mathfrak{p}=1$ case using frequency methods.

Note first that it suffices to prove the estimate (\ref{est:Carlemanmain}) for $t^\ast_2 \geq t^\ast_1+1$ and replacing $t_1^\ast$ on the left by $t_1^\ast+1$. Indeed, for $t_2^\ast \leq t_1^\ast +1$ the estimate (\ref{est:Carlemanmain}) is easily established as follows. For the first term term on the left hand side  (\ref{est:Carlemanmain}) is immediate from the boundedness statement of Proposition \ref{prop:energybounded}. For the second term, one may use the identity (\ref{est:pfCarleman1int1}) with a $p$ that is equal to $1$ at the conformal boundary and zero near the horizon and again control all terms except $|\partial_{r^\ast} \Psi|^2$ and $|\partial_t \Psi|^2$ on the conformal boundary by the Proposition \ref{prop:energybounded}. This argument in fact only uses $\ell$ polynomially. 
 
To prove the restricted estimate, let $\chi$ be a smooth cut-off function such that $\chi|_{[1,+\infty)}=1$ and $\chi|_{(-\infty,0]}=0$. Define
\begin{align*}
\Psic & := \chi(t^\ast-t^\ast_1)\Psi.
\end{align*}
Note that $\Psic=\Psi$ for $t^\ast \geq t_1^\ast+1$. 
By the result of the previous section (Proposition~\ref{prop:integralestimatesv} in the $\mathfrak{p}=2$ case) and the above definition of $\Psic$, we have that (for each angular frequency) $\Psic$ is $L^2$ in time. 

Let $\eta$ be a smooth cut-off function such that $\eta|_{[-1,1]}=1$ and $\eta|_{(-\infty,-2]\cup[+2,+\infty)}=0$ and $\de>0$. One can define the following Fourier decomposition in time
\begin{align*}
  \widetilde{\Psi} & = \mathcal{F}^{-1}\le(\eta\le(\frac{\xi}{\de\ellmode}\ri)\mathcal{F}(\widetilde{\Psi})(\xi)\ri) + \mathcal{F}^{-1}\le((1-\eta)\le(\frac{\xi}{\de\ellmode}\ri)\mathcal{F}(\widetilde{\Psi})(\xi)\ri) =: \Psic^\flat +\Psic^\sharp.
\end{align*}
By Plancherel we have
\begin{align}\label{est:Plancherelsharp}
  \begin{aligned}
    \int_{-\infty}^{+\infty} |\pr_{t^\ast}\Psic^\sharp|^2 \,\d t^\ast = \int_{-\infty}^{+\infty} |\xi|^2|\mathcal{F}(\Psic^\sharp)|^2 \,\d \xi \gtrsim (\de\ellmode)^2 \int_{-\infty}^{+\infty} |\mathcal{F}(\Psic^\sharp)|^2 \,\d \xi = (\de\ellmode)^2 \int_{-\infty}^{+\infty} |\Psic^\sharp|^2 \,\d t^\ast, 
  \end{aligned}
\end{align}
and
\begin{align}\label{est:Plancherelflat}
  \begin{aligned}
    \int_{-\infty}^{+\infty} |\pr_{t^\ast}\Psic^\flat|^2 \,\d t^\ast \les (\de\ellmode)^2 \int_{-\infty}^{+\infty} |\Psic^\flat|^2 \,\d t^\ast  ,
  \end{aligned}
\end{align}
as well as 
\begin{align}\label{est:PlancherelErrorsourceterm}
  \int_{-\infty}^{+\infty}|\RW\Psic^\flat|^2\,\d t^\ast  + \int_{-\infty}^{+\infty}|\RW\Psic^\sharp|^2\,\d t^\ast \les \int_{-\infty}^{+\infty}|\RW\widetilde{\Psi}|^2\,\d t^\ast = \int_{t^\ast_1}^{t^\ast_1+1} |\RW\widetilde{\Psi}|^2 \,\d t^\ast .
\end{align}
Now from the definition of $\Psic$ and the fact that $\Psi$ satisfies $\RW \Psi=0$, we deduce
\begin{align}\label{est:PlancherelErrorsourcetermbis}
  \int_{-\infty}^{\frac{\pi}{2}}\int_{-\infty}^{+\infty} w^{-1}|\RW\widetilde{\Psi}|^2 \,\d t^\ast\d r^\ast = \int_{-\infty}^{\frac{\pi}{2}}\int_{t_1^\ast}^{t_1^\ast+1} w^{-1}|\RW\widetilde{\Psi}|^2 \,\d t^\ast\d r^\ast  & \les_{M,k} \mathrm{E}^{\mathfrak{R}}[\Psi](t^\ast_1),
\end{align}
with the last step following from  the energy and redshift estimates of Propositions~\ref{prop:energybounded} and~\ref{prop:redshift}.

Arguing along similar lines, we have
\begin{align}\label{est:Plancherelhorizon}
  \begin{aligned}
    \widetilde{\mathrm{E}}^\HH[\Psic^\flat](+\infty,-\infty) + \widetilde{\mathrm{E}}^\HH[\Psic^\sharp](+\infty,-\infty) & \les_{M,k} \lim_{r\to r_+}\int_{-\infty}^{+\infty}|\pr_{t^\ast}\widetilde{\Psi}|^2\,\d t^\ast \\
    & \les_{M,k} \overline{\mathrm{E}}^\HH[\Psi](+\infty,0) \les_{M,k} \mathrm{E}^{\mathfrak{R}}[\Psi](t_1^\ast),
    \end{aligned}
\end{align}
and depending on the fact that $\Psi$ satisfies the Dirichlet boundary condition~\eqref{eq:RWBCdecoDirichlet} or the Robin boundary condition~\eqref{eq:RWBCdecoRobin}, we have respectively
\begin{subequations}\label{eq:RWdecoFourier}
  \begin{align}
    \Psic^\flat,\Psic^\sharp & \xrightarrow{r\to+\infty} 0,\\
    \lim_{r\to+\infty}\int_{-\infty}^{+\infty}\le( \big|\mathfrak{r}\Psic^\flat\big|^2 +\big|\mathfrak{r}\Psic^\sharp\big|^2\ri)\,\d t^\ast & \les \lim_{r\to+\infty}\int_{t_1^\ast}^{t_1^\ast+1}\le(|\pr_t\Psi|^2+|\Psi|^2\ri)\,\d t^\ast \les \Einfty[\Psi](t_1^\ast) \les \mathrm{E}^{\mathfrak{R}}[\Psi](t_1^\ast),
  \end{align}
\end{subequations}
where we have used the notation 
\begin{align*}
  \mathfrak{r}\Psi & := 2\pr_{t}^2\Psi + \frac{\LL(\LL-2)}{6M}\pr_{r^\ast}\Psi + k^2\LL\Psi.
\end{align*}

\subsubsection{Estimates for $\Psic^\flat$}\label{sec:Psiflat}
Applying~\eqref{eq:ellipticFourier} with $f=1$ and $\Phi=\Psic^\flat$, we get
\begin{align}\label{eq:ellipticFourierrewrite}
  \begin{aligned}
    0 & = \int_{t^\ast_1}^{t^\ast_2}\int_{-\infty}^{\frac{\pi}{2}} \bigg(|\pr_{r^\ast}\Psic^\flat|^2 + V|\Psic^\flat|^2- \frac{\De_-\De_+}{\De_0^2}|\pr_{t^\ast}\Psic^\flat|^2 - \frac{\De_--\De_+}{\De_0}\Re\le(\pr_{t^\ast}\Psi^{\flat,\ast}\pr_{r^\ast}\Psic^\flat\ri) \bigg) \,\d t^\ast\d r^\ast  \\
      & \quad -\lim_{r\to+\infty}\int_{t^\ast_1}^{t^\ast_2}\Re\le(\Psi^{\flat,\ast}\pr_{r^\ast}\Psic^\flat\ri) \,\d t^\ast +\lim_{r\to r_+}\int_{t^\ast_1}^{t^\ast_2}\Re\le(\Psi^{\flat,\ast}\pr_{r^\ast}\Psic^\flat\ri)\,\d t^\ast \\
      & \quad +\le[\int_{-\infty}^{\frac{\pi}{2}}\le(\frac{\De_-\De_+}{\De_0^2}\Re\le(\Psi^{\flat,\ast}\pr_{t^\ast}\Psic^\flat\ri)+ \half \pr_{r^\ast}\le(\frac{\De_-}{\De_0}\ri)|\Psic^\flat|^2 + \frac{\De_--\De_+}{\De_0}\Re\le(\Psi^{\flat,\ast}\pr_{r^\ast}\Psic^\flat\ri)\ri)\,\d r^\ast\ri]_{t^\ast_1}^{t^\ast_2} \\
      & \quad +\int_{t^\ast_1}^{t^\ast_2}\int_{-\infty}^{\frac{\pi}{2}}\Re\le(\Psi^{\flat,\ast}\RW\Psic^\flat\ri)\,\d t^\ast\d r^\ast.
  \end{aligned}
\end{align}
Using that $\Psi$, and therefore $\Psic^\flat$, is regular at the horizon~\eqref{eq:defreghorgeneral}, we have
\begin{align}\label{eq:boundaryellipticFourier}
  \begin{aligned}
    & -\lim_{r\to+\infty}\int_{t^\ast_1}^{t^\ast_2}\Re\le(\Psi^{\flat,\ast}\pr_{r^\ast}\Psic^\flat\ri) \,\d t^\ast + \lim_{r\to r_+}\int_{t^\ast_1}^{t^\ast_2}\Re\le(\Psi^{\flat,\ast}\pr_{r^\ast}\Psic^\flat\ri)\,\d t^\ast \\
    = & \; \lim_{r\to+\infty}\frac{6M}{\ellmode(\ellmode+1)\le(\ellmode(\ellmode+1)-2\ri)}\int_{t^\ast_1}^{t^\ast_2}\le(-2|\pr_t\Psic^\flat|^2+k^2\ellmode(\ellmode+1)|\Psic^\flat|^2\ri)\,\d t^\ast\\
    & \quad + \lim_{r\to+\infty}\frac{12M}{\ellmode(\ellmode+1)\le(\ellmode(\ellmode+1)-2\ri)}\le[\Re\le(\Psic^\flat\pr_t\Psic^\flat\ri)\ri]_{t^\ast_1}^{t^\ast_2}\\
    & \quad + \lim_{r\to+\infty}\frac{6M}{\ellmode(\ellmode+1)\le(\ellmode(\ellmode+1)-2\ri)}\int_{t^\ast_1}^{t^\ast_2}\Re\le(\Psi^{\flat,\ast}\mathfrak{r}\Psic^\flat\ri)\,\d t^\ast.
  \end{aligned}
\end{align}
Plugging~\eqref{eq:boundaryellipticFourier} in~\eqref{eq:ellipticFourier} and letting $t^\ast_2\to+\infty$ and $t^\ast_1\to-\infty$, using that $\Psic^\flat\to0$ when $t^\ast\to\pm\infty$, we get 
\begin{align}\label{eq:ellipticFourierrewritebis}
  \begin{aligned}
    0 & = \int_{-\infty}^{+\infty}\int_{-\infty}^{\frac{\pi}{2}} \bigg(|\pr_{r^\ast}\Psic^\flat|^2 + V|\Psic^\flat|^2- \frac{\De_-\De_+}{\De_0^2}|\pr_{t^\ast}\Psic^\flat|^2 - \frac{\De_--\De_+}{\De_0}\Re\le(\pr_{t^\ast}\Psi^{\flat,\ast}\pr_{r^\ast}\Psic^\flat\ri) \bigg) \,\d t^\ast\d r^\ast  \\
      & \quad +\lim_{r\to+\infty}\frac{6M}{\ellmode(\ellmode+1)\le(\ellmode(\ellmode+1)-2\ri)}\int_{-\infty}^{+\infty}\le(-2|\pr_t\Psic^\flat|^2+k^2\ellmode(\ellmode+1)|\Psic^\flat|^2\ri)\,\d t^\ast\\
      & \quad +\int_{-\infty}^{+\infty}\int_{-\infty}^{\frac{\pi}{2}}\Re\le(\Psi^{\flat,\ast}\RW\Psic^\flat\ri)\,\d t^\ast\d r^\ast + \lim_{r\to+\infty}\frac{6M}{\ellmode(\ellmode+1)\le(\ellmode(\ellmode+1)-2\ri)}\int_{-\infty}^{+\infty}\Re\le(\Psi^{\flat,\ast}\mathfrak{r}\Psic^\flat\ri)\,\d t^\ast.
  \end{aligned}
\end{align}

\begin{lemma}\label{lem:Plancherelelliptic}
  There exists $\de_0=\de_0(M,k)>0$ and $\ellmode_0=\ellmode_0(M,k)>0$ such that for all $\de\leq \de_0$ and $\ellmode\geq\ellmode_0$,
  \begin{align}\label{est:ellipticPlancherelspacetime}
    \begin{aligned}
    & \int_{-\infty}^{+\infty}\int_{-\infty}^{\frac{\pi}{2}}\le(w|\pr_{t^\ast}\Psic^\flat| + w \ellmode^2 |\Psic^\flat|^2\ri)\,\d t^\ast\d r^\ast \\
      \les_{M,k} & \; \int_{-\infty}^{+\infty}\int_{-\infty}^{\frac{\pi}{2}} \bigg(V|\Psic^\flat|^2- \frac{\De_-\De_+}{\De_0^2}|\pr_{t^\ast}\Psic^\flat|^2 - \frac{\De_--\De_+}{\De_0}\Re\le(\pr_{t^\ast}\Psi^{\flat,\ast}\pr_{r^\ast}\Psic^\flat\ri) \bigg) \,\d t^\ast\d r^\ast \\
      & \quad + \de \int_{-\infty}^{+\infty}\int_{-\infty}^{\frac{\pi}{2}}w^{-1}|\pr_{r^\ast}\Psic^\flat|^2\,\d t^\ast\d r^\ast.
    \end{aligned}
  \end{align}
\end{lemma}
\begin{proof}
Let us show the estimate in the form where we take the terms with minus signs on the right to the left hand side.  
For $\ellmode$ sufficiently large depending only on $M,k$, we have
  \begin{align}\label{est:ell2byV}
    w\ellmode^2 \leq V ,
  \end{align}
  which deals with the last term on the left hand side. 
For the other terms we notice  
  \begin{align}\label{est:timederivprePlancherel}
    \begin{aligned}
      & \bigg|\int_{-\infty}^{+\infty}\int_{-\infty}^{\frac{\pi}{2}}\bigg(\frac{\De_-\De_+}{\De_0^2}|\pr_{t^\ast}\Psic^\flat|^2 - \frac{\De_--\De_+}{\De_0}\Re\le(\pr_{t^\ast}\Psi^{\flat,\ast}\pr_{r^\ast}\Psic^\flat\ri) \bigg) \,\d t^\ast\d r^\ast\bigg| \\
      \les_{M,k} & \; \int_{-\infty}^{+\infty}\int_{-\infty}^{\frac{\pi}{2}} \le( w\le(1 + \de^{-1}\ri) |\pr_{t^\ast}\Psic^\flat|^2 + \de w^{-1}|\pr_{r^\ast}\Psic^\flat|^2 \ri)\,\d t^\ast\d r^\ast.
    \end{aligned}
  \end{align}
  Using the Plancherel estimate~\eqref{est:Plancherelflat}, we have
  \begin{align}\label{est:Plancherelinellipticproof}
    \begin{aligned}
      \int_{-\infty}^{+\infty}\int_{-\infty}^{\frac{\pi}{2}} w\le(1 + \de^{-1}\ri) |\pr_{t^\ast}\Psic^\flat|^2 \,\d t^\ast\d r^\ast \les_{M,k} \de^2\le(1 + \de^{-1}\ri)\int_{-\infty}^{+\infty}\int_{-\infty}^{\frac{\pi}{2}} w\ellmode^2 |\Psic^\flat|^2\,\d t^\ast\d r^\ast,
    \end{aligned}
  \end{align}
  Combining~\eqref{est:ell2byV},~\eqref{est:timederivprePlancherel} and~\eqref{est:Plancherelinellipticproof}, we obtain~\eqref{est:ellipticPlancherelspacetime} and this finishes the proof of the lemma.
\end{proof}
Along the same lines as in Lemma~\ref{lem:Plancherelelliptic}, we have the following consequence of Plancherel estimate at the boundary. The proof is left to the reader.
\begin{lemma}\label{lem:Plancherelellipticboundary}
  There exists $\de_1=\de_1(M,k)>0$ such that for all $\de\leq \de_1$,
  \begin{align}\label{est:Plancherelellipticboundary}
    \begin{aligned}
    & \lim_{r\to+\infty}\frac{6M}{\ellmode(\ellmode+1)\le(\ellmode(\ellmode+1)-2\ri)}\int_{-\infty}^{+\infty}\le(|\pr_t\Psic^\flat|^2+k^2\ellmode(\ellmode+1)|\Psic^\flat|^2\ri)\,\d t^\ast\\
      \les_{M,k} & \; \lim_{r\to+\infty}\frac{6M}{\ellmode(\ellmode+1)\le(\ellmode(\ellmode+1)-2\ri)}\int_{-\infty}^{+\infty}\le(-2|\pr_t\Psic^\flat|^2+k^2\ellmode(\ellmode+1)|\Psic^\flat|^2\ri)\,\d t^\ast.
    \end{aligned}
  \end{align}
\end{lemma}

Now, let $\de\leq \max(\de_0,\de_1)$ and $\ellmode\geq\ellmode_0$. Plugging~\eqref{est:ellipticPlancherelspacetime} and~\eqref{est:Plancherelellipticboundary} in \eqref{eq:ellipticFourierrewritebis}, we get
\begin{align}\label{est:ellipticFourierrewritebisbis}
  \begin{aligned}
    & \int_{-\infty}^{+\infty}\int_{-\infty}^{\frac{\pi}{2}} \bigg(w|\pr_{t^\ast}\Psic^\flat|^2 + |\pr_{r^\ast}\Psic^\flat|^2 + w\ellmode^2|\Psic^\flat|^2\bigg) \,\d t^\ast\d r^\ast  \\
    & +\lim_{r\to+\infty}\frac{6M}{\ellmode(\ellmode+1)\le(\ellmode(\ellmode+1)-2\ri)}\int_{-\infty}^{+\infty}\le(|\pr_t\Psic^\flat|^2+k^2\ellmode(\ellmode+1)|\Psic^\flat|^2\ri)\,\d t^\ast\\
    \les_{M,k} & \; \int_{-\infty}^{+\infty}\int_{-\infty}^{\frac{\pi}{2}}|\Psi^{\flat,\ast}||\RW\Psic^\flat|\,\d t^\ast\d r^\ast + \lim_{r\to+\infty}\frac{6M}{\ellmode(\ellmode+1)\le(\ellmode(\ellmode+1)-2\ri)}\int_{-\infty}^{+\infty}|\Psi^{\flat,\ast}||\mathfrak{r}\Psic^\flat|\,\d t^\ast \\
    & + \de \int_{-\infty}^{+\infty}\int_{-\infty}^{\frac{\pi}{2}}w^{-1}|\pr_{r^\ast}\Psic^\flat|^2 \,\d t^\ast\d r^\ast.
  \end{aligned}
\end{align}
Estimate~\eqref{est:ellipticFourierrewritebisbis} (where we use an absorption argument for the inhomogeneous terms) together with the redshift estimate from Lemma~\ref{lem:redshift} with $t^\ast_1\to-\infty, t^\ast_2\to+\infty$, yields
\begin{align}\label{est:ellipticFourierrewritebisbisbis}
  \begin{aligned}
    & \int_{-\infty}^{+\infty}\overline{\mathrm{E}}[\Psic^\flat](t^\ast)\,\d t^\ast +\lim_{r\to+\infty}\frac{6M}{\ellmode(\ellmode+1)\le(\ellmode(\ellmode+1)-2\ri)}\int_{-\infty}^{+\infty}\le(|\pr_t\Psic^\flat|^2+k^2\ellmode(\ellmode+1)|\Psic^\flat|^2\ri)\,\d t^\ast\\
    \les_{M,k} & \; \int_{-\infty}^{+\infty}\int_{-\infty}^{\frac{\pi}{2}}w^{-1}|\RW\Psic^\flat|^2\,\d t^\ast\d r^\ast  + \lim_{r\to+\infty}\frac{6M}{\ellmode(\ellmode+1)\le(\ellmode(\ellmode+1)-2\ri)}\int_{-\infty}^{+\infty}|\mathfrak{r}\Psic^\flat|^2\,\d t^\ast \\
    & +\widetilde{\mathrm{E}}^\HH[\Psic^\flat](+\infty,-\infty) + \de \int_{-\infty}^{+\infty} \overline{\mathrm{E}}[\Psic^\flat](t^\ast)\,\d t^\ast, 
  \end{aligned}
\end{align}
where we recall the norm definitions~\eqref{eq:defenergynormsRWhigh}. In the Robin case, using Plancherel and the definition of $\Psic^\flat$, we have
\begin{align}\label{est:DirichletorRobintoNeumannellipticFourierRobin}
  \begin{aligned}
    \lim_{r\to+\infty} \int_{-\infty}^{+\infty}|\pr_{r^\ast}\Psic^\flat|^2\,\d t^\ast & \les_{M,k} \lim_{r\to+\infty} \ellmode^{-8} \int_{-\infty}^{+\infty}\le(|\pr_{t}^2\Psic^\flat|^2 + \ellmode^4|\Psic^\flat|^2\ri)\,\d t^\ast  + \lim_{r\to+\infty} \ellmode^{-8} \int_{-\infty}^{+\infty} |\mathfrak{r}\Psic^\flat|^2\,\d t^\ast\\
                                                                                     & \les_{M,k}\lim_{r\to+\infty} \ellmode^{-4} \int_{-\infty}^{+\infty}|\Psic^\flat|^2\,\d t^\ast + \mathrm{E}^{\mathfrak{R}}[\Psi](1).
  \end{aligned}
\end{align}
In the Dirichlet case, using the fundamental theorem of calculus and Plancherel, we have
\begin{align}\label{est:DirichletorRobintoNeumannellipticFourierDirichlet}
  \begin{aligned}
    \lim_{r\to+\infty} \int_{-\infty}^{+\infty}|\pr_{r^\ast}\Psic^\flat|^2\,\d t^\ast & \les_{M,k} \int_{-\infty}^{+\infty}\int_{r^\ast_{3M}}^{\frac{\pi}{2}}|\pr_{r^\ast}\Psic^\flat||\pr_{r^\ast}^2\Psic^\flat|\,\d t^\ast\d r^\ast\\
                                                                                     & \les_{M,k} \int_{-\infty}^{+\infty}\int_{r^\ast_{3M}}^{\frac{\pi}{2}}|\pr_{r^\ast}\Psic^\flat|\le(|\pr_{t^\ast}^2\Psic^\flat| + |\pr_{t^\ast}\pr_{r^\ast}\Psic^\flat| + \ellmode^2|\Psic^\flat| + |\RW\Psic^\flat|\ri)\,\d t^\ast\d r^\ast\\
                                                                                     & \les_{M,k} \ellmode\int_{-\infty}^{+\infty}\overline{\mathrm{E}}[\Psic^\flat](t^\ast)\,\d t^\ast.
  \end{aligned}
\end{align}
Using~\eqref{est:PlancherelErrorsourceterm},~\eqref{est:PlancherelErrorsourcetermbis},~\eqref{est:Plancherelhorizon},~\eqref{eq:RWdecoFourier},~\eqref{est:DirichletorRobintoNeumannellipticFourierRobin} or~\eqref{est:DirichletorRobintoNeumannellipticFourierDirichlet} and an absorption argument with $\de$ sufficiently small depending only on $M,k$ in~\eqref{est:ellipticFourierrewritebisbisbis}, yields\footnote{The $\ellmode$ dependency on the RHS comes from the Dirichlet to Neumann estimate~\eqref{est:DirichletorRobintoNeumannellipticFourierDirichlet}. It will be absorbed by the $e^{C\ellmode}$ dependency of the high-frequency estimates of the next section.}
\begin{align}\label{est:ellipticFourierfinal}
  \begin{aligned}
    \int_{-\infty}^{+\infty}\overline{\mathrm{E}}[\Psic^\flat](t^\ast)\,\d t^\ast + \overline{\mathrm{E}}^\II[\Psic^\flat](+\infty,-\infty) & \les_{M,k} \ellmode \cdot \mathrm{E}^{\mathfrak{R}}[\Psi](1).
  \end{aligned}
\end{align}

\subsubsection{Estimates for $\Psic^\sharp$}\label{sec:Psisharp}
The proof of the Carleman estimates follow the same scheme as in Section~\ref{sec:Carlemanestimatesfirst}. Take $p=e^f$ with $f=\frac{\Ka}{r}$, where $\Ka>0$ is a constant which will be determined. Setting $p=e^f$ with $f=\frac{\Ka}{r}$ in the integral estimate~\eqref{est:easytouseCarleman} with $t^\ast_1\to-\infty, t^\ast_2\to+\infty$, using that $p'=-\Ka w e^f$, that $\le(\sup |p|\ri)+\le(\sup w^{-1}|p'|\ri) \leq (1+\Ka)e^{\frac{\Ka}{r_+}}=e^{C\Ka}$ with $C=C(M,k)>0$, and the estimates~\eqref{est:PlancherelErrorsourcetermbis},~\eqref{est:Plancherelhorizon}, we get
\begin{align}\label{est:CarlemanPlancherel}
  \begin{aligned}
    &\int_{-\infty}^{+\infty}\int_{-\infty}^{\frac{\pi}{2}} \le(\Ka w e^f \le(2\frac{\De_-\De_+}{\De_0^2}+\quar\le(\frac{\De_--\De_+}{\De_0}\ri)^2\ri)|\pr_{t^\ast}\Psic^\sharp|^2 + \Ka w e^f|\pr_{r^\ast}\Psic^\sharp|^2\ri)\,\d t^\ast\d r^\ast \\
    & +\int_{-\infty}^{+\infty}\int_{-\infty}^{\frac{\pi}{2}}\le(2 (pV)' + \half\le(-pV' - \half p'''\ri)\ri)|\Psic^\sharp|^2\,\d t^\ast\d r^\ast\\
    & + \half \lim_{r\to+\infty}\int_{-\infty}^{+\infty}\le(|\pr_{r^\ast}\Psic^\sharp|^2 + |\pr_{t^\ast}\Psic^\sharp|^2 + k^2\ellmode(\ellmode+1)|\Psic^\sharp|^2\ri)\,\d t^\ast \\
    & \les_{M,k} e^{C\Ka}\bigg(\mathrm{E}^{\mathfrak{R}}[\Psi](1) + \le(\mathrm{E}^{\mathfrak{R}}[\Psi](1)\ri)^{1/2}\le(\int_{-\infty}^{+\infty}\widetilde{\mathrm{E}}[\Psic^\sharp](t^\ast)\,\d t^\ast\ri)^{1/2}\bigg),
  \end{aligned}
\end{align}
where the boundary terms were already controlled by~\eqref{est:bdyineasyCarleman}, provided that $\Ka\geq 8k^{-1}\sqrt{\ellmode(\ellmode+1)}$. The difference is that one can rely on Plancherel estimates to bound the potential terms by the time-derivative terms. This gives the following lemma.
\begin{lemma}\label{lem:PlancherelsharpCarlemanised}
  For all $\de>0$, there exists $r_{\mathrm{Pla}}=r_{\mathrm{Pla}}(M,k,\de)>r_+$ and $\ellmode_{\mathrm{Pla}} = \ellmode_{\mathrm{Pla}}(M,k,\de) \geq 2$, such that, for all $\ellmode\geq \ellmode_{\mathrm{Pla}}$ and for all $\Ka \geq \ellmode$, setting $p=e^f=e^{\frac{\Ka}{r}}$, we have
  \begin{align}\label{est:PlancherelsharpCarlemanised}
    \begin{aligned}
      \int_{-\infty}^{+\infty} \ellmode^2 w e^f |\Psic^{\sharp}|^2 \,\d t^\ast & \les_{M,k,\de} \half \int_{-\infty}^{+\infty}\Ka w e^{\frac{\Ka}{r}}\le(\frac{\De_-\De_+}{\De_0^2} + \quar\le(\frac{\De_--\De_+}{\De_0}\ri)^2\ri) |\pr_{t^\ast}\Psic^{\sharp}|^2 \,\d t^\ast \\
      & \quad\quad\quad\quad+ \int_{-\infty}^{+\infty} \le(2(pV)' -\half pV'-\quar p'''\ri)|\Psic^{\sharp}|^2 \,\d t^\ast,
    \end{aligned}
  \end{align}
  for all $r_+<r<r_{\mathrm{Pla}}$.
\end{lemma}
\begin{proof}
  From the signs in the expression~\eqref{eq:easypotential}, we have the following general estimate
  \begin{align}\label{est:easyboundsontheerrorpotentialsforFourier}
    \begin{aligned}
      -2(pV)'+\half pV' +\quar p''' & \les_{M,k} \ellmode^2\le(1+\Ka w\ri)we^f, 
    \end{aligned}
  \end{align}
  for all $r>r_+$. Using Plancherel~\eqref{est:Plancherelsharp}, we have
  \begin{align*}
    \begin{aligned}
      & \int_{-\infty}^{+\infty}\Ka w e^{f}\le(\frac{\De_-\De_+}{\De_0^2} + \quar\le(\frac{\De_--\De_+}{\De_0}\ri)^2\ri) |\pr_{t^\ast}\Psic^{\sharp}|^2 \,\d t^\ast \\
      & \gtrsim_{M,k} \int_{-\infty}^{+\infty} \Ka w e^{f} |\pr_{t^\ast}\Psic^{\sharp}|^2 \,\d t^\ast \\
      & \geq C\de^2\int_{-\infty}^{+\infty}\Ka \ellmode^2 w e^{f} |\Psic^{\sharp}|^2 \,\d t^\ast.
    \end{aligned}
  \end{align*}
  where $C=C(M,k)>0$. Hence, using~\eqref{est:easyboundsontheerrorpotentialsforFourier}, we get
  \begin{align*}
    \begin{aligned}
      & \half \int_{-\infty}^{+\infty}\Ka w e^{\frac{\Ka}{r}}\le(\frac{\De_-\De_+}{\De_0^2} + \quar\le(\frac{\De_--\De_+}{\De_0}\ri)^2\ri) |\pr_{t^\ast}\Psic^{\sharp}|^2 \,\d t^\ast \\
      & + \int_{-\infty}^{+\infty} \le(2(pV)' -\half pV'-\quar p'''\ri)|\Psic^{\sharp}|^2 \,\d t^\ast \\
      \geq & \; \de^2\int_{-\infty}^{+\infty} \le(\half C\Ka - \de^{-2}(1+\Ka w)\ri) \ellmode^2we^f |\Psic^{\sharp}|^2 \,\d t^\ast.      
    \end{aligned}
  \end{align*}
  Taking $r_{\mathrm{Pla}}>r_+$ and $\ellmode_{\mathrm{Pla}}\geq 2$ such that
  \begin{align*}
    \de^{-2}w(r_{\mathrm{Pla}}) & \leq \frac{C}{8}, & \de^{-2} \leq \frac{1}{8} C\ellmode_{\mathrm{Pla}},
  \end{align*}
  using that, by assumption, $\Ka\geq\ellmode$, and that $w$ is increasing, this concludes the proof of the lemma.  
\end{proof}

\begin{lemma}\label{lem:goodpotFourier}
  Let $r_{\mathrm{Pla}}>r_+$. There exists $\varep_{\mathrm{Pla}}=\varep_{\mathrm{Pla}}(M,k,r_{\mathrm{Pla}})>0$ such that, for all $\Ka \geq \varep^{-1}_{\mathrm{Pla}} \ellmode$, we have
  \begin{align}\label{est:goodpotFourier}
    2 (pV)' - \half pV'-\quar p''' & \gtrsim_{M,k,r_{\mathrm{Pla}}} \ellmode^2 w e^f,
  \end{align}
  for all $r>r_{\mathrm{Pla}}$.
\end{lemma}
\begin{proof}
  From~\eqref{eq:easypotential}, we have
  \begin{align*}
    \le|2 (pV)' - \half pV'-\quar p''' -\frac{1}{8}\Ka^3we^f\ri| & \les_{M,k,r_{\mathrm{Pla}}} \le(\Ka^2 + \Ka \ellmode^2\ri) we^f.
  \end{align*}
  Hence, provided that $\Ka^{-1}\ellmode$ is sufficiently small depending only on $M,k,r_{\mathrm{Pla}}$, we have
  \begin{align*}
    2 (pV)' - \half pV'-\quar p''' > \frac{1}{16}\Ka^3we^f \gtrsim_{M,k} \ellmode^2 w e^f,
  \end{align*}
  which is the desired~\eqref{est:goodpotFourier}.
\end{proof}

Now, taking $\Ka$ such that
\begin{align*}
  \Ka \geq \max\le(8k^{-1}\sqrt{\ellmode(\ellmode+1)},\ellmode,\varep_{\mathrm{Pla}}^{-1}\ellmode\ri),
\end{align*}
and plugging~\eqref{est:PlancherelsharpCarlemanised} and~\eqref{est:goodpotFourier} in~\eqref{est:CarlemanPlancherel}, we get
\begin{align}\label{est:CarlemanPlancherelfinal}
  \begin{aligned}
    & \int_{-\infty}^{+\infty}\int_{-\infty}^{\frac{\pi}{2}} w e^f \le(|\pr_{t^\ast}\Psic^\sharp|^2 + |\pr_{r^\ast}\Psic^\sharp|^2 +\ellmode^2|\Psic^\sharp|^2\ri)\,\d t^\ast\d r^\ast \\
    & + \lim_{r\to+\infty}\int_{-\infty}^{+\infty}\le(|\pr_{r^\ast}\Psic^\sharp|^2 + |\pr_{t^\ast}\Psic^\sharp|^2 + \ellmode^2 |\Psic^\sharp|^2 \ri) \,\d t^\ast \\
    & \les_{M,k} e^{C\Ka}\bigg(\mathrm{E}^{\mathfrak{R}}[\Psi](1) + \le(\mathrm{E}^{\mathfrak{R}}[\Psi](1)\ri)^{1/2}\le(\int_{-\infty}^{+\infty}\widetilde{\mathrm{E}}[\Psic^\sharp](t^\ast)\,\d t^\ast\ri)^{1/2}\bigg),
  \end{aligned}
\end{align}
from which, using the redshift estimate of Lemma~\ref{lem:redshift}, we get
\begin{align}\label{est:CarlemanPlancherelfinalbis}
  \begin{aligned}
    \int_{-\infty}^{+\infty}\overline{\mathrm{E}}[\Psic^\sharp](t^\ast)\,\d t^\ast + \overline{\mathrm{E}}^\II[\Psic^\sharp](+\infty,-\infty) & \leq e^{C\ellmode}\mathrm{E}^{\mathfrak{R}}[\Psi](t^\ast_1),
  \end{aligned}
\end{align}
where the last integral term of~\eqref{est:CarlemanPlancherelfinal} was absorbed on the LHS of~\eqref{est:CarlemanPlancherelfinalbis}.

\subsubsection{Conclusion}
Combining~\eqref{est:ellipticFourierfinal} and~\eqref{est:CarlemanPlancherelfinalbis}, we get for all $t^\ast_2 \geq t^\ast_1+1$
\begin{align}
  \begin{aligned}
    \int_{t^\ast_1+1}^{t^\ast_2}\overline{\mathrm{E}}[\Psi](t^\ast)\,\d t^\ast + \overline{\mathrm{E}}^\II[\Psi](t^\ast_2,t^\ast_1) \les_{M,k} \int_{-\infty}^{+\infty}\overline{\mathrm{E}}[\widetilde{\Psi}](t^\ast)\,\d t^\ast + \overline{\mathrm{E}}^\II[\widetilde{\Psi}](+\infty,-\infty) & \leq e^{C\ellmode}\mathrm{E}^{\mathfrak{R}}[\Psi](t^\ast_1),
  \end{aligned}
\end{align}
where $C=C(M,k)>0$. Recalling (\ref{merel}) and the fact that $\overline{\mathrm{E}}^\II[\Psi](t^\ast_2,t^\ast_1)$ controls $\int_{t_1^\ast}^{t_2^\ast} dt^\ast \Einfty_{m\ellmode}[\Psi](t^\ast)$, this finishes the proof of Proposition~\ref{prop:integralestimatesv} in the (strong) $\mathfrak{p}=1$ case.

\section{Estimates for the Teukolsky quantities}\label{sec:RWtoTeuk}
We now establish estimates on the original Teukolsky quantities from the estimates on the Regge-Wheeler quantities proven in the previous section. The main result is the following:

\begin{theorem}[Boundedness and integrated decay estimates for $\alt^{[\pm2]}$]\label{thm:Teukboundednessdecay}
  Let $(\al^{[\pm2]})$  be solutions of the Teukolsky problem on $\mathcal{M}$ as in Definition~\ref{def:sysTeuk}. Let $\Psi^{[\pm2]}$ be the Chandrasekhar transformations of $\alt^{[\pm2]}$ defined by~\eqref{eq:Chandra}. Let $n>2$, $\ellmode\geq 2$ and $|m|\leq\ellmode$. We have
  \begin{itemize}
  \item \underline{Boundedness of $\alt^{[\pm2]}$:}
    \begin{align}
      \label{est:boundednessalt}
      \begin{aligned}
        \mathrm{E}^{\mathfrak{T},\mathfrak{R},n}[\alt](t^\ast_2) & \les_{M,k,n} \mathrm{E}^{\mathfrak{T},\mathfrak{R},n}[\alt](t^\ast_1),
      \end{aligned}
    \end{align}
    for all $t^\ast_2\geq t^\ast_1$,
  \item \underline{Integrated decay of $\alt^{[\pm2]}$:}
    \begin{align}\label{est:Teukfullestimates}
      \begin{aligned}
        \mathrm{E}^{\mathfrak{T},\mathfrak{R},n}_{m\ellmode}[\alt](t^\ast_2) + \int_{t^\ast_1}^{t^\ast_2}\mathrm{E}^{\mathfrak{T},\mathfrak{R},n}_{m\ellmode}[\alt](t^\ast)\,\d t^\ast \leq e^{C\ellmode^\mathfrak{p}}\mathrm{E}^{\mathfrak{T},\mathfrak{R},n}_{m\ellmode}[\alt](t^\ast_1),
      \end{aligned}
    \end{align}
    for all $t^\ast_2\geq t^\ast_1$, where $\mathfrak{p}=1$ (see Proposition~\ref{prop:integralestimatesv}) and with $C=C(M,k,n)>0$.
  \end{itemize}
\end{theorem}

Theorem \ref{thm:Teukboundednessdecay} will be a consequence of the following key-proposition and the decay estimates for $\Psi^{D},\Psi^R$ obtained in already in Theorem~\ref{thm:mainRW1a} and Theorem~\ref{thm:mainRW1b}:

\begin{proposition}\label{prop:TeukfromRW}
  Let $(\al^{[\pm2]})$  be solutions of the Teukolsky problem on $\mathcal{M}$ as in Definition~\ref{def:sysTeuk}. Let $\Psi^{[\pm2]}$ be the Chandrasekhar transformations of $\alt^{[\pm2]}$ defined by~\eqref{eq:Chandra}. For all integers $n>2$ and all $t^\ast_2\geq t^\ast_1$, we have
  \begin{align}\label{est:TeukfromRW}
    \begin{aligned}
      & \mathrm{E}^{\mathfrak{T},n}[\alt](t^\ast_2) + \int_{t^\ast_1}^{t^\ast_2} \mathrm{E}^{\mathfrak{T},n}[\alt](t^\ast)[\alt](t^\ast)\,\d t^\ast + \overline{\mathrm{E}}^{\HH,n}[\alt^{[+2]}](t^\ast_2;t^\ast_1) + \overline{\mathrm{E}}^{\HH,n}[w^{-2}\alt^{[-2]}](t^\ast_2;t^\ast_1)   \\
      \les_{M,k,n} & \; \mathrm{E}^{\mathfrak{T},n}[\alt](t^\ast_1) + \overline{\mathrm{E}}^{\HH,n-2}[\Psi^{[\pm2]}](t^\ast_2;t^\ast_1) + \overline{\mathrm{E}}^{n-2}[\Psi^{[\pm2]}](t^\ast_1) + \overline{\mathrm{E}}^{n-2}[\Psi^{[\pm2]}](t^\ast_2) \\
      & \quad\quad+ \int_{t^\ast_1}^{t^\ast_2}\overline{\mathrm{E}}^{n-2}[\Psi^{[\pm2]}](t^\ast)\,\d t^{\ast}.
    \end{aligned}
  \end{align}
  \emph{For convenience, we have re-collected all the definitions of the energy norms of the paper in Appendix~\ref{sec:energynorms}.}
\end{proposition}

Section~\ref{sec:hierarchyRW} to Section~\ref{sec:estaltpm2proofpropTeukfromRW} are dedicated to the proof of Proposition~\ref{prop:TeukfromRW}. The proof of Theorem~\ref{thm:Teukboundednessdecay} then easily follows and is presented in Section \ref{sec:proofthmTeukbdddecay}. In Section~\ref{sec:proofexpodecay} we finally prove our main theorem, Theorem~\ref{thm:main1}, as a consequence of Theorem~\ref{thm:Teukboundednessdecay} and an application of a pigeonhole and an interpolation argument.

In Sections~\ref{sec:hierarchyRW} to~\ref{sec:estaltpm2proofpropTeukfromRW}, we shall lighten the notation and drop the $m\ellmode$ indices for $(\al^{[\pm2]})$ and $\Psi^{[\pm2]}$, it always being implicit that a fixed $m\ell$ mode is considered. The proof of Proposition~\ref{prop:TeukfromRW} is inspired (and simpler in some aspects) by~\cite[Chapter 11]{Daf.Hol.Rod.Tay21}. For the benefit of the reader, we give a brief outline of the proof:
\begin{itemize}
\item In Section~\ref{sec:hierarchyRW} we recall the hierarchy of (inhomogeneous) Regge-Wheeler equations satisfied by the quantities $\alt^{[\pm2]},\psi^{[\pm2]},\Psi^{[\pm2]}$ as well as their boundary conditions at the conformal boundary. 
\item In Section~\ref{sec:prelimtrans} we derive preliminary transport estimates for $\alt^{[\pm2]}$ and $\psi^{[\pm2]}$ which lose derivatives but which can be used in the subsequent sections to estimate the lower order terms in the inhomogeneous Regge-Wheeler equations satisfied by $\alt^{[\pm2]}$ and $\psi^{[\pm2]}$.
\item In Section~\ref{sec:energyestpsipm2} we derive energy and redshift estimates for $\psi^{[+2]}$ and $\psi^{[-2]}$ using the inhomogeneous Regge-Wheeler equations satisfied by $\psi^{[+2]}$ and $\psi^{[-2]}$, considering $\Psi^{[\pm2]}$ as source terms. This is similar to the energy and redshift estimates applied in Sections~\ref{sec:RWbound} and~\ref{sec:RWhigh}. These estimates provide bounds for the energies of $\psi^{[\pm2]}$ on each spacelike slice.
\item To obtain global spacetime integral estimates for $\psi^{[+2]}$ and $\psi^{[-2]}$ we first use elliptic estimates in Section~\ref{sec:ellestpsipm2}, with $\Psi^{[\pm2]}$ as source terms, to estimate the spacetime integrals of $\ellmode\psi^{[+2]}$ and $\ellmode w^{-1}\psi^{[-2]}$.
\item The spacetime integrals of $\ellmode\psi^{[+2]}$ and $\ellmode w^{-1}\psi^{[-2]}$ being controlled we can prove a global Morawetz estimate (with these terms as source terms) in Section~\ref{sec:Morawetzpsipm2}. This controls the spacetime integrals of $L\psi^{[+2]}$ and $\Lb\psi^{[-2]}$. Since $\Lb\psi^{[+2]}$ is directly controlled by the Chandrasekhar transformations~\eqref{eq:Chandra}, and also provides the desired weights at the horizon, this closes the estimates for $\psi^{[+2]}$.
\item The Morawetz estimates of Section~\ref{sec:Morawetzpsipm2} do not feature the desired weight at the horizon for $\psi^{[-2]}$ since our goal is to estimate $w^{-1}\psi^{[-2]}$ instead of $\psi^{[-2]}$. To close the estimates for $w^{-1}\psi^{[-2]}$ we thus need an additional redshift estimate for the wave equation satisfies by $w^{-1}\psi^{[-2]}$, which is carried out in Section~\ref{sec:redshiftpsim2}.
\item Arguing similarly for $\alt^{[+2]}$, $w^{-2}\alt^{[-2]}$ and commuting and combining these estimates, we can conclude the proof of Proposition~\ref{prop:TeukfromRW}, which is done in Section \ref{sec:fipo}.
\end{itemize}

\subsection{The hierarchy of Regge-Wheeler equations}\label{sec:hierarchyRW}
We have the following proposition, which follows by direct computation. 
\begin{proposition}\label{prop:ChandraAux}
 Let $(\al^{[\pm2]})$  be solutions of the Teukolsky problem on $\mathcal{M}$ as in Definition~\ref{def:sysTeuk}. Let $\psi^{[\pm2]},\Psi^{[\pm2]}$ be the Chandrasekhar transformations defined by~\eqref{eq:Chandra}. Then, the following inhomogeneous Regge-Wheeler equations hold:
  \begin{align}
    \begin{aligned}
      \RW^{[+2]}\widetilde{\al}^{[+2]} = - 2w'\psi^{[+2]} - w\le(2-\frac{12M}{r}\ri)\widetilde{\al}^{[+2]}, \\
      \RW^{[-2]}\widetilde{\al}^{[-2]} = + 2w'\psi^{[-2]} - w\le(2-\frac{12M}{r}\ri)\widetilde{\al}^{[-2]}, \label{eq:Teukinho}
    \end{aligned}\\ \nonumber \\
    \begin{aligned}
      \RW^{[+2]}\psi^{[+2]} = - w'\Psi^{[+2]} + 6Mw\widetilde{\al}^{[+2]}, \\
      \RW^{[-2]}\psi^{[-2]} = + w'\Psi^{[-2]} - 6Mw\widetilde{\al}^{[-2]}, \label{eq:Rwinterminho}
    \end{aligned}\\ \nonumber \\
    \begin{aligned}
      \RW^{[+2]}\Psi^{[+2]} = 0, \\
      \RW^{[-2]}\Psi^{[-2]} = 0.\label{eq:RWinhorev}
    \end{aligned}
  \end{align}
  Moreover,
  \begin{itemize}
  \item $\psi^{[+2]}$, $\De^{-1}\psi^{[-2]}$ are regular at the horizon in the sense of~\eqref{eq:defreghorgeneral} and satisfy the boundary conditions
    \begin{subequations}\label{eq:RWintermBC}
      \begin{align}
        \psi^{[+2]}-(\psi^{[-2]})^\ast & \xrightarrow{r\to+\infty} 0,\label{eq:RWintermBCD}\\
        r^2\pr_r\psi^{[+2]} + r^2\pr_r(\psi^{[-2]})^\ast & \xrightarrow{r\to+\infty} 0,\label{eq:RWintermBCN}                    
      \end{align}
    \end{subequations}
  \item $\Psi^{[+2]},\Psi^{[-2]}$ are regular at the horizon in the sense of~\eqref{eq:defreghorgeneral} and satisfy the conformal Regge-Wheeler boundary conditions~\eqref{eq:RWBC}.
  \end{itemize}
\end{proposition}

\subsection{Preliminary transport estimates with loss of derivatives}\label{sec:prelimtrans}

\begin{proposition} \label{prop:bastra}
Under the assumptions of Proposition \ref{prop:TeukfromRW} we have for any $t_2^\ast \geq t_1^\ast$ the estimates
\begin{align}\label{est:prelimtransaltbis-2w}
  \begin{aligned}
    & \norm{w^{-2}\alt^{[-2]}}^2_{L^2(\Si_{t^\ast_2})} + \int_{t^\ast_1}^{t^\ast_2}\norm{w^{-2}\alt^{[-2]}}^2_{L^2(\Si_{t^\ast})}\,\d t^\ast + \lim_{r\to +\infty}\int_{t^\ast_1}^{t^\ast_2}|\alt^{[-2]}|^2\,\d t^\ast\\
    \les_{M,k} &  \norm{w^{-2}\alt^{[-2]}}^2_{L^2(\Si_{t^\ast_1})} + \norm{w^{-1}\psi^{[-2]}}^2_{L^2(\Si_{t^\ast_1})} + \int_{t^\ast_1}^{t^\ast_2}\norm{\Psi^{[-2]}}^2_{L^2(\Si_{t^{\ast}})}\,\d t^{\ast}.
  \end{aligned}
\end{align}
\begin{align}\label{est:prelimtransalt+2bisnow}
  \begin{aligned}
    & \norm{\alt^{[+2]}}^2_{L^2(\Si_{t^\ast_2})} + \int_{t^\ast_1}^{t^\ast_2}\norm{\alt^{[+2]}}^2_{L^2(\Si_{t^\ast})}\,\d t^\ast \\
    & + \lim_{r\to r_+} \int_{t^\ast_1}^{t^\ast}|\alt^{[+2]}|^2\,\d t^\ast + \lim_{r\to+\infty}\int_{t^\ast_1}^{t^\ast_2}|\alt^{[+2]}|^2\,\d t^\ast \\
    \les_{M,k} & \norm{\alt^{[+2]}}^2_{L^2(\Si_{t^\ast_1})} + \norm{w^{-2}\alt^{[-2]}}^2_{L^2(\Si_{t^\ast_1})} \\
    & + \norm{\psi^{[+2]}}^2_{L^2(\Si_{t^\ast_1})} + \norm{w^{-1}\psi^{[-2]}}^2_{L^2(\Si_{t^\ast_1})} + \int_{t^\ast_1}^{t^\ast_2}\norm{\Psi^{[\pm2]}}^2_{L^2(\Si_{t^{\ast}})}\,\d t^{\ast}.
  \end{aligned}
\end{align}
\end{proposition}

The proposition will follow from some general transport estimates, which are proven in the next lemma. 
\begin{lemma}[Transport lemma]\label{lem:transport}
  Let $f$ be a smooth spacetime function. For all $t^\ast_2\geq t^\ast_1$, we have
  \begin{align}\label{est:transportL}
    \begin{aligned}
      & \norm{w^{-\kappa}f}^2_{L^2(\Si_{t^\ast_2})} + \int_{t^\ast_1}^{t^\ast_2}\norm{w^{-\kappa}f}^2_{L^2(\Si_{t^\ast})}\,\d t^\ast + \lim_{r\to+\infty}\int_{t^\ast_1}^{t^\ast_2}|f|^2\,\d t^\ast \\
      \les_{M,k} & \norm{w^{-\kappa}f}^2_{L^2(\Si_{t^\ast_1})} + \lim_{r\to r_+}\int_{t^\ast_1}^{t^\ast_2}w^{-2\kappa+1}|f|^2\,\d t^\ast + \int_{t^\ast_1}^{t^\ast_2}\norm{w^{-\kappa}L f}^2_{L^2(\Si_{t^\ast})}\,\d t^\ast,
    \end{aligned}
  \end{align}
  for $\kappa=1,2$, and
  \begin{align}\label{est:transportLb}
    \begin{aligned}
      & \norm{f}_{L^2(\Si_{t^\ast_2})}^2 + \int_{t^\ast_1}^{t^\ast_2}\norm{f}^2_{L^2(\Si_{t^\ast})}\,\d t^\ast + \lim_{r\to r_+}\int_{t^\ast_1}^{t^\ast_2}|f|^2\,\d t^\ast \\
      \les_{M,k} & \norm{f}^2_{L^2(\Si_{t^\ast_1})} + \lim_{r\to+\infty}\int_{t^\ast_1}^{t^\ast_2}|f|^2\,\d t^\ast + \int_{t^\ast_1}^{t^\ast_2}\norm{w^{-1}\Lb f}^2_{L^2(\Si_{t^\ast})}\,\d t^\ast.
    \end{aligned}
  \end{align}
\end{lemma}
\begin{proof}
  Using~\eqref{eq:defLLb}, we have
  \begin{align}\label{eq:transportid}
    \begin{aligned}
      \le(1+\frac{1}{r-r_+}\ri)^{2\kappa-1}\Re\le(f^\ast w^{-1}L f\ri) & = \half\pr_{t^\ast}\le(w^{-1}\le(1+\frac{1}{r-r_+}\ri)^{2\kappa-1}\frac{\De_+}{\De_0}|f|^2\ri) \\
      & \quad + \half r^2\pr_{r}\le(\le(1+\frac{1}{r-r_+}\ri)^{2\kappa-1}|f|^2\ri) \\
      & \quad + \frac{2\kappa-1}{2} \frac{r^2}{(r-r_+)^2}\le(1+\frac{1}{r-r+}\ri)^{2\kappa-2} |f|^2.
    \end{aligned}
  \end{align}
  Integrating~\eqref{eq:transportid}, we have
  \begin{align*}
    \begin{aligned}
      & \int_{t^\ast_1}^{t^\ast_2}\int_{r_+}^{+\infty} \le(1+\frac{1}{r-r_+}\ri)^{2\kappa-1}\Re\le(f^\ast w^{-1}L f\ri) \, \frac{\d r}{r^2}\d t^\ast \\
      = & \; \le[\half \int_{r_+}^{+\infty}\le(w^{-1}\le(1+\frac{1}{r-r_+}\ri)^{2\kappa-1}\frac{\De_+}{\De_0}|f|^2\ri) \, \frac{\d r}{r^2} \ri]_{t^\ast_1}^{t^\ast_2} \\
      & \quad + \le[\half \int_{t^\ast_1}^{t^\ast_2} \le(1+\frac{1}{r-r+}\ri)^{2\kappa-1}|f|^2\,\d t^\ast \ri]_{r_+}^{+\infty}  \\
      & \quad + \frac{2\kappa-1}{2} \int_{t^\ast_1}^{t^\ast_2}\int_{r_+}^{+\infty} \frac{r^2}{(r-r_+)^2}\le(1+\frac{1}{r-r_+}\ri)^{2\kappa-2}|f|^2\,\frac{\d r}{r^2}\d t^\ast.
    \end{aligned}
  \end{align*}
  Hence,
  \begin{align*}
    \begin{aligned}
      & \quad \half \int_{r_+}^{+\infty}\le(w^{-1}\le(1+\frac{1}{r-r_+}\ri)^{2\kappa-1}\frac{\De_+}{\De_0}|f|^2\ri) \, \frac{\d r}{r^2} \bigg|_{t^\ast=t^\ast_2}  + \half \lim_{r\to+\infty} \int_{t^\ast_1}^{t^\ast_2} |f|^2\,\d t^\ast \\
      & + \frac{2\kappa-1}{2} \int_{t^\ast_1}^{t^\ast_2}\int_{r_+}^{+\infty} \frac{r^2}{(r-r_+)^2}\le(1+\frac{1}{r-r_+}\ri)^{2\kappa-2}|f|^2\,\frac{\d r}{r^2}\d t^\ast \\
      \les_{M,k} & \; \half \int_{r_+}^{+\infty}\le(w^{-1}\le(1+\frac{1}{r-r_+}\ri)^{2\kappa-1}\frac{\De_+}{\De_0}|f|^2\ri) \, \frac{\d r}{r^2} \bigg|_{t^\ast=t^\ast_1} + \half \lim_{r\to r_+}\int_{t^\ast_1}^{t^\ast_2} \le(1+\frac{1}{r-r+}\ri)^{2\kappa-1}|f|^2\,\d t^\ast\\
      & + \int_{t^\ast_1}^{t^\ast_2}\norm{(r-r_+)^{-\kappa}f}_{L^2(\Si_{t^\ast})}\norm{w^{-\kappa}Lf}_{L^2(\Si_{t^\ast})}\,\d t^\ast, 
    \end{aligned}
  \end{align*}
  from which~\eqref{est:transportL} follows by an absorption argument. Estimate~\eqref{est:transportLb} follows along analogous lines and this finishes the proof of the lemma.
\end{proof}

\begin{proof}[Proof of Proposition \ref{prop:bastra}]
Applying the transport estimates~\eqref{est:transportL} with $\kappa=1$ and $\kappa=2$ respectively to the Chandrasekhar transformations~\eqref{eq:Chandra} of $\psi^{[-2]}$ and $\alt^{[-2]}$ respectively, using that $\psi^{[-2]}=O(\De)$ and $\alt^{[-2]}=O(\De^2)$ when $r\to r_+$, we have for all $t^\ast_2\geq t^\ast_1$
\begin{align}
  \begin{aligned}\label{est:prelimtranspsi-2w}
     & \norm{w^{-1}\psi^{[-2]}}^2_{L^2(\Si_{t^\ast_2})} + \int_{t^\ast_1}^{t^\ast_2}\norm{w^{-1}\psi^{[-2]}}^2_{L^2(\Si_{t^\ast})}\,\d t^\ast + \lim_{r\to +\infty}\int_{t^\ast_1}^{t^\ast_2}|\psi^{[-2]}|^2\,\d t^\ast \\
    \les_{M,k} & \norm{w^{-1}\psi^{[-2]}}^2_{L^2(\Si_{t^\ast_1})}  + \int_{t^\ast_1}^{t^\ast_2}\norm{\Psi^{[-2]}}^2_{L^2(\Si_{t^{\ast}})}\,\d t^{\ast}.
  \end{aligned}
\end{align}
and 
\begin{align}
  \begin{aligned}\label{est:prelimtransalt-2w}
    & \norm{w^{-2}\alt^{[-2]}}^2_{L^2(\Si_{t^\ast_2})} + \int_{t^\ast_1}^{t^\ast_2}\norm{w^{-2}\alt^{[-2]}}^2_{L^2(\Si_{t^\ast})}\,\d t^\ast + \lim_{r\to +\infty}\int_{t^\ast_1}^{t^\ast_2}|\alt^{[-2]}|^2\,\d t^\ast\\
    \les_{M,k} &  \norm{w^{-2}\alt^{[-2]}}^2_{L^2(\Si_{t^\ast_1})} + \int_{t^\ast_1}^{t^\ast_2}\norm{w^{-1}\psi^{[-2]}}^2_{L^2(\Si_{t^{\ast}})}\,\d t^{\ast}.
  \end{aligned}
\end{align}
From~\eqref{est:prelimtranspsi-2w} and~\eqref{est:prelimtransalt-2w} we infer the desired (\ref{est:prelimtransaltbis-2w}).

Applying the transport estimate~\eqref{est:transportLb} to the definition of the Chandrasekhar transformations~\eqref{eq:Chandra}, using~\eqref{est:prelimtranspsi-2w} and~\eqref{est:prelimtransalt-2w} and that $\alt^{[+2]} - (\alt^{[-2]})^\ast\to0$ and $\psi^{[+2]}-(\psi^{[-2]})^\ast\to0$ when $r\to+\infty$ (see~\eqref{eq:TeukBCDirichlet} and~\eqref{eq:RWintermBCD}), we have for all $t^\ast_2\geq t^\ast_1$
\begin{align}\label{est:prelimtranspsi+2now}
  \begin{aligned}
    & \norm{\psi^{[+2]}}^2_{L^2(\Si_{t^\ast_2})} + \int_{t^\ast_1}^{t^\ast_2}\norm{\psi^{[+2]}}^2_{L^2(\Si_{t^\ast})}\,\d t^\ast \\
    & + \lim_{r\to r_+} \int_{t^\ast_1}^{t^\ast}|\psi^{[+2]}|^2\,\d t^\ast + \lim_{r\to+\infty}\int_{t^\ast_1}^{t^\ast_2}|\psi^{[+2]}|^2\,\d t^\ast \\
    \les_{M,k} &  \norm{\psi^{[+2]}}^2_{L^2(\Si_{t^\ast_1})} + \norm{w^{-1}\psi^{[-2]}}^2_{L^2(\Si_{t^\ast_1})} + \int_{t^\ast_1}^{t^\ast_2}\norm{\Psi^{[\pm2]}}^2_{L^2(\Si_{t^{\ast}})}\,\d t^{\ast},
  \end{aligned}
\end{align}
\begin{align}\label{est:prelimtransalt+2now}
  \begin{aligned}
    & \norm{\alt^{[+2]}}^2_{L^2(\Si_{t^\ast_2})} + \int_{t^\ast_1}^{t^\ast_2}\norm{\alt^{[+2]}}^2_{L^2(\Si_{t^\ast})}\,\d t^\ast \\
    & + \lim_{r\to r_+} \int_{t^\ast_1}^{t^\ast}|\alt^{[+2]}|^2\,\d t^\ast + \lim_{r\to+\infty}\int_{t^\ast_1}^{t^\ast_2}|\alt^{[+2]}|^2\,\d t^\ast \\
    \les_{M,k} & \norm{\alt^{[+2]}}^2_{L^2(\Si_{t^\ast_1})} + \norm{w^{-2}\alt^{[-2]}}^2_{L^2(\Si_{t^\ast_1})} \\
    & + \int_{t^\ast_1}^{t^\ast_2}\norm{\psi^{[+2]}}^2_{L^2(\Si_{t^{\ast}})}\,\d t^{\ast}+ \int_{t^\ast_1}^{t^\ast_2}\norm{w^{-1}\psi^{[-2]}}^2_{L^2(\Si_{t^{\ast}})}\,\d t^{\ast},
  \end{aligned}
\end{align}
from which we infer (\ref{est:prelimtransalt+2bisnow}).
\end{proof}
\subsection{Energy and redshift estimates for $\psi^{[\pm2]}$}\label{sec:energyestpsipm2}
We derive a preliminary estimate for $\psi^{[\pm2]}$ from the inhomogeneous Regge-Wheeler equation satisfied by these quantities. 
\begin{proposition}\label{prop:psiintv1}
Under the assumptions of Proposition \ref{prop:TeukfromRW} we have for any $t_2^\ast \geq t_1^\ast$ the estimate
\begin{align}\label{est:psiintv1}
  \begin{aligned}
    \widetilde{\mathrm{E}}[\psi^{[\pm2]}](t^\ast_2) + \widetilde{\mathrm{E}}^\HH[\psi^{[\pm2]}](t^\ast_2;t^\ast_1) 
     &\les_{M,k} \widetilde{\mathrm{E}}[\psi^{[\pm2]}](t^\ast_1) \\
     & \qquad + \norm{\psi^{[\pm2]}}^2_{L^2(\Si_{t^\ast_2})}  + \int_{t^\ast_1}^{t^\ast_2} \left[ \norm{\Psi^{[\pm2]}}^2_{L^2(\Si_{t^\ast})} \d t^\ast   + \norm{\alt^{[\pm2]}}_{L^2(\Si_{t^\ast})}^2 \right]\d t^\ast .
  \end{aligned}
\end{align}
\end{proposition}

The proposition will follow immediately from the following general lemma. 
\begin{lemma}[Inhomogeneous energy and redshift estimates]\label{lem:inhoenergest}
  Let $\Phi$ be a smooth spacetime function, regular at the horizon and at infinity in the sense of~\eqref{eq:defreghorgeneral} and~\eqref{eq:defreginfgeneral}, and assume that $\Phi$ satisfies either the Dirichlet or the Neumann condition
  \begin{align}\label{eq:PsiDorNBC}
    \Phi & \xrightarrow{r\to+\infty} 0, & r^2\pr_r\Phi & \xrightarrow{r\to+\infty} 0,
  \end{align}
Then, for all $t^\ast_2\geq t^\ast_1$, we have
  \begin{align}\label{est:inhoenergestPsi}
    \begin{aligned}
      \overline{\mathrm{E}}[\Phi](t_2^\ast) + \overline{\mathrm{E}}^\HH[\Phi](t^\ast_2,t^\ast_1) & \les_{M,k} \overline{\mathrm{E}}[\Phi](t_1^\ast) + \sup_{t^\ast_1\leq t^\ast\leq t^\ast_2}\norm{\Phi}^2_{L^2(\Si_{t^\ast})} + \int_{t^\ast_1}^{t^\ast_2}\norm{w^{-1}\RW\Phi}_{L^2(\Si_{t^\ast})}^2\,\d t^\ast,
    \end{aligned}
  \end{align}
  with the term $\norm{\Phi}^2_{L^2(\Si_{t^\ast})}$ only needed on the RHS in the Neumann case.
\end{lemma}
\begin{proof}[Proof of Lemma~\ref{lem:inhoenergest}]
We consider the energy identity~\eqref{eq:pfenergid3} of Lemma~\ref{lem:energyid} noting that with Dirichlet or Neumann conditions~\eqref{eq:PsiDorNBC}, the boundary term $\Re\le(\pr_{r^\ast}\Phi\pr_{t^\ast}\Phi^\ast\ri)$ vanishes. In the Dirichlet case, the energy $\mathring{\mathrm{E}}_{m\ellmode}\le[\Phi\ri]$ is already coercive by Lemma \ref{lem:Hardy}, hence in particular $\mathring{\mathrm{E}}_{m\ellmode}\le[\Phi\ri] \simeq_{M,k}\widetilde{\mathrm{E}}_{m\ellmode}\le[\Phi\ri]$ and we conclude (\ref{est:inhoenergestPsiproof}) without the $\|\Phi\|^2$-term. In general, we have by definition
  \begin{align*}
    \widetilde{\mathrm{E}}[\Phi] & = \mathring{\mathrm{E}}[\Phi] + \int_{-\infty}^{\frac{\pi}{2}}w \frac{6M}{r}|\Phi|^2\,\d r^\ast,
  \end{align*}
and deduce
 \begin{align}\label{est:inhoenergestPsiproof}
    \begin{aligned}
      \widetilde{\mathrm{E}}[\Phi](t_2^\ast) + \widetilde{\mathrm{E}}^\HH[\Phi](t^\ast_2,t^\ast_1) & \les_{M,k} \widetilde{\mathrm{E}}[\Phi](t_1^\ast) + \norm{\Phi}^2_{L^2(\Si_{t^\ast_2})} + \int_{t^\ast_1}^{t^\ast_2}\norm{w^{-1}\RW\Phi}_{L^2(\Si_{t^\ast})}\le(\widetilde{\mathrm{E}}[\Phi](t^\ast)\ri)^{1/2}\,\d t^\ast,
    \end{aligned}
  \end{align}
  for all $t^\ast_2\geq t^\ast_1$. Taking the sup and using a standard absorption argument, we can conclude (\ref{est:inhoenergestPsi}) for the $\widetilde{E}$ energies instead of the $\overline{E}$ energies. To upgrade the statement we can apply the redshift Lemma~\ref{lem:redshift} and use the standard pigeonhole argument to conclude (\ref{est:inhoenergestPsi})
as stated. 
\end{proof}

\begin{proof}[Proof of Proposition \ref{prop:psiintv1}]
From \eqref{eq:RWintermBC}, we know that the quantities $\psi^D := \psi^{[+2]}-(\psi^{[-2]})^\ast$ and $\psi^N:=\psi^{[+2]}+(\psi^{[-2]})^\ast$ satisfy respectively the Dirichlet and Neumann conditions~\eqref{eq:PsiDorNBC} as well as the inhomogeneous Regge-Wheeler equations (see~\eqref{eq:Rwinterminho}). Thus, applying~\eqref{est:inhoenergestPsi} directly yields the desired estimate.
\end{proof}

\subsection{Elliptic estimates for $\psi^{[+2]}$ and $w^{-1}\psi^{[-2]}$}\label{sec:ellestpsipm2}

The next estimate can be understood as \emph{angular} derivatives of $(\psi^{[+2]})$ without loss of derivatives.
\begin{proposition}
Under the assumptions of Proposition \ref{prop:TeukfromRW} we have for any $t_2^\ast \geq t_1^\ast$ the estimate (\ref{est:spacetimeell2psipm2}) below.
\end{proposition}

\begin{proof}
Multiplying the first equation of~\eqref{eq:Rwinterminho} by $(\psi^{[+2]})^\ast$ and the second by $k^4w^{-2}(\psi^{[-2]})^\ast$ and summing them, using~\eqref{eq:Chandra} and~\eqref{eq:RW}, gives
\begin{align}\label{eq:sumellidpsipm2}
  \begin{aligned}
    0 & = L(w\Psi^{[+2]})(\psi^{[+2]})^\ast + w\ellmode(\ellmode+1)|\psi^{[+2]}|^2 - w\frac{6M}{r}|\psi^{[+2]}|^2 -w'\Psi^{[+2]}(\psi^{[+2]})^\ast + 6Mw\alt^{[+2]}(\psi^{[+2]})^\ast \\
    & \quad + k^4\Lb\Psi^{[-2]}(w^{-1}\psi^{[-2]})^\ast + k^4w\ellmode(\ellmode+1)|w^{-1}\psi^{[-2]}|^2 -k^4 w\frac{6M}{r}|w^{-2}\psi^{[-2]}|^2\\
    & \quad - 6Mk^4\alt^{[-2]}(w^{-1}\psi^{[-2]})^\ast.
  \end{aligned}
\end{align}
Taking the real part of~\eqref{eq:sumellidpsipm2} and integrating on $(t^\ast_1,t^\ast_2)_{t^\ast}\times(-\infty,\frac{\pi}{2})_{r^\ast}$ we get
\begin{align}\label{est:sumellpsipm2}
  \begin{aligned}
    & \int_{t^\ast_1}^{t^\ast_2}\int_{-\infty}^{\frac{\pi}{2}} w\ellmode(\ellmode+1)\le(|\psi^{[+2]}|^2 + |w^{-1}\psi^{[-2]}|^2\ri)  \,\d t^\ast\d r^\ast \\
    \les_{M,k} & \; \bigg|\int_{t^\ast_1}^{t^\ast_2}\int_{-\infty}^{\frac{\pi}{2}} \Re\bigg(L(w\Psi^{[+2]})(\psi^{[+2]})^\ast + k^4(\Lb\Psi^{[-2]})(w^{-1}\psi^{[-2]})^\ast \bigg) \,\d t^\ast\d r^\ast \bigg|\\
    & \quad + \int_{t^\ast_1}^{t^\ast_2}\int_{-\infty}^{\frac{\pi}{2}} w\le(|\alt^{[+2]}|^2 + |w^{-2}\alt^{[-2]}|^2 + |\psi^{[+2]}|^2+|w^{-1}\psi^{[-2]}|^2\ri) \,\d t^\ast\d r^\ast \\
    & \quad + \int_{t^\ast_1}^{t^\ast_2}\int_{-\infty}^{\frac{\pi}{2}} w|\Psi^{[\pm2]}|^2 \,\d t^\ast\d r^\ast.
  \end{aligned}
\end{align}
Integrating by part, using~\eqref{eq:defLLb}, we have 
\begin{align}\label{est:sourceellpsipm2}
  \begin{aligned}
    & \le|\int_{t^\ast_1}^{t^\ast_2}\int_{-\infty}^{\frac{\pi}{2}} \Re\le(L(w\Psi^{[+2]})(\psi^{[+2]})^\ast + k^4(\Lb\Psi^{[-2]})(w^{-1}\psi^{[-2]})^\ast\ri) \,\d t^\ast\d r^\ast\ri| \\
    \les & \; \int_{t^\ast_1}^{t^\ast_2}\int_{-\infty}^{\frac{\pi}{2}} w |\Psi^{[+2]}||L(\psi^{[+2]})| + \int_{t^\ast_1}^{t^\ast_2}\int_{-\infty}^{\frac{\pi}{2}} |\Psi^{[-2]}||\Lb(w^{-1}\psi^{[-2]})| \,\d t^\ast\d r^\ast \\
    & + \le[\int_{-\infty}^{\frac{\pi}{2}} \le(w\frac{\De_{+}}{\De_0}\Re\le(\Psi^{[+2]}(\psi^{[+2]})^\ast\ri) + k^4\frac{\De_{-}}{\De_0}\Re\le(\Psi^{[-2]}(w^{-1}\psi^{[-2]})^\ast\ri)\ri)\,\d r^\ast\ri]_{t^\ast_1}^{t^\ast_2} \\
    & + \le[\int_{t^\ast_1}^{t^\ast_2} \Re\le(w\Psi^{[+2]}(\psi^{[+2]})^\ast - k^4\Psi^{[-2]}(w^{-1}\psi^{[-2]})^\ast\ri) \,\d t^\ast\ri]_{-\infty}^{\frac{\pi}{2}}.
  \end{aligned}
\end{align}
Using the boundary conditions~\eqref{eq:RWintermBCD},~\eqref{eq:RWBCD} at infinity\footnote{In fact, we summed the $\pm2$ estimates at the beginning of the proof in order to have cancellation of the boundary terms at infinity.} and the regularity at the horizon~\eqref{eq:defreghorgeneral}, for $\psi^{[+2]}, w^{-1}\psi^{[-2]}$ and $\Psi^{[\pm2]}$, three terms in the last line of the RHS of~\eqref{est:sourceellpsipm2} vanish. Moreover the second line can be estimated using the transport estimates~\eqref{est:prelimtranspsi-2w} and~\eqref{est:prelimtranspsi+2now}. Thus, plugging~\eqref{est:sourceellpsipm2} in~\eqref{est:sumellpsipm2} and using the again the transport estimates~\eqref{est:prelimtranspsi-2w},~\eqref{est:prelimtransaltbis-2w},~\eqref{est:prelimtranspsi+2now} and~\eqref{est:prelimtransalt+2now} to estimate the integrals of $\alt^{[\pm2]},\psi^{[\pm2]}$ in the RHS of~\eqref{est:sumellpsipm2}, we obtain
\begin{align}\label{est:spacetimeell2psipm2}
  \begin{aligned}
    & \int_{t^\ast_1}^{t^\ast_2}\int_{-\infty}^{\frac{\pi}{2}} w\ellmode(\ellmode+1)\le(|\psi^{[+2]}|^2 + |w^{-1}\psi^{[-2]}|^2\ri)  \,\d t^\ast\d r^\ast \\
    \les_{M,k} & \; \int_{t^\ast_1}^{t^\ast_2}\int_{-\infty}^{\frac{\pi}{2}} w |\Psi^{[+2]}||L(\psi^{[+2]})|\,\d t^\ast\d r^\ast + \int_{t^\ast_1}^{t^\ast_2}\int_{-\infty}^{\frac{\pi}{2}} |\Psi^{[-2]}||\Lb(w^{-1}\psi^{[-2]})| \,\d t^\ast\d r^\ast \\
    & \quad + \lim_{r\to r_+} \int_{t^\ast_1}^{t^\ast_2}|\Psi^{[-2]}||w^{-1}\psi^{[-2]}|\,\d t^\ast \\
    & \quad + \norm{\alt^{[+2]}}_{L^2(\Si_{t^\ast_1})}^2 +\norm{\psi^{[+2]}}^2_{L^2(\Si_{t^\ast_1})} + \norm{w^{-2}\alt^{[-2]}}_{L^2(\Si_{t^\ast_1})}^2 +\norm{w^{-1}\psi^{[-2]}}^2_{L^2(\Si_{t^\ast_1})} \\
    & \quad +\norm{\Psi^{[\pm2]}}^2_{L^2(\Si_{t^\ast_1})} +\norm{\Psi^{[\pm2]}}^2_{L^2(\Si_{t^\ast_2})} + \int_{t^\ast_1}^{t^\ast_2} \norm{\Psi^{[\pm2]}}_{L^2(\Si_{t^\ast})}^2 \,\d t^\ast.
  \end{aligned}
\end{align}
\end{proof} 

\subsection{Morawetz estimates for $\psi^{[\pm2]}$}\label{sec:Morawetzpsipm2}
Having obtained estimates for the angular derivatives in the previous subsection we turn to estimating the spacetime integrals of $L\psi^{[+2]}$ and $\Lb\psi^{[-2]}$ derivatives. (Note that by definition, $\underline{L}\psi^{[+2]}=w\Psi^{[+2]}$ and $L \psi^{[-2]}=w\Psi^{[-2]}$ and that $\Psi^{[\pm 2]}$ has already been estimated.)
\begin{proposition}
Under the assumptions of Proposition \ref{prop:TeukfromRW} we have for any $t_2^\ast \geq t_1^\ast$ the estimate (\ref{est:spacetimeLpsip2Lbpsim2}) below.
\end{proposition}

\begin{proof}
We multiply each of the equations of~\eqref{eq:Rwinterminho} by respectively $pL(\psi^{[+2]})^\ast$ and $p(\Lb\psi^{[-2]})^\ast$, where $p$ is a function to be determined. Taking the difference of these identities gives
\begin{align}\label{eq:Moridpsipm2}
  \begin{aligned}
    0 & = \int_{t^\ast_1}^{t^\ast_2}\int_{-\infty}^{\frac{\pi}{2}} \le(p(L\psi^{[+2]})^\ast\Lb L\psi^{[+2]} - p(\Lb\psi^{[-2]})^\ast L\Lb(\psi^{[-2]}) \ri)\,\d t^\ast\d r^\ast\\
    & \quad + \int_{t^\ast_1}^{t^\ast_2}\int_{-\infty}^{\frac{\pi}{2}} pw\le(\ellmode(\ellmode+1)-\frac{6M}{r}\ri)\le((L\psi^{[+2]})^\ast\psi^{[+2]} - (\Lb\psi^{[-2]})^\ast\psi^{[-2]}\ri) \,\d t^\ast\d r^\ast\\
    & \quad + \int_{t^\ast_1}^{t^\ast_2}\int_{-\infty}^{\frac{\pi}{2}} \le(-w'pL(\psi^{[+2]})^\ast\Psi^{[+2]} - w'p\Lb(\psi^{[-2]})^\ast\Psi^{[-2]}\ri) \,\d t^\ast\d r^\ast\\
    & \quad + \int_{t^\ast_1}^{t^\ast_2}\int_{-\infty}^{\frac{\pi}{2}} \le(6Mw pL(\psi^{[+2]})^\ast\alt^{[+2]} +6Mw p\Lb(\psi^{[-2]})^\ast\alt^{[-2]}\ri) \,\d t^\ast\d r^\ast.
  \end{aligned}
\end{align}
For the term in the first line of (\ref{eq:Moridpsipm2}) we have 
\begin{align}\label{eq:Moridpsipm2bis}
  \begin{aligned}
    & \Re\le(\int_{t^\ast_1}^{t^\ast_2}\int_{-\infty}^{\frac{\pi}{2}} \le(p(L\psi^{[+2]})^\ast\Lb L\psi^{[+2]} - p(\Lb\psi^{[-2]})^\ast L\Lb(\psi^{[-2]}) \ri)\,\d t^\ast\d r^\ast\ri) \\
    = & \; \half \int_{t^\ast_1}^{t^\ast_2}\int_{-\infty}^{\frac{\pi}{2}} \le(p' |L\psi^{[+2]}|^2 +p' |\Lb(\psi^{[-2]})|^2 \ri)\,\d t^\ast\d r^\ast \\
    & + \half \le[\int_{-\infty}^{\frac{\pi}{2}}p \le(\frac{\De_{-}}{\De_0}|L\psi^{[+2]}|^2 - \frac{\De_+}{\De_0}|\Lb\psi^{[-2]}|^2\ri)\,\d r^\ast\ri]_{t^\ast_1}^{t^\ast_2} \\
    & + \half\bigg[\int_{t^\ast_1}^{t^\ast_2}\le(-p|L\psi^{[+2]}|^2 - p|\Lb\psi^{[-2]}|^2\ri)\,\d t^\ast\bigg]_{-\infty}^{\frac{\pi}{2}}.
  \end{aligned}
\end{align}
For the term in the second line of (\ref{eq:Moridpsipm2}) we have 
\begin{align}\label{eq:Moridpsipm2bisbis}
  \begin{aligned}
    & \Re\le(\int_{t^\ast_1}^{t^\ast_2}\int_{-\infty}^{\frac{\pi}{2}} pw\le(\ellmode(\ellmode+1)-\frac{6M}{r}\ri)\le((L\psi^{[+2]})^\ast\psi^{[+2]} - (\Lb\psi^{[-2]})^\ast\psi^{[-2]}\ri) \,\d t^\ast\d r^\ast\ri)\\
    = & \; -\half\int_{t^\ast_1}^{t^\ast_2}\int_{-\infty}^{\frac{\pi}{2}} \le(pw\le(\ellmode(\ellmode+1)-\frac{6M}{r}\ri)\ri)' \le(|\psi^{[+2]}|^2 + |\psi^{[-2]}|^2\ri)\,\d t^\ast\d r^\ast \\
    & + \le[\half\int_{-\infty}^{\frac{\pi}{2}}pw\le(\ellmode(\ellmode+1)-\frac{6M}{r}\ri) \le(|\psi^{[+2]}|^2 - |\psi^{[-2]}|^2\ri)\,\d r^\ast\ri]_{t^\ast_1}^{t^\ast_2} \\
    & + \le[\half\int_{t^\ast_1}^{t^\ast_2}pw\le(\ellmode(\ellmode+1)-\frac{6M}{r}\ri) \le(|\psi^{[+2]}|^2 + |\psi^{[-2]}|^2\ri)\,\d t^\ast\ri]_{-\infty}^{\frac{\pi}{2}}.
  \end{aligned}
\end{align}
Taking $p(r)=-\frac{1}{r}$ (so that $p'=w$ and $p$ vanishes at infinity), plugging~\eqref{eq:Moridpsipm2bis} and~\eqref{eq:Moridpsipm2bisbis} in the real part of~\eqref{eq:Moridpsipm2}, using the regularity at the horizon~\eqref{eq:defreghorgeneral} for $\psi^{[+2]}$ and $w^{-1}\psi^{[-2]}$, we get
\begin{align}\label{est:spacetimeLpsip2Lbpsim2}
  \begin{aligned}
    & \half \int_{t^\ast_1}^{t^\ast_2}\int_{-\infty}^{\frac{\pi}{2}} \le(w |L\psi^{[+2]}|^2 + w |\Lb\psi^{[-2]}|^2 \ri)\,\d t^\ast\d r^\ast + \half \int_{-\infty}^{\frac{\pi}{2}}\frac{\De_+}{\De_0} |\Lb\psi^{[-2]}|^2\,\d r^\ast \bigg|_{t^\ast_2} \\
    \les_{M,k} & \; \half \int_{-\infty}^{\frac{\pi}{2}}\frac{\De_+}{\De_0} |\Lb\psi^{[-2]}|^2\,\d r^\ast \bigg|_{t^\ast_1} + \widetilde{\mathrm{E}}[\psi^{[\pm2]}](t^\ast_1) + \widetilde{\mathrm{E}}[\psi^{[\pm2]}](t^\ast_2) \\
    & + \widetilde{\mathrm{E}}^\HH[\psi^{[\pm2]}](t^\ast_2;t^\ast_1) + \int_{t^\ast_1}^{t^\ast_2}\int_{-\infty}^{\frac{\pi}{2}} w\ellmode(\ellmode+1)|\psi^{[\pm2]}|^2 \,\d t^\ast\d r^\ast \\
    & \quad + \int_{t^\ast_1}^{t^\ast_2}\int_{-\infty}^{\frac{\pi}{2}} \le(w |L\psi^{[+2]}||\Psi^{[+2]}| + w |\Lb\psi^{[-2]}||\Psi^{[-2]}|\ri) \,\d t^\ast\d r^\ast\\
    & \quad + \int_{t^\ast_1}^{t^\ast_2}\int_{-\infty}^{\frac{\pi}{2}} \le(w |L\psi^{[+2]}||\alt^{[+2]}| +w|\Lb\psi^{[-2]}||\alt^{[-2]}|\ri) \,\d t^\ast\d r^\ast.
  \end{aligned}
\end{align}
\end{proof}

\subsection{Redshift estimate for $w^{-1}\psi^{[-2]}$}\label{sec:redshiftpsim2}
Estimate~\eqref{est:spacetimeLpsip2Lbpsim2} does yet estimate  $\Lb(w^{-1}\psi^{[-2]})$ with the weights (near the horizon) claimed in Proposition \ref{prop:TeukfromRW}.  We apply another redshift estimate summarised in the next proposition. 
\begin{proposition}
Under the assumptions of Proposition \ref{prop:TeukfromRW} we have for any $t_2^\ast \geq t_1^\ast$ the estimate 
\begin{align} \label{muffo}
  \begin{aligned}
    &  \int_{-\infty}^{3M}\le( \frac{1}{w} \le|\Lb\le(w^{-1}\psi^{[-2]}\ri)\ri|^2 +  \ellmode(\ellmode+1) \le|w^{-1}\psi^{[-2]}\ri|^2 \ri) \,\d r^\ast\bigg|_{t^\ast=t^\ast_2} \\
    & + \int_{t^\ast_1}^{t^\ast_2}\int_{-\infty}^{r=3M}  w^{-1} \le|\Lb(w^{-1}\psi^{[-2]})\ri|^2 \,\d t^\ast \d r^\ast  + \lim_{r\to r_+}\int_{t^\ast_1}^{t^\ast_2} \ellmode(\ellmode+1) \le|w^{-1}\psi^{[-2]}\ri|^2  \,\d t^\ast \\
    \les_{M,k} & \;  \int_{-\infty}^{4M}\le( \frac{1}{w} \le|\Lb\le(w^{-1}\psi^{[-2]}\ri)\ri|^2 +  \ellmode(\ellmode+1) \le|w^{-1}\psi^{[-2]}\ri|^2 \ri) \,\d r^\ast\bigg|_{t^\ast=t^\ast_1}  \\
    & + \norm{w^{-2}\alt^{[-2]}}_{L^2(\Si_{t^\ast_1})}^2 + \norm{w^{-1}\psi^{[-2]}}_{L^2(\Si_{t^\ast_1})}^2 + \int_{t^\ast_1}^{t^\ast_2}\norm{\Psi^{[-2]}}^2_{L^2(\Si_{t^\ast})}\,\d t^\ast \\
    & + \int_{t^\ast_1}^{t^\ast_2}\int_{-\infty}^{\frac{\pi}{2}} w \le(|\Lb\psi^{[-2]}|^2 + |w^{-1}\psi^{[-2]}|^2 \ri)\,\d t^\ast\d r^\ast.
  \end{aligned}
\end{align}
\end{proposition}

\begin{proof}
First, we remark that the Teukolsky equation~\eqref{eq:Teukinho} for $\alt^{[-2]}$ rewrites as
\begin{align}\label{eq:ellw-2alt}
  \begin{aligned}
    0 & = w^{-1}\Lb(w^{-1}\psi^{[-2]}) + \le(\ellmode(\ellmode+1) - 2 +\frac{6M}{r} \ri)\le(w^{-2}\alt^{[-2]}\ri).
  \end{aligned}
\end{align}
Taking a $L$ derivative in $w$ times~\eqref{eq:ellw-2alt}, using the definition of the Chandrasekhar transformations~\eqref{eq:Chandra} and re-using~\eqref{eq:ellw-2alt}, we get
\begin{align}\label{eq:RWw-1psi}
  \begin{aligned}
    0 & = L\Lb\le(w^{-1}\psi^{[-2]}\ri) + w\le(\ellmode(\ellmode+1) - 2 +\frac{6M}{r} \ri)(w^{-1}\psi^{[-2]}) + w^{-1}w'\Lb(w^{-1}\psi^{[-2]}) - 6M \alt^{[-2]},
  \end{aligned}
\end{align}
which we rewrite as
\begin{align}\label{eq:RWw-1psibis}
  \begin{aligned}
    -\mathfrak{F} & = L\Lb\le(w^{-1}\psi^{[-2]}\ri) + w\ellmode(\ellmode+1)(w^{-1}\psi^{[-2]})  + w^{-1}w'\Lb(w^{-1}\psi^{[-2]}),
  \end{aligned}
\end{align}
with
\begin{align*}
  \mathfrak{F} & := \le(-2 +\frac{6M}{r}\ri)\psi^{[-2]} - 6M\alt^{[-2]}.
\end{align*}
Let $p$ be a radial function to be determined. Multiplying~\eqref{eq:RWw-1psibis} by $p\Lb(w^{-1}\psi^{[-2]})^\ast$ and taking the real part, using~\eqref{eq:defLLb}, we have
\begin{align}\label{eq:energyidpsi-1}
  \begin{aligned}
    -\Re\le(\mathfrak{F}p\Lb(w^{-1}\psi^{[-2]})^\ast\ri) &  = \half p L\le(\le|\Lb\le(w^{-1}\psi^{[-2]}\ri)\ri|^2\ri) + \half pw\Lb\le(\ellmode(\ellmode+1) \le|w^{-1}\psi^{[-2]}\ri|^2 \ri) \\
                                                              & \quad + pw^{-1}w'\le|\Lb(w^{-1}\psi^{[-2]})\ri|^2 \\
                                                              & = \half p\le(\frac{\De_+}{\De_0}\pr_{t^\ast} + \pr_{r^\ast}\ri)\le(\le|\Lb\le(w^{-1}\psi^{[-2]}\ri)\ri|^2\ri) \\
                                                              & \quad + \half pw\le(\frac{\De_-}{\De_0}\pr_{t^\ast}-\pr_{r^\ast}\ri)\le(\ellmode(\ellmode+1) \le|w^{-1}\psi^{[-2]}\ri|^2 \ri) \\
                                                              & \quad + pw^{-1}w'\le|\Lb(w^{-1}\psi^{[-2]})\ri|^2.
  \end{aligned}
\end{align}
Integrating~\eqref{eq:energyidpsi-1} on $(t^\ast_1,t^\ast_2)_{t^\ast}\times(-\infty,\frac{\pi}{2})_{r^\ast}$ gives
\begin{align*}
  \begin{aligned}
    & -\int_{t^\ast_1}^{t^\ast_2}\int_{-\infty}^{\frac{\pi}{2}}\Re\le(\mathfrak{F}p\Lb(w^{-1}\psi^{[-2]})^\ast\ri)\,\d t^\ast\d r^\ast \\
    = & \; \int_{t^\ast_1}^{t^\ast_2}\int_{-\infty}^{\frac{\pi}{2}}  \le(pw^{-1}w' - \half p'\ri)\le|\Lb(w^{-1}\psi^{[-2]})\ri|^2 \,\d t^\ast \d r^\ast \\
    & + \int_{t^\ast_1}^{t^\ast_2}\int_{-\infty}^{\frac{\pi}{2}} \half (pw)' \ellmode(\ellmode+1)|w^{-1}\psi^{[-2]}|^2 \,\d t^\ast \d r^\ast \\
    & + \half \le[\int_{-\infty}^{\frac{\pi}{2}}\le(p\frac{\De_+}{\De_0}\le|\Lb\le(w^{-1}\psi^{[-2]}\ri)\ri|^2 +  (pw)\frac{\De_-}{\De_0} \ellmode(\ellmode+1) \le|w^{-1}\psi^{[-2]}\ri|^2 \ri) \,\d r^\ast\ri]_{t^\ast_1}^{t^\ast_2} \\
    & + \half \le[\int_{t^\ast_1}^{t^\ast_2}\le(p\le|\Lb\le(w^{-1}\psi^{[-2]}\ri)\ri|^2 - (pw)\ellmode(\ellmode+1) \le|w^{-1}\psi^{[-2]}\ri|^2  \ri) \,\d t^\ast\ri]_{-\infty}^{\frac{\pi}{2}}.
  \end{aligned}
\end{align*}
We take $p=w^{-1}\chi$ in~\eqref{est:pfCarleman1int1}, where $\chi$ is a standard cut-off function with support in $[r_+,4M]$, equal to $1$ on $[r_+,3M]$. Using that $p'=\le(\frac{2}{r}-\frac{6M}{r^2}\ri)w^{-1}$ on $(r_+,3M)$ and the regularity at the horizon~\eqref{eq:defreghorgeneral} for $w^{-1}\psi^{[-2]}$, the above identity rewrites as
\begin{align}\label{eq:energyidpsi-1int}
  \begin{aligned}
    &  \half\int_{-\infty}^{\frac{\pi}{2}}\le(w^{-1}\chi\frac{\De_+}{\De_0}\le|\Lb\le(w^{-1}\psi^{[-2]}\ri)\ri|^2 +  \frac{\De_-}{\De_0} \chi \ellmode(\ellmode+1) \le|w^{-1}\psi^{[-2]}\ri|^2 \ri) \,\d r^\ast\bigg|_{t^\ast=t^\ast_2} \\
    & + \frac{3}{2} \int_{t^\ast_1}^{t^\ast_2}\int_{-\infty}^{r=3M}  w^{-1}\chi\le(\frac{6M}{r^2}-\frac{2}{r}\ri)\le|\Lb(w^{-1}\psi^{[-2]})\ri|^2 \,\d t^\ast \d r^\ast \\
    & + \half\lim_{r\to r_+}\int_{t^\ast_1}^{t^\ast_2} \ellmode(\ellmode+1) \le|w^{-1}\psi^{[-2]}\ri|^2  \,\d t^\ast \\
    = & \; \half\int_{-\infty}^{\frac{\pi}{2}}\le(w^{-1}\chi\frac{\De_+}{\De_0}\le|\Lb\le(w^{-1}\psi^{[-2]}\ri)\ri|^2 +  \frac{\De_-}{\De_0}\chi \ellmode(\ellmode+1) \le|w^{-1}\psi^{[-2]}\ri|^2 \ri) \,\d r^\ast\bigg|_{t^\ast=t^\ast_1} \\
    & +  \frac{3}{2}  \int_{t^\ast_1}^{t^\ast_2}\int_{r=3M}^{\frac{\pi}{2}}  w^{-1}\le(\chi\le(\frac{2}{r}-\frac{6M}{r^2}\ri)-\frac{1}{3}\chi'\ri)\le|\Lb(w^{-1}\psi^{[-2]})\ri|^2 \,\d t^\ast \d r^\ast \\
    & - \int_{t^\ast_1}^{t^\ast_2}\int_{-\infty}^{\frac{\pi}{2}} \half \chi' \ellmode(\ellmode+1)|w^{-1}\psi^{[-2]}|^2 \,\d t^\ast \d r^\ast \\
    & -\int_{t^\ast_1}^{t^\ast_2}\int_{-\infty}^{\frac{\pi}{2}}\Re\le(\mathfrak{F}w^{-1}\Lb(w^{-1}\psi^{[-2]})^\ast\ri)\,\d t^\ast\d r^\ast.
  \end{aligned}
\end{align}
Using the transport estimates~\eqref{est:prelimtranspsi-2w} and~\eqref{est:prelimtransaltbis-2w} to estimate the lower order terms composing $\mathfrak{F}$, we have
\begin{align}\label{est:energyidpsi-1int2}
  \begin{aligned}
    & \le|\int_{t^\ast_1}^{t^\ast_2}\int_{-\infty}^{\frac{\pi}{2}}\Re\le(\mathfrak{F}w^{-1}\Lb(w^{-1}\psi^{[-2]})^\ast\ri)\,\d t^\ast\d r^\ast\ri| \\
    \les_{M,k} & \; \de \int_{t^\ast_1}^{t^\ast_2}\int_{-\infty}^{\frac{\pi}{2}}w^{-1}\le|\Lb(w^{-1}\psi^{[-2]})^\ast\ri|^2\,\d t^\ast\d r^\ast + \de^{-1} \int_{t^\ast_1}^{t^\ast_2}\int_{-\infty}^{\frac{\pi}{2}}w^{-1}\le|\mathfrak{F}\ri|^2\,\d t^\ast\d r^\ast \\
    \les_{M,k} & \; \de \int_{t^\ast_1}^{t^\ast_2}\int_{-\infty}^{r=3M}w^{-1}\le|\Lb(w^{-1}\psi^{[-2]})^\ast\ri|^2\,\d t^\ast\d r^\ast + \int_{t^\ast_1}^{t^\ast_2}\int_{r=3M}^{\frac{\pi}{2}}w^{-1}\le|\Lb(w^{-1}\psi^{[-2]})^\ast\ri|^2\,\d t^\ast\d r^\ast\\
    & \quad +\de^{-1} \int_{t^\ast_1}^{t^\ast_2}\int_{-\infty}^{\frac{\pi}{2}}w\le(|w^{-1}\psi^{[-2]}|^2 + |w^{-2}\alt^{[-2]}|^2\ri)\,\d t^\ast\d r^\ast \\
    \les_{M,k} & \; \de \int_{t^\ast_1}^{t^\ast_2}\int_{-\infty}^{r=3M}w^{-1}\le|\Lb(w^{-1}\psi^{[-2]})^\ast\ri|^2\,\d t^\ast\d r^\ast + \int_{t^\ast_1}^{t^\ast_2}\int_{r=3M}^{\frac{\pi}{2}}w^{-1}\le|\Lb(w^{-1}\psi^{[-2]})^\ast\ri|^2\,\d t^\ast\d r^\ast\\
    & \quad + \norm{w^{-2}\alt^{[-2]}}_{L^2(\Si_{t^\ast_1})}^2 + \norm{w^{-1}\psi^{[-2]}}_{L^2(\Si_{t^\ast_1})}^2 + \int_{t^\ast_1}^{t^\ast_2}\norm{\Psi^{[-2]}}^2_{L^2(\Si_{t^\ast})}\,\d t^\ast,
  \end{aligned}
\end{align}
where $\de=\de(M,k)>0$. We have
\begin{align}\label{est:energyidpsi-1int1}
  \begin{aligned}
    & \int_{t^\ast_1}^{t^\ast_2}\int_{r=3M}^{\frac{\pi}{2}}  w^{-1}\le|\Lb(w^{-1}\psi^{[-2]})\ri|^2 \,\d t^\ast \d r^\ast \\
    & +  \frac{3}{2}  \int_{t^\ast_1}^{t^\ast_2}\int_{r=3M}^{\frac{\pi}{2}}  w^{-1}\le(\chi\le(\frac{2}{r}-\frac{6M}{r^2}\ri)-\frac{1}{3}\chi'\ri)\le|\Lb(w^{-1}\psi^{[-2]})\ri|^2 \,\d t^\ast \d r^\ast \\
    \les_{M,k} & \; \int_{t^\ast_1}^{t^\ast_2}\int_{-\infty}^{\frac{\pi}{2}} w \le(|\Lb\psi^{[-2]}|^2 + |w^{-1}\psi^{[-2]}|^2 \ri)\,\d t^\ast\d r^\ast.
  \end{aligned}
\end{align}
Plugging~\eqref{est:energyidpsi-1int2} and~\eqref{est:energyidpsi-1int1} in~\eqref{eq:energyidpsi-1int}, taking $\de$ sufficiently small to absorb the first term on the RHS of~\eqref{est:energyidpsi-1int2} into the LHS of~\eqref{eq:energyidpsi-1int}, we get
\begin{align}\label{est:energyidpsi-1int}
  \begin{aligned}
    &  \int_{-\infty}^{\frac{\pi}{2}}\le(w^{-1}\chi\frac{\De_+}{\De_0}\le|\Lb\le(w^{-1}\psi^{[-2]}\ri)\ri|^2 +  \frac{\De_-}{\De_0} \chi\ellmode(\ellmode+1) \le|w^{-1}\psi^{[-2]}\ri|^2 \ri) \,\d r^\ast\bigg|_{t^\ast=t^\ast_2} \\
    & + \frac{3}{2} \int_{t^\ast_1}^{t^\ast_2}\int_{-\infty}^{r=3M}  w^{-1}\le(\frac{6M}{r^2}-\frac{2}{r}\ri)\le|\Lb(w^{-1}\psi^{[-2]})\ri|^2 \,\d t^\ast \d r^\ast \\
    & + \lim_{r\to r_+}\int_{t^\ast_1}^{t^\ast_2} \ellmode(\ellmode+1) \le|w^{-1}\psi^{[-2]}\ri|^2  \,\d t^\ast \\
    \les_{M,k} & \; \int_{-\infty}^{\frac{\pi}{2}}\le(w^{-1}\chi\frac{\De_+}{\De_0}\le|\Lb\le(w^{-1}\psi^{[-2]}\ri)\ri|^2 +  \frac{\De_-}{\De_0} \chi\ellmode(\ellmode+1) \le|w^{-1}\psi^{[-2]}\ri|^2 \ri) \,\d r^\ast\bigg|_{t^\ast=t^\ast_1} \\
    & + \norm{w^{-2}\alt^{[-2]}}_{L^2(\Si_{t^\ast_1})}^2 + \norm{w^{-1}\psi^{[-2]}}_{L^2(\Si_{t^\ast_1})}^2 + \int_{t^\ast_1}^{t^\ast_2}\norm{\Psi^{[-2]}}^2_{L^2(\Si_{t^\ast})}\,\d t^\ast \\
    & + \int_{t^\ast_1}^{t^\ast_2}\int_{-\infty}^{\frac{\pi}{2}} w \le(|\Lb\psi^{[-2]}|^2 + |w^{-1}\psi^{[-2]}|^2 \ri)\,\d t^\ast\d r^\ast.
  \end{aligned}
\end{align}
from which (\ref{muffo}) follows.
\end{proof}

\subsection{Completing the proof of Proposition \ref{prop:TeukfromRW}} \label{sec:fipo}
We now have all the estimates in place to complete the proof of Proposition \ref{prop:TeukfromRW}. We first collect the integral estimates for $\psi^{[+2]}$ and $w^{-1}\psi^{[-2]}$ in Section \ref{sec:conclusonpsip2w-1psim2} and then the ones for $\al^{[+2]}$ and $w^{-2}\al^{[-2]}$ in Section \ref{sec:estaltpm2proofpropTeukfromRW}. We finally combine the two in Section \ref{sec:finishit} concluding the proof.

\subsubsection{Final estimates for $\psi^{[+2]}$ and $w^{-1}\psi^{[-2]}$} \label{sec:conclusonpsip2w-1psim2}
We first record that  \eqref{eq:Chandra} directly gives us
\begin{align}\label{est:Lbpsip2Lpsim2easy}
  \begin{aligned}
    \int_{-\infty}^{\frac{\pi}{2}} w^{-1}\le|\Lb\psi^{[+2]}\ri|^2 \,\d r^\ast & \les_{M,k} \norm{\Psi^{[+2]}}^2_{L^2(\Si_{t^\ast})},
  \end{aligned}
\end{align}
and
\begin{align}\label{est:Lbpsip2Lpsim2easybis}
  \begin{aligned}
    \int_{-\infty}^{\frac{\pi}{2}} w\le|L\le(w^{-1}\psi^{[-2]}\ri)\ri|^2 \,\d r^\ast & \les_{M,k} \norm{\Psi^{[-2]}}^2_{L^2(\Si_{t^\ast})} + \norm{w^{-1}\psi^{[-2]}}_{L^2(\Si_{t^\ast})}^2.
  \end{aligned}
\end{align}
We leave to the reader to check that
\begin{align*}
  \begin{aligned}
    w|L\psi|^2 + w^{-1}|\Lb\psi|^2 & \simeq_{M,k} w|\pr_{t^\ast}\psi|^2 + w^{-1}|\pr_{r^\ast}\psi|^2,
  \end{aligned}
\end{align*}
from which we infer that 
\begin{align}\label{est:compoverlinemathrmELLb}
  \begin{aligned}
    \overline{\mathrm{E}}[\psi^{[+2]}] & \simeq_{M,k} \int_{-\infty}^{\frac{\pi}{2}} \le(w|L\psi^{[+2]}|^2 + w^{-1}|\Lb\psi^{[+2]}|^2 + \ellmode^2|\psi^{[+2]}|^2 \ri)\,\d r^\ast \\
    \overline{\mathrm{E}}[w^{-1}\psi^{[-2]}] & \simeq_{M,k} \int_{-\infty}^{\frac{\pi}{2}} \le(w|L(w^{-1}\psi^{[-2]})|^2 + w^{-1}|\Lb(w^{-1}\psi^{[-2]})|^2 + \ellmode^2|w^{-1}\psi^{[-2]}|^2 \ri)\,\d r^\ast.
  \end{aligned}
\end{align}
Combining now the estimates \eqref{est:Lbpsip2Lpsim2easy},~\eqref{est:Lbpsip2Lpsim2easybis}, 
the transport estimate~\eqref{est:prelimtranspsi-2w}, the energy estimate~\eqref{est:psiintv1}, the spacetime elliptic estimates~\eqref{est:spacetimeell2psipm2} for $\psi^{[+2]}$ and $w^{-1}\psi^{[-2]}$, the estimates~\eqref{est:spacetimeLpsip2Lbpsim2} for $L\psi^{[+2]}$ and $\Lb\psi^{[-2]}$~, the estimates~\eqref{est:energyidpsi-1int} for $\Lb(w^{-1}\psi^{[-2]})$, using~\eqref{est:compoverlinemathrmELLb}, we obtain
\begin{align}\label{est:spacetimepsipm2combined}
  \begin{aligned}
    & \quad \overline{\mathrm{E}}[\psi^{[+2]}](t^\ast_2) + \int_{t^\ast_1}^{t^\ast_2}\overline{\mathrm{E}}[\psi^{[+2]}](t^\ast)\,\d t^\ast \\
    & +\overline{\mathrm{E}}[w^{-1}\psi^{[-2]}](t^\ast_2) + \int_{t^\ast_1}^{t^\ast_2}\overline{\mathrm{E}}[w^{-1}\psi^{[-2]}](t^\ast)\,\d t^\ast \\
    & + \overline{\mathrm{E}}^\HH[\psi^{[+2]}](t^\ast_2;t^\ast_1) + \overline{\mathrm{E}}^\HH[w^{-1}\psi^{[-2]}](t^\ast_2;t^\ast_1) \\
    \les_{M,k} & \; \overline{\mathrm{E}}[\psi^{[+2]}](t^\ast_1) + \overline{\mathrm{E}}[w^{-1}\psi^{[-2]}](t^\ast_1) \\
    & + \int_{t^\ast_1}^{t^\ast_2}\int_{-\infty}^{\frac{\pi}{2}} \le(w |L\psi^{[+2]}||\Psi^{[+2]}| + w |\Lb\psi^{[-2]}||\Psi^{[-2]}|\ri) \,\d t^\ast\d r^\ast\\
    & + \int_{t^\ast_1}^{t^\ast_2}\int_{-\infty}^{\frac{\pi}{2}} \le(w |L\psi^{[+2]}||\alt^{[+2]}| +w|\Lb\psi^{[-2]}||\alt^{[-2]}|\ri) \,\d t^\ast\d r^\ast \\
    & + \int_{t^\ast_1}^{t^\ast_2}\int_{-\infty}^{\frac{\pi}{2}} w |\Psi^{[+2]}||L(\psi^{[+2]})|\,\d t^\ast\d r^\ast + \int_{t^\ast_1}^{t^\ast_2}\int_{-\infty}^{\frac{\pi}{2}} |\Psi^{[-2]}||\Lb(w^{-1}\psi^{[-2]})| \,\d t^\ast\d r^\ast \\
    & + \lim_{r\to r_+} \int_{t^\ast_1}^{t^\ast_2}|\Psi^{[-2]}||w^{-1}\psi^{[-2]}|\,\d t^\ast \\
    & + \norm{\alt^{[+2]}}_{L^2(\Si_{t^\ast_1})}^2 +\norm{\psi^{[+2]}}^2_{L^2(\Si_{t^\ast_1})} + \norm{w^{-2}\alt^{[-2]}}_{L^2(\Si_{t^\ast_1})}^2 +\norm{w^{-1}\psi^{[-2]}}^2_{L^2(\Si_{t^\ast_1})} \\
    & +\norm{\Psi^{[\pm2]}}^2_{L^2(\Si_{t^\ast_1})} +\norm{\Psi^{[\pm2]}}^2_{L^2(\Si_{t^\ast_2})} + \int_{t^\ast_1}^{t^\ast_2} \norm{\Psi^{[\pm2]}}_{L^2(\Si_{t^\ast})}^2 \,\d t^\ast.
  \end{aligned}
\end{align}
By absorption on the LHS of~\eqref{est:spacetimepsipm2combined}, using the transport estimates~\eqref{est:prelimtransaltbis-2w}, \eqref{est:prelimtransalt+2bisnow} to control the $\alt$ terms, we deduce
\begin{align}\label{est:spacetimepsipm2combinedfinal}
  \begin{aligned}
    & \quad \overline{\mathrm{E}}[\psi^{[+2]}](t^\ast_2) + \int_{t^\ast_1}^{t^\ast_2}\overline{\mathrm{E}}[\psi^{[+2]}](t^\ast)\,\d t^\ast \\
    & +\overline{\mathrm{E}}[w^{-1}\psi^{[-2]}](t^\ast_2) + \int_{t^\ast_1}^{t^\ast_2}\overline{\mathrm{E}}[w^{-1}\psi^{[-2]}](t^\ast)\,\d t^\ast \\
    & + \overline{\mathrm{E}}^\HH[\psi^{[+2]}](t^\ast_2;t^\ast_1) + \overline{\mathrm{E}}^\HH[w^{-1}\psi^{[-2]}](t^\ast_2;t^\ast_1) \\
    \les_{M,k} & \; \overline{\mathrm{E}}[\psi^{[+2]}](t^\ast_1) + \overline{\mathrm{E}}[w^{-1}\psi^{[-2]}](t^\ast_1) +\norm{\alt^{[+2]}}_{L^2(\Si_{t^\ast_1})}^2 + \norm{w^{-2}\alt^{[-2]}}_{L^2(\Si_{t^\ast_1})}^2 \\
    & + \lim_{r\to r_+}\int_{t^\ast_1}^{t^\ast_2}|\Psi^{[-2]}|^2\,\d t^\ast +\norm{\Psi^{[\pm2]}}^2_{L^2(\Si_{t^\ast_1})} +\norm{\Psi^{[\pm2]}}^2_{L^2(\Si_{t^\ast_2})} + \int_{t^\ast_1}^{t^\ast_2} \norm{\Psi^{[\pm2]}}_{L^2(\Si_{t^\ast})}^2 \,\d t^\ast.
  \end{aligned}
\end{align}

\subsubsection{Final estimates for $\alt^{[+2]}$ and $w^{-2}\alt^{[-2]}$}\label{sec:estaltpm2proofpropTeukfromRW}
Arguing along the same lines as in the previous sections but with $\psi^{[+2]}$ replaced by $\alt^{[+2]}$, $\Psi^{[+2]}$ replaced by $\psi^{[+2]}$, $w^{-1}\psi^{[-2]}$ replaced by $w^{-2}\alt^{[-2]}$ and $\Psi^{[-2]}$ replaced by $w^{-1}\psi^{[-2]}$, we get 

\begin{align}\label{est:spacetimealtpm2combinedfinal}
  \begin{aligned}
    & \quad \overline{\mathrm{E}}[\alt^{[+2]}](t^\ast_2) + \int_{t^\ast_1}^{t^\ast_2}\overline{\mathrm{E}}[\alt^{[+2]}](t^\ast)\,\d t^\ast \\
    & +\overline{\mathrm{E}}[w^{-2}\alt^{[-2]}](t^\ast_2) + \int_{t^\ast_1}^{t^\ast_2}\overline{\mathrm{E}}[w^{-2}\alt^{[-2]}](t^\ast)\,\d t^\ast \\
    & + \overline{\mathrm{E}}^\HH[\alt^{[+2]}](t^\ast_2;t^\ast_1) + \overline{\mathrm{E}}^\HH[w^{-2}\alt^{[-2]}](t^\ast_2;t^\ast_1) \\
    \les_{M,k} & \; \overline{\mathrm{E}}[\alt^{[+2]}](t^\ast_1) + \overline{\mathrm{E}}[w^{-2}\alt^{[-2]}](t^\ast_1) + \lim_{r\to r_+}\int_{t^\ast_1}^{t^\ast_2}|w^{-1}\psi^{[-2]}|^2\,\d t^\ast \\
    & +\norm{\psi^{[+2]}}^2_{L^2(\Si_{t^\ast_2})} +\norm{\psi^{[+2]}}^2_{L^2(\Si_{t^\ast_1})}  + \int_{t^\ast_1}^{t^\ast_2} \norm{\psi^{[+2]}}_{L^2(\Si_{t^\ast})}^2 \,\d t^\ast \\
    & +\norm{w^{-1}\psi^{[-2]}}^2_{L^2(\Si_{t^\ast_1})} +\norm{w^{-1}\psi^{[-2]}}^2_{L^2(\Si_{t^\ast_2})} + \int_{t^\ast_1}^{t^\ast_2} \norm{w^{-1}\psi^{[-2]}}_{L^2(\Si_{t^\ast})}^2 \,\d t^\ast.
  \end{aligned}
\end{align}
Using that the Teukolsky problem (see Definition~\eqref{def:sysTeuk}) commutes with $\pr_{t^\ast}$ and $\LL^{1/2}$ angular derivatives (multiplication by $\ellmode$ in the angular projected setting), the estimate~\eqref{est:spacetimealtpm2combinedfinal} also holds with the quantities multiplies by $\pr_{t^\ast}$ and $\LL^{1/2}$, and combining this with the fact that $\mathrm{E}^{\mathfrak{T},2}[\alt] = \mathrm{E}^{\mathfrak{T}}[\pr_t\alt] + \mathrm{E}^{\mathfrak{T}}[\LL^{1/2}\alt]$, we obtain
\begin{align}\label{est:secondorderalt}
  \begin{aligned}
    & \quad \mathrm{E}^{\mathfrak{T},2}[\alt](t^\ast_2) + \int_{t^\ast_1}^{t^\ast_2}\mathrm{E}^{\mathfrak{T},2}[\alt](t^\ast)\,\d t^\ast \\
    & + \overline{\mathrm{E}}^\HH[\pr_{t^\ast}\alt^{[+2]}](t^\ast_2;t^\ast_1) + \overline{\mathrm{E}}^\HH[\LL^{1/2}\alt^{[+2]}](t^\ast_2;t^\ast_1) \\
    & + \overline{\mathrm{E}}^\HH[w^{-2}\pr_{t^\ast}\alt^{[-2]}](t^\ast_2;t^\ast_1) + \overline{\mathrm{E}}^\HH[w^{-2}\LL^{1/2}\pr_{t^\ast}\alt^{[-2]}](t^\ast_2;t^\ast_1) \\
    \les_{M,k} & \mathrm{E}^{\mathfrak{T},2}[\alt](t^\ast_1) + \underbrace{\lim_{r\to r_+}\int_{t^\ast_1}^{t^\ast_2}|w^{-1}\pr_{t^\ast}\psi^{[-2]}|^2\,\d t^\ast}_{\simeq_{M,k} \widetilde{\mathrm{E}}^\HH[w^{-1}\psi^{[-2]}](t^\ast_2;t^\ast_1)} + \underbrace{\lim_{r\to r_+}\int_{t^\ast_1}^{t^\ast_2}|w^{-1}\LL^{1/2}\psi^{[-2]}|^2\,\d t^\ast}_{\simeq_{M,k} \lim_{r\to r_+}\int_{t^\ast_1}^{t^\ast_2}\ellmode(\ellmode+1)|w^{-1}\psi^{[-2]}|^2\,\d t^\ast}\\
    & + \overline{\mathrm{E}}[\psi^{[+2]}](t^\ast_2) + \overline{\mathrm{E}}[\psi^{[+2]}](t^\ast_1)  + \int_{t^\ast_1}^{t^\ast_2} \overline{\mathrm{E}}[\psi^{[+2]}](t^\ast) \,\d t^\ast \\
    & + \overline{\mathrm{E}}[w^{-1}\psi^{[-2]}](t^\ast_2) + \overline{\mathrm{E}}[w^{-1}\psi^{[-2]}](t^\ast_1) + \int_{t^\ast_1}^{t^\ast_2} \overline{\mathrm{E}}[w^{-1}\psi^{[-2]}](t^\ast) \,\d t^\ast.
  \end{aligned}
\end{align}

\subsubsection{Combining the estimates} \label{sec:finishit}
Combining~\eqref{est:secondorderalt} with~\eqref{est:spacetimepsipm2combinedfinal}, we get
\begin{align}\label{est:secondorderaltpsifinal}
  \begin{aligned}
    & \quad \mathrm{E}^{\mathfrak{T},2}[\alt](t^\ast_2) + \int_{t^\ast_1}^{t^\ast_2}\mathrm{E}^{\mathfrak{T},2}[\alt](t^\ast)\,\d t^\ast \\
    & + \overline{\mathrm{E}}^\HH[\pr_{t^\ast}\alt^{[+2]}](t^\ast_2;t^\ast_1) + \overline{\mathrm{E}}^\HH[\LL^{1/2}\alt^{[+2]}](t^\ast_2;t^\ast_1) \\
    & + \overline{\mathrm{E}}^\HH[w^{-2}\pr_{t^\ast}\alt^{[-2]}](t^\ast_2;t^\ast_1) + \overline{\mathrm{E}}^\HH[w^{-2}\LL^{1/2}\pr_{t^\ast}\alt^{[-2]}](t^\ast_2;t^\ast_1) \\
    & + \overline{\mathrm{E}}[\psi^{[+2]}](t^\ast_2) + \int_{t^\ast_1}^{t^\ast_2}\overline{\mathrm{E}}[\psi^{[+2]}](t^\ast)\,\d t^\ast \\
    & +\overline{\mathrm{E}}[w^{-1}\psi^{[-2]}](t^\ast_2) + \int_{t^\ast_1}^{t^\ast_2}\overline{\mathrm{E}}[w^{-1}\psi^{[-2]}](t^\ast)\,\d t^\ast \\
    & + \overline{\mathrm{E}}^\HH[\psi^{[+2]}](t^\ast_2;t^\ast_1) + \overline{\mathrm{E}}^\HH[w^{-1}\psi^{[-2]}](t^\ast_2;t^\ast_1) \\
    \les_{M,k} & \;  \mathrm{E}^{\mathfrak{T},2}[\alt](t^\ast_1) + \overline{\mathrm{E}}[\psi^{[+2]}](t^\ast_1) + \overline{\mathrm{E}}[w^{-1}\psi^{[-2]}](t^\ast_1) + \lim_{r\to r_+}\int_{t^\ast_1}^{t^\ast_2}|\Psi^{[-2]}|^2\,\d t^\ast \\
    & +\norm{\Psi^{[\pm2]}}^2_{L^2(\Si_{t^\ast_1})} +\norm{\Psi^{[\pm2]}}^2_{L^2(\Si_{t^\ast_2})} + \int_{t^\ast_1}^{t^\ast_2} \norm{\Psi^{[\pm2]}}_{L^2(\Si_{t^\ast})}^2 \,\d t^\ast.
  \end{aligned}
\end{align}
Using that the Teukolsky equation commutes with $\pr_{t^\ast}$ and $\LL^{1/2}$ derivatives, the desired estimate~\eqref{est:TeukfromRW} follows from~\eqref{est:secondorderaltpsifinal} by further commutation. This finishes the proof of Proposition~\ref{prop:TeukfromRW}.

\subsection{Proof of Theorem~\ref{thm:Teukboundednessdecay}}\label{sec:proofthmTeukbdddecay}
Relations~\eqref{eq:defPsiDR} rewrite as
\begin{align}\label{eq:defPsiDRinv}
  \begin{aligned}
    \Psi^{[+2]} & = +\half\Psi^D+\half\le(\Psi^R-12M\LL^{-1}(\LL-2)^{-1} \pr_t\Psi^D\ri),\\
    \big(\Psi^{[-2]}\big)^\ast & = -\half\Psi^D+\half\le(\Psi^R-12M\LL^{-1}(\LL-2)^{-1}\pr_t\Psi^D\ri).
  \end{aligned}
\end{align}
Using~\eqref{eq:defPsiDR} and~\eqref{eq:defPsiDRinv}, we have
\begin{align}\label{est:compPsipm2PsiDRv1}
  \begin{aligned}
    & \mathrm{E}^{\mathfrak{R}}[\Psi^{[\pm2]}] + \mathrm{E}^{\mathfrak{R}}\le[\LL^{-1}(\LL-2)^{-1}\pr_{t^\ast}\le(\Psi^{[+2]}-(\Psi^{[-2]})^\ast\ri)\ri] \\
    \simeq_{M,k} & \; \mathrm{E}^{\mathfrak{R}}[\Psi^{D}] + \mathrm{E}^{\mathfrak{R}}\le[\Psi^{R}\ri] + \mathrm{E}^{\mathfrak{R}}\le[\LL^{-1}(\LL-2)^{-1}\pr_{t^\ast}\Psi^{D}\ri],
  \end{aligned}
\end{align}
with analogous relations for $\overline{\mathrm{E}}^\HH$. Hence, since $\Psi^D,\Psi^{R},\LL^{-1}(\LL-2)^{-1}\pr_{t^\ast}\Psi^{D}$ are solutions to the Dirichlet or Robin Regge-Wheeler problem~\eqref{sys:RWdeco}, the boundedness and decay result of Theorems~\ref{thm:mainRW1a} and~\ref{thm:mainRW1b} for $\Psi^D,\Psi^{R},\LL^{-1}(\LL-2)^{-1}\pr_{t^\ast}\Psi^{D}$ gives
\begin{align}\label{est:boundednessPsipm2}
  \begin{aligned}
    & \mathrm{E}^{\mathfrak{R}}[\Psi^{[\pm2]}](t^\ast_2) + \mathrm{E}^{\mathfrak{R}}\le[\LL^{-1}(\LL-2)^{-1}\pr_{t^\ast}\Psi^D\ri](t^\ast_2) \\
    \les_{M,k} & \; \mathrm{E}^{\mathfrak{R}}[\Psi^{[\pm2]}](t^\ast_1) +\mathrm{E}^{\mathfrak{R}}\le[\LL^{-1}(\LL-2)^{-1}\pr_{t^\ast}\Psi^D\ri](t^\ast_1),
  \end{aligned}
\end{align}
and
\begin{align}\label{est:fullestimatesPsipm2}
  \begin{aligned}
    & \quad \mathrm{E}^{\mathfrak{R}}_{m\ellmode}[\Psi^{[\pm2]}](t^\ast_2) + \int_{t^\ast_1}^{t^\ast_2} \mathrm{E}^{\mathfrak{R}}_{m\ellmode}[\Psi^{[\pm2]}](t^\ast)\,\d t^\ast + \overline{\mathrm{E}}^\HH_{m\ellmode}[\Psi^{[\pm2]}](t^\ast_2;t^\ast_1) + \overline{\mathrm{E}}^\II_{m\ellmode}[\Psi^{[\pm2]}](t^\ast_2;t^\ast_1) \\
    & + \mathrm{E}^{\mathfrak{R}}_{m\ellmode}\le[\LL^{-1}(\LL-2)^{-1}\pr_{t^\ast}\Psi^D\ri](t^\ast_2) + \int_{t^\ast_1}^{t^\ast_2} \mathrm{E}^{\mathfrak{R}}_{m\ellmode}\le[\LL^{-1}(\LL-2)^{-1}\pr_{t^\ast}\Psi^D\ri](t^\ast)\,\d t^\ast  \\
    & + \overline{\mathrm{E}}^\HH_{m\ellmode}\le[\LL^{-1}(\LL-2)^{-1}\pr_{t^\ast}\Psi^D\ri](t^\ast_2;t^\ast_1) + \overline{\mathrm{E}}^\II_{m\ellmode}\le[\LL^{-1}(\LL-2)^{-1}\pr_{t^\ast}\Psi^D\ri](t^\ast_2;t^\ast_1) \\
    \leq & \; e^{C\ellmode^\mathfrak{p}}\bigg(\mathrm{E}^{\mathfrak{R}}_{m\ellmode}[\Psi^{[\pm2]}](t^\ast_1) + \mathrm{E}^{\mathfrak{R}}_{m\ellmode}\le[\LL^{-1}(\LL-2)^{-1}\pr_{t^\ast}\Psi^D\ri](t^\ast_1)\bigg),
  \end{aligned}
\end{align}
for all $t^\ast_2\geq t^\ast_1$. Combining~\eqref{est:TeukfromRW} and the boundedness of $\Psi^{[\pm2]},\LL^{-1}(\LL-2)^{-1}\pr_t\Psi^D$ of~\eqref{est:boundednessPsipm2}, and using the redshift estimate~\eqref{est:redshiftmain} to estimate the flux $\overline{\mathrm{E}}^{\HH,n-2}[\Psi^{[\pm2]}](t^\ast_2;t^\ast_1)$, we have
\begin{align*}
  \begin{aligned}
    & \quad \mathrm{E}^{\mathfrak{T},\mathfrak{R},n}[\alt](t^\ast_2) + \int_{t^\ast_1}^{t^\ast_2} \mathrm{E}^{\mathfrak{T},\mathfrak{R},n}[\alt](t^\ast)\,\d t^\ast \\
    \les_{M,k,n} & \; \mathrm{E}^{\mathfrak{T},\mathfrak{R},n}[\alt](t^\ast_1) + \int_{t^\ast_1}^{t^\ast_2}\le(\mathrm{E}^{\mathfrak{R},n-2}[\Psi^{[\pm2]}](t^\ast)+\mathrm{E}^{\mathfrak{R},n-2}\le[\LL^{-1}(\LL-2)^{-1}\pr_{t^\ast}\Psi^D\ri](t^\ast)\ri)\,\d t^{\ast} \\
    \les_{M,k,n} & \; \mathrm{E}^{\mathfrak{T},\mathfrak{R},n}[\alt](t^\ast_1) + \le(t^\ast_2-t^\ast_1\ri)\le(\mathrm{E}^{\mathfrak{R},n-2}[\Psi^{[\pm2]}](t^\ast_1)+\mathrm{E}^{\mathfrak{R},n-2}\le[\LL^{-1}(\LL-2)^{-1}\pr_{t^\ast}\Psi^D\ri](t^\ast_1)\ri),
  \end{aligned}
\end{align*}
for all $t^\ast_2\geq t^\ast_1$. Estimate~\eqref{est:boundednessalt} then follows from a classical pigeonhole argument. The integral decay estimate~\eqref{est:Teukfullestimates} follows directly from a combination of~\eqref{est:TeukfromRW} with~\eqref{est:fullestimatesPsipm2}. This finishes the proof of Theorem~\ref{thm:Teukboundednessdecay}.

\subsection{Proof of Theorem~\ref{thm:main1}}\label{sec:proofexpodecay}
The proof of Theorem~\ref{thm:main1} follows from the estimates of Theorem~\ref{thm:Teukboundednessdecay} for $\mathrm{E}^{\mathfrak{T},\mathfrak{R},n}_{m\ellmode}[\alt]$, taking $e_{m\ellmode}=\mathrm{E}^{\mathfrak{T},\mathfrak{R},n}_{m\ellmode}[\alt]$ and $\mathfrak{p}=1$ in the following general calculus lemma.
\begin{lemma}\label{lem:interpolation}
  Let $\le(e_{m\ellmode}(t^\ast)\ri)_{\ellmode\geq2,|m|\leq\ellmode}$ be a sequence of non-negative real functions such that $\sum_{m\ellmode} \ellmode^2e_{m\ellmode}(t^\ast)<\infty$. Assume that there exists $\mathfrak{p}>0$ and a constant $C>0$ such that 
  \begin{align}
    \label{est:Einebound}
    e_{m\ellmode}(t^\ast_2) & \leq Ce_{m\ellmode}(t^\ast_1), 
  \end{align}
  and
  \begin{align}\label{est:Eineqdiff}
    e_{m\ellmode}(t^\ast_2) + \int_{t^\ast_1}^{t^\ast_2} e_{m\ellmode}(t^\ast)\,\d t^\ast & \leq e^{C\ellmode^\mathfrak{p}} e_{m\ellmode}(t^\ast_1),
  \end{align}
  for all $t^\ast_2\geq t^\ast_1$. Then, for all $\ellmode\geq2, |m|\leq\ellmode$,
  \begin{align}\label{est:expodecayaltlemma}
    e_{m\ellmode}(t^\ast) & \leq 10C \, \exp\le(e^{-C\ellmode^\mathfrak{p}}(t^\ast_0-t^\ast)\ri)e_{m\ellmode}(t^\ast_0),
  \end{align}
  for all $t^\ast\geq t^\ast_0$. Moreover, 
  \begin{align}\label{est:logdecayaltlemma}
    \sum_{\ellmode\geq 2}\sum_{|m|\leq\ellmode} e_{m\ellmode}(t^\ast) & \les_{C,\mathfrak{p}} \le(\log(t^\ast-t^\ast_0)\ri)^{-\frac{2}{\mathfrak{p}}}\le(\sum_{\ellmode\geq 2}\sum_{|m|\leq\ellmode} \ellmode^2e_{m\ellmode}(t^\ast_0)\ri).
  \end{align}
\end{lemma}
\begin{proof}
  From~\eqref{est:Eineqdiff} we infer in particular that $e_{m\ellmode}$ is integrable in time, and we define
  \begin{align*}
    f(t^\ast) & := \exp\le(e^{-C\ellmode^\mathfrak{p}}t^\ast\ri)\int_{t^\ast}^{+\infty} e_{m\ellmode}(t^{\ast,'}) \,\d t^{\ast,'}.
  \end{align*}
  Using~\eqref{est:Eineqdiff}, we have
  \begin{align*}
    f'(t^\ast)\exp\le(-e^{-C\ellmode^\mathfrak{p}}t^\ast\ri) & = e^{-C\ellmode^\mathfrak{p}}\int_{t^\ast}^{+\infty} e_{m\ellmode}(t^{\ast,'}) \,\d t^{\ast,'} -e_{m\ellmode}(t^{\ast}) \leq 0.
  \end{align*}
  Thus, $f$ decreases and
  \begin{align}\label{est:pfcorCarl2}
    \int_{t^\ast}^{+\infty} e_{m\ellmode}(t^{\ast,'}) \,\d t^{\ast,'} & \leq \exp\le(-e^{-C\ellmode^\mathfrak{p}}(t^\ast-t^\ast_0)\ri) e^{C\ellmode^\mathfrak{p}} e_{m\ellmode}(t^\ast_0).
  \end{align}
  Let us define $t^\ast_L := L e^{C\ellmode^\mathfrak{p}}+t^\ast_0$ for some $L\in\mathbb{N}$. From the mean value theorem applied to (\ref{est:Eineqdiff}), for any $L\in\mathbb{N}$, there exists $\widetilde{t}^\ast_L\in(t^\ast_L,t^\ast_{L+1})$ such that
  \begin{align*}
    e_{m\ellmode}(\widetilde{t}^\ast_L) & = e^{-C\ellmode^\mathfrak{p}} \int_{t^\ast_L}^{t^\ast_{L+1}}e_{m\ellmode}(t^\ast)\,\d t^\ast.
  \end{align*}
  Using~\eqref{est:pfcorCarl2}, we have
  \begin{align*}
    e_{m\ellmode}(\widetilde{t}^\ast_L) & \leq \exp\le(-L\ri)e_{m\ellmode}(t^\ast_0).
  \end{align*}
  Using~\eqref{est:Einebound}, we infer from the above estimate that, for any $L\in\mathbb{N}$ and all $t^\ast\in(t^\ast_{L+1},t^\ast_{L+2})$, we have
  \begin{align*}
    e_{m\ellmode}(t^\ast) & \leq C \exp\le(-L\ri)e_{m\ellmode}(t^\ast_0) = Ce^2 \exp(-(L+2))e_{m\ellmode}(t^\ast_0) \leq Ce^2\exp\le(-e^{-C\ellmode^\mathfrak{p}}(t^\ast-t^\ast_0)\ri)e_{m\ellmode}(t^\ast_0),
  \end{align*}
  which is the desired exponential decay estimate~\eqref{est:expodecayaltlemma}.\\
  
  Summing the exponential decay estimates~\eqref{est:expodecayaltlemma} in $m,\ellmode$, we have for any $\ell_{max}\geq 2$ and any $t^\ast\geq t^\ast_0$
  \begin{align*}
    \sum_{\ellmode\geq 2}\sum_{|m|\leq\ellmode} e_{m\ellmode}(t^\ast) & \leq 10C \sum_{\ellmode\geq 2} \sum_{|m|\leq\ellmode}e_{m\ellmode}(t^\ast_0)\exp\le(e^{-C\ellmode^\mathfrak{p}}(t^\ast_0-t^\ast)\ri) \\
                                                        & \leq 10C \Big(\sum_{2\leq \ellmode \leq \ell_{max}}\sum_{|m|\leq\ellmode}e_{m\ellmode}(t^\ast_0)\Big)\exp\le(e^{-C \ell_{max}^\mathfrak{p}}(t^\ast_0-t^\ast)\ri) + \frac{10C}{\ell_{max}^2} \Big(\sum_{\ellmode\geq \ell_{max}}\sum_{|m|\leq\ellmode}\ellmode^2e_{m\ellmode}(t^\ast_0)\Big) \, .
  \end{align*}
Fixing now $C \ell_{max}^{\mathfrak{p}} = \log(t^\ast-t^\ast_0)/2$, we infer
  \begin{align*}
    \sum_{\ellmode\geq 2}\sum_{|m|\leq\ellmode} e_{m\ellmode}(t^\ast) & \leq  10C \le(\sum_{\ellmode\geq 2}\sum_{|m|\leq\ellmode}\ellmode^2e_{m\ellmode}(t^\ast_0)\ri)\le(\exp\le(-(t^\ast-t^\ast_0)^{1/2}\ri) + (2{C})^{\frac{2}{\mathfrak{p}}}(\log(t^\ast-t^\ast_0))^{-\frac{2}{\mathfrak{p}}}\ri), \\
                                                        & \les_{C,\mathfrak{p}} \le(\sum_{\ellmode\geq 2}\sum_{|m|\leq\ellmode} \ellmode^2e_{m\ellmode}(t^\ast_0)\ri)\le(\log(t^\ast-t^\ast_0)\ri)^{-\frac{2}{\mathfrak{p}}},
  \end{align*}
  which finishes the proof of~\eqref{est:logdecayaltlemma} and of the lemma.
\end{proof}

\appendix

\section{Asymptotics of the potential in the Carleman estimates}\label{sec:Carlemanmultipliers}
In this section we prove the estimates that are used in Section~\ref{sec:Carlemanestimatesfirst}. Recall that $p=e^f$, with $f=\frac{\Ka}{r}$. Using that $f' = -\Ka w$, and that $w' = w \le(-\frac{2}{r}+\frac{6M}{r^2}\ri)$, we have
\begin{align*}
  \begin{aligned}
    (pV)' & = -\le(\frac{2}{r}-\frac{6M}{r^2}\ri)\le(\ellmode(\ellmode+1) -\frac{6M}{r}\ri)w e^f  -\le(\Ka\le(\ellmode(\ellmode+1) -\frac{6M}{r}\ri)-6M\ri)  w^2 e^f.
  \end{aligned}
\end{align*}
and
\begin{align*}
  \begin{aligned}
    -pV' - \half p''' & = w\le(\frac{2}{r}-\frac{6M}{r^2}\ri)\le(\ellmode(\ellmode+1)-\frac{6M}{r}\ri)e^f - 6Mw^2e^f + \half\le(-f''' - 2f''f' -(f')^3\ri)e^f \\
    & = w\le(\frac{2}{r}-\frac{6M}{r^2}\ri)\le(\ellmode(\ellmode+1)-\frac{6M}{r}\ri)e^f \\
    & \quad + \half\le(\Ka w \le(\frac{2}{r}-\frac{6M}{r^2}\ri)^2 + \Ka w^2 \le(2-\frac{12M}{r}\ri)  +2\Ka^2w^2\le(\frac{2}{r}-\frac{6M}{r^2}\ri)+\Ka^3w^3\ri)e^f \\
    & = w\le(\frac{2}{r}-\frac{6M}{r^2}\ri)\le(\ellmode(\ellmode+1)-\frac{6M}{r}\ri)e^f + \half\Ka w^2 \le(2-\frac{12M}{r}\ri)e^f \\
                      & \quad  + \half \Ka w\le(\Ka w+\le(\frac{2}{r}-\frac{6M}{r^2}\ri) \ri)^2e^f.
  \end{aligned}
\end{align*}
Combining the above, we have the following expression for the potential in~\eqref{est:easyCarleman}
\begin{align}\label{eq:easypotential}
  \begin{aligned}
    \frac{\le(2(pV)'+\half\le(-pV' - \half p'''\ri)\ri)}{we^f} & = -\frac{3}{2}\le(\frac{2}{r}-\frac{6M}{r^2}\ri)\le(\ellmode(\ellmode+1) -\frac{6M}{r}\ri)  -2\Ka w\le(\ellmode(\ellmode+1) -\quar -\frac{9M}{2r}\ri) \\
    & \quad +12Mw + \quar \Ka \le(\Ka w+\le(\frac{2}{r}-\frac{6M}{r^2}\ri) \ri)^2 \\
    & =: \VV.
  \end{aligned}
\end{align}

\subsection{Proof of Lemma~\ref{lem:easypotentialsigns}}\label{sec:prooflem:easypotentialsigns}
Let us fix a $r_m\in(r_+,+\infty)$ (depending only on $M,k$) such that
\begin{align}\label{eq:prooflemrK2rc}
  \pr_r^2\le(w^2\ri) = 2(\pr_rw)^2+2w\pr_{r}^2w & \geq \le(\frac{1}{r_+^2}\le(\frac{2}{r_+}-\frac{6M}{r_+^2}\ri)\ri)^2, & \text{for all $r\leq r_m$.}
\end{align}
Provided that $\Ka^{-1}\ellmode^2\leq 1$, we have
\begin{align*}
  \le|\frac{\d^2 \VV}{\d r^2} - \quar\Ka^3\pr_r^2\le(w^2\ri) \ri|\les_{M,k} \Ka^2. 
\end{align*}
Hence, using~\eqref{eq:prooflemrK2rc}, and provided that $\Ka^{-1}\ellmode^2$ is sufficiently small depending only on $M,k$, we have
\begin{align*}
  \frac{\d^2 \VV}{\d r^2}(r) > \frac{\Ka^3}{8}\le(\frac{1}{r_+^2}\le(\frac{2}{r_+}-\frac{6M}{r_+^2}\ri)\ri)^2>0, 
\end{align*}
uniformly for all $r\in (r_+,r_{m})$. Hence $\VV$ is convex on $(r_+,r_m)$. Provided that $\Ka^{-1}\ellmode^2$ is sufficiently small depending only on $M,k$, we have
\begin{align*}
  \VV(r_+) & = -\frac{3}{2}\le(\frac{2}{r_+}-\frac{6M}{r_+^2}\ri)\le(\ellmode(\ellmode+1)-\frac{6M}{r_+}\ri) + \quar \Ka \le(\frac{2}{r_+}-\frac{6M}{r_+^2}\ri)^2 > 0.
\end{align*}
Moreover, provided that $\Ka^{-1}\ellmode^2$ is sufficiently small depending only on $M,k$, we also have 
\begin{align*}
  \VV(r) & > 0, & \text{for all $r\geq r_m$.}
\end{align*}
Recall that $r_{\Ka,0} = r_++\Ka^{-1}r_+^2$. We have the following Taylor expansion for $w$ 
\begin{align}\label{eq:Taylorwrrp}
  w(r) = \frac{1}{r_+^2}\le(\frac{2}{r_{+}}-\frac{6M}{r^2_{+}}\ri)(r_+-r) + O_{r\to r_+}\le((r-r_+)^2\ri).
\end{align}
Using~\eqref{eq:Taylorwrrp}, we have
\begin{align*}
  \le|\Ka w(r_{\Ka,0})+\le(\frac{2}{r_{\Ka,0}}-\frac{6M}{r^2_{\Ka,0}}\ri) \ri| & \les_{M,k} \Ka^{-1}.
\end{align*}
Plugging the above in~\eqref{eq:easypotential}, we have
\begin{align*}
  \bigg|\VV(r_{\Ka,0}) - \underbrace{\le(-\frac{3}{2}\le(\frac{2}{r_+}-\frac{6M}{r^2_+}\ri)\le(\ellmode(\ellmode+1)-\frac{6M}{r_+}\ri) +2\le(\frac{2}{r_+}-\frac{6M}{r^2_+}\ri)\le(\ellmode(\ellmode+1)-\quar-\frac{9M}{2r_+}\ri)\ri)}_{=\half\big(\frac{2}{r_+}-\frac{6M}{r^2_+}\big)\le(\ellmode(\ellmode+1)-\half\ri) < 0} \bigg| \les_{M,k} \Ka^{-1}\ellmode^2, 
\end{align*}
and $\VV(r_{\Ka,0})<0$ provided that $\Ka^{-1}\ellmode^2$ is sufficiently small depending only on $M,k$. Hence, there exists
\begin{align*}
  r_+ < r_{\Ka,-1} < r_{\Ka,0} < r_{\Ka,+1}
\end{align*}
such that~\eqref{eq:defrK0rK2insec} holds. Let $r_{\Ka,3} = r_++2\Ka^{-1}r_+^2$. Using~\eqref{eq:Taylorwrrp}, we have
\begin{align*}
  \le|\Ka w(r_{\Ka,3}) + 2\le(\frac{2}{r_+}-\frac{6M}{r_+^2}\ri)\ri| \les_{M,k}\Ka^{-1}
\end{align*}
Plugging the above in~\eqref{eq:easypotential}, we have
\begin{align*}
  \le|\VV(r_{\Ka,3}) - \quar\Ka\le(\frac{2}{r_+}-\frac{6M}{r_+^2}\ri)^2\ri| & \les_{M,k} \ellmode^2,
\end{align*}
hence, for $\Ka^{-1}\ellmode^2$ sufficiently small depending only on $M,k$, we have $\VV(r_{\Ka,3})>0$. Thus, $r_{\Ka,\pm1} < r_{\Ka,3}$ and $|r_{\Ka,\pm1}-r_+| \les_{M,k} \Ka^{-1}$. Hence $\Ka w(r_{\Ka,\pm1}) \les_{M,k}1$, which plugged in~\eqref{eq:easypotential}, using that $\VV(r_{\Ka,\pm1})=0$, implies that
\begin{align}\label{est:KawrKapm1proof}
  \le|\Ka w (r_{\Ka,\pm1}) + \le(\frac{2}{r_{\Ka,\pm1}}-\frac{6M}{r^2_{\Ka,\pm1}}\ri)\ri| & \les_{M,k} \Ka^{-1/2}\ellmode.
\end{align}
Re-plugging~\eqref{est:KawrKapm1proof} in~\eqref{eq:easypotential}, using that $\VV(r_{\Ka,\pm1})=0$, we deduce that 
\begin{align*}
  \le|\half \Ka \le(\Ka w(r_{\Ka,\pm1})+\le(\frac{2}{r_{\Ka,\pm1}}-\frac{6M}{r_{\Ka,\pm1}^2}\ri) \ri)^2 +\half \le(\frac{2}{r_+}-\frac{6M}{r^2_+}\ri)\le(\ellmode(\ellmode+1)-\half\ri) \ri| \les_{M,k} \Ka^{-1/2}\ellmode^3,
\end{align*}
provided that $\Ka^{-1/2}\ellmode$ is sufficiently small depending on $M,k$. Hence
\begin{align}\label{eq:KawarKa+-==}
  \le|\Ka w(r_{\Ka,\pm1}) +\le(\frac{2}{r_+}-\frac{6M}{r_+^2}\ri) \pm\Ka^{-1/2}\sqrt{\le(\frac{6M}{r_+^2}-\frac{2}{r_+}\ri)\le(\ellmode(\ellmode+1)-\half\ri)} \ri| & \les_{M,k} \Ka^{-1}\ellmode^2,
\end{align}
provided that $\Ka^{-1/2}\ellmode$ is sufficiently small depending on $M,k$. Since $w$ is increasing on $(r_+,3M)$, we have $\Ka w(r_{\Ka,-1}) < \Ka w (r_{\Ka,0}) < \Ka w(r_{\Ka,+1})$, hence, from~\eqref{eq:KawarKa+-==}, we deduce that
\begin{align*}
  \le|\Ka w(r_{\Ka,-1}) +\le(\frac{2}{r_+}-\frac{6M}{r_+^2}\ri) +\Ka^{-1/2}\sqrt{\le(\frac{6M}{r_+^2}-\frac{2}{r_+}\ri)\le(\ellmode(\ellmode+1)-\half\ri)} \ri| & \les_{M,k} \Ka^{-1}\ellmode^2,\\
  \le|\Ka w(r_{\Ka,+1}) +\le(\frac{2}{r_+}-\frac{6M}{r_+^2}\ri) -\Ka^{-1/2}\sqrt{\le(\frac{6M}{r_+^2}-\frac{2}{r_+}\ri)\le(\ellmode(\ellmode+1)-\half\ri)} \ri| & \les_{M,k} \Ka^{-1}\ellmode^2.
\end{align*}
Using~\eqref{eq:Taylorwrrp}, we infer~\eqref{eq:DLrK0rK2insec} from the above and this finishes the proof of Lemma~\ref{lem:easypotentialsigns}.

\subsection{Proof of Lemma~\ref{lem:easyboundspotential}}
From the proof of Lemma~\ref{lem:easypotentialsigns} in Section~\ref{sec:prooflem:easypotentialsigns}, we recall that there exists $r_m>r_+$ -- depending only on $M,k$ -- such that, for $\Ka^{-1}\ellmode^2$ sufficiently small, $\VV$ is convex on $(r_+,r_m)$, positive at $r_+$ and $r_m$, and has two simple zeros $r_{\Ka,-1}$ and $r_{\Ka,+1}$ in $(r_+,r_m)$. Hence, using that by definition $r_{\Ka,\pm2} = r_+ + \Ka^{-1}r_+^2 \pm \Ka^{-1}r_+^2\le(\Ka^{-1}\ellmode^2\ri)^{1/4}$ and the estimates~\eqref{eq:DLrK0rK2insec} from Lemma~\ref{lem:easypotentialsigns}, we have
\begin{itemize}
\item $\VV$ is decreasing on $(r_+,r_{\Ka,-2})$,
\item $\VV$ is increasing on $(r_{\Ka,+2},r_m)$.
\end{itemize}
Thus,
\begin{align}\label{est:BgeqBrKapm2}
  \begin{aligned}
    \VV & \geq \VV(r_{\Ka,-2}), && \text{for $r_+<r<r_{\Ka,-2}$,}\\
    \VV & \geq \VV(r_{\Ka,+2}), && \text{for $r_{\Ka,+2}<r<r_m$.}
  \end{aligned}
\end{align}
Using~\eqref{eq:Taylorwrrp} and the definition of $r_{\Ka,\pm2}$, we have
\begin{align}\label{est:DLKawrKapm2}
  \begin{aligned}
    \le|\Ka w(r_{\Ka,\pm2}) + \le(\frac{2}{r_+}-\frac{6M}{r_+^2}\ri)\le(1 \pm \le(\Ka^{-1}\ellmode^2\ri)^{1/4}\ri) \ri| & \les_{M,k} \Ka^{-1},
  \end{aligned}
\end{align}
provided that $\Ka^{-1}\ellmode^2$ is sufficiently small depending on $M,k$. Using~\eqref{est:DLKawrKapm2} in~\eqref{eq:easypotential}, we get
\begin{align*}
  \begin{aligned}
    \VV\le(r_{\Ka,\pm2}\ri) & = \underbrace{-\frac{3}{2}\le(\frac{2}{r_{\Ka,\pm2}}-\frac{6M}{r_{\Ka,\pm2}^2}\ri)\le(\ellmode(\ellmode+1)-\frac{6M}{r_{\Ka,\pm2}}\ri)}_{\les_{M,k}\ellmode^2}  \\
    & \quad +\underbrace{12Mw(r_{\Ka,\pm2})}_{\les_{M,k}\Ka^{-1}} - \underbrace{2\Ka w(r_{\Ka,\pm2})\le(\ellmode(\ellmode+1) -\quar -\frac{9M}{2r_{\Ka,\pm2}}\ri)}_{\les_{M,k}\ellmode^2} \\
    & \quad + \quar \Ka \Bigg(\underbrace{\Ka w(r_{\Ka,\pm2})+\le(\frac{2}{r_{\Ka,\pm2}}-\frac{6M}{r_{\Ka,\pm2}^2}\ri)}_{=\underbrace{\Ka w(r_{\Ka,\pm2})+\le(\frac{2}{r_+}-\frac{6M}{r_+^2}\ri)+\les_{M,k}\Ka^{-1}}_{=\pm \le(\frac{2}{r_+}-\frac{6M}{r_+^2}\ri)\le(\Ka^{-1}\ellmode^2\ri)^{1/4}+\les_{M,k}\Ka^{-1}}}\Bigg)^2.
  \end{aligned}
\end{align*}
Hence,
\begin{align*}
  \begin{aligned}
    \le|\Ka^{-1}\VV\le(r_{\Ka,\pm2}\ri) -\half\le(\frac{2}{r_+}-\frac{6M}{r_+^2}\ri)^2 \le(\Ka^{-1}\ellmode^2 \ri)^{1/2} \ri| \les_{M,k} \Ka^{-1}\ellmode^2,
  \end{aligned}
\end{align*}
from which we infer
\begin{align}\label{est:BrKapm2}
  \begin{aligned}
    \VV\le(r_{\Ka,\pm2}\ri) & \gtrsim_{M,k} \frac{\Ka}{4}\le(\frac{2}{r_+}-\frac{6M}{r_+^2}\ri)^2 \le(\Ka^{-1}\ellmode^2 \ri)^{1/2},
  \end{aligned}
\end{align}
provided that $\Ka^{-1}\ellmode^2$ is sufficiently small depending on $M,k$. Combining~\eqref{eq:easypotential},~\eqref{est:BgeqBrKapm2} and~\eqref{est:BrKapm2}, we get
\begin{align*}
  \begin{aligned}
    \le((pV)'-\half pV' - \quar p'''\ri) & = \VV we^f \gtrsim_{M,k} \Ka\le(\Ka^{-1}\ellmode^2\ri)^{1/2}we^f \gtrsim_{M,k} \ellmode^2we^f,
  \end{aligned}
\end{align*}
for $r_+<r<r_{\Ka,-2}$. This finishes the proof of~\eqref{est:Br+rKa-2insec}. Let $\eta>0$ a constant of $M,k$ to be determined. Let us define $r_{\Ka,\eta^{-1}} = r_++ \eta^{-1}\Ka^{-1}r_+^2$. Combining~\eqref{eq:easypotential},~\eqref{est:BgeqBrKapm2} and~\eqref{est:BrKapm2}, we get
\begin{align}\label{est:BwefrKap2rm}
  \begin{aligned}
    \le((pV)'-\half pV' - \quar p'''\ri) & = \VV we^f \gtrsim_{M,k} \Ka\le(\Ka^{-1}\ellmode^2\ri)^{1/2}we^f \gtrsim_{M,k}\le(\Ka^{-1}\ellmode^2\ri)^{1/2}(\Ka w)^2e^f,
  \end{aligned}
\end{align}
for $r_{\Ka,+2}<r<r_{\Ka,\eta^{-1}}$ (we have $r_{\Ka,\eta^{-1}}<r_m$ for $\Ka$ sufficiently large depending on $\eta$), and where, in the last inequality, we used that $\Ka w\les_{M,k}1$ for $r_{\Ka,+2}<r<r_{\Ka,\eta^{-1}}$.\footnote{This last bound $\Ka w\les_{M,k}1$ cannot hold in the full range $(r_{\Ka,+2},r_m)$ which is the reason why we introduced $r_{\Ka,\eta^{-1}}$.} For all $r_+<r<r_m$, we have
\begin{align}\label{eq:upperboundlogwr+rm}
  \le|\frac{\pr_rw}{w}\ri| & \leq \frac{c}{r-r_+},
\end{align}
where $c=c(M,k)>0$ is a constant depending only on $M,k$. On $r_{\Ka,\eta^{-1}} < r < r_m$, using that $\pr_rw = -\frac{1}{r^2}\le(\frac{2}{r}-\frac{6M}{r^2}\ri)$ and the bound~\eqref{eq:upperboundlogwr+rm}, we have
\begin{align*}
  \begin{aligned}
    \VV(r) & = -\frac{3}{2}\le(\frac{2}{r}-\frac{6M}{r^2}\ri)\le(\ellmode(\ellmode+1)-\frac{6M}{r}\ri) - 2\Ka w \le(\ellmode(\ellmode+1) - \quar - \frac{9M}{2r}\ri) + 12Mw+ \quar \Ka \le(\Ka w+\le(\frac{2}{r}-\frac{6M}{r^2}\ri) \ri)^2 \\
    & \geq -C\ellmode^2 - C\ellmode^2(\Ka w) + \quar \Ka (\Ka w)^2 \le(1 - \le|\frac{r^2\pr_rw}{\Ka w}\ri|\ri)^2 \\
    & \geq -C\ellmode^2 - C\ellmode^2(\Ka w) + \quar \Ka (\Ka w)^2 \le(1 - c \frac{\Ka^{-1} r^{2}}{r-r_+}\ri)^2 \\
    & \geq -C\ellmode^2 - C\ellmode^2(\Ka w) + \quar \Ka (\Ka w)^2 \le(1 - c\eta\frac{r_m^2}{r_+^2}\ri)^2,
  \end{aligned}
\end{align*}
where $C=C(M,k)>0$ is a constant depending only on $M,k$. Fixing $\eta=\half c^{-1}\frac{r_+^2}{r_m^2}$ (which only depends on $M,k$), we get
\begin{align*}
  \begin{aligned}
    \VV(r) & \geq -C\ellmode^2 - C\ellmode^2(\Ka w) + \frac{1}{16} \Ka (\Ka w)^2.
  \end{aligned}
\end{align*}
Using that $\Ka w \gtrsim_{M,k} 1$ for $r_{\Ka,\eta^{-1}} < r < r_m$, we deduce, for $\Ka^{-1}\ellmode^2$ sufficiently small depending only on $M,k$, that
\begin{align}\label{est:BrKavareprm}
  \VV(r) & \gtrsim_{M,k} \Ka (\Ka w)^2, && \text{for all $r_{\Ka,\eta^{-1}}<r<r_m$.}
\end{align}
For $r>r_m$, we easily have from~\eqref{eq:easypotential}
\begin{align}\label{est:Brm+infty}
  \begin{aligned}
    \VV(r) \gtrsim_{M,k} \Ka^3 \gtrsim_{M,k} (\Ka w)^2,
  \end{aligned}
\end{align}
provided that $\Ka^{-1}\ellmode^2$ is sufficiently small depending on $M,k$. Combining~\eqref{est:BwefrKap2rm},~\eqref{est:BrKavareprm} and~\eqref{est:Brm+infty} we obtain the desired~\eqref{est:BrKa2+inftyinsec}.\\

From a direct observation of the signs in~\eqref{eq:easypotential} and using that $\Ka w \les_{M,k}1$ in $(r_{\Ka,-1},r_{\Ka,+1})$, we have
\begin{align*}
  \VV(r) & \geq -\frac{3}{2}\le(\frac{2}{r}-\frac{6M}{r^2}\ri)\le(\ellmode(\ellmode+1)-\frac{6M}{r}\ri) - 2\Ka w \le(\ellmode(\ellmode+1) - \quar - \frac{9M}{2r}\ri) \\
         & \gtrsim_{M,k} -(1+\Ka w)\ellmode^2 \gtrsim_{M,k} -\ellmode^2,
\end{align*}
for all $r_{\Ka,-1}<r<r_{\Ka,+1}$, which, using~\eqref{eq:defrK0rK2insec}, proves~\eqref{est:BrKa+-rKa++insec} and finishes the proof of Lemma~\ref{lem:easyboundspotential}.

\section{Index of the energy norms}\label{sec:energynorms}
Let $\ellmode\geq 2$ and $|m|\leq\ellmode$ and $\Phi$ a smooth spacetime function. We recapitulate below the definitions of most of the energies which are used in the paper. All energies will be finite if $\Phi$ is regular at the horizon and regular at infinity as in Definition \ref{def:regular}.
\begin{align}\label{eq:defenergynorms}
  \begin{aligned}
    \mathrm{E}^w_{m\ellmode}[\Phi] & := \int_{-\infty}^{\frac{\pi}{2}}\le(w|\pr_{t^\ast}\Phi_{m\ellmode}|^2 + w|\pr_{r^\ast}\Phi_{m\ellmode}|^2 + w\ellmode^2|\Phi_{m\ellmode}|^2\ri)\,\d r^\ast,\\
    \mathring{\mathrm{E}}_{m\ellmode}[\Phi](t^\ast) & := \half \int_{-\infty}^{\frac{\pi}{2}} \le(\frac{\De_-\De_+}{\De_0^2}|\pr_{t^\ast}\Phi_{m\ellmode}|^2 + |\pr_{r^\ast}\Phi_{m\ellmode}|^2 + \frac{\De}{r^2}\le(\frac{\ellmode(\ellmode+1)}{r^2}-\frac{6M}{r^3}\ri)|\Phi_{m\ellmode}|^2\ri) \,\d r^\ast,\\
    \widetilde{\mathrm{E}}_{m\ellmode}[\Phi](t^\ast) & := \half \int_{-\infty}^{\frac{\pi}{2}} \le(\frac{\De_-\De_+}{\De_0^2}|\pr_{t^\ast}\Phi_{m\ellmode}|^2 + |\pr_{r^\ast}\Phi_{m\ellmode}|^2 + \frac{\De}{r^4}\ellmode(\ellmode+1)|\Phi_{m\ellmode}|^2\ri) \,\d r^\ast,\\
    \overline{\mathrm{E}}_{m\ellmode}[\Phi](t^\ast) & := \widetilde{\mathrm{E}}_{m\ellmode}[\Phi](t^\ast) + \int_{-\infty}^{\frac{\pi}{2}}w^{-1}|\pr_{r^\ast}\Phi_{m\ellmode}|^2\,\d r^\ast,\\
    \Einfty_{m\ellmode}[\Phi](t^\ast) & := \lim_{r\to+\infty}\frac{6M}{\ellmode(\ellmode+1)\le(\ellmode(\ellmode+1)-2\ri)} \le[|\pr_{t^\ast}\Phi_{m\ellmode}|^2 + k^2\frac{\ellmode(\ellmode+1)}{2}|\Phi_{m\ellmode}|^2\ri],\\   
    \widetilde{\mathrm{E}}^\HH_{m\ellmode}[\Phi](t^\ast_2;t^\ast_1) & := \lim_{r\to r_+}\int_{t^\ast_1}^{t^\ast_2} |\pr_{t^\ast}\Phi_{m\ellmode}|^2 \,\d t^{\ast},\\
    \overline{\mathrm{E}}^\HH_{m\ellmode}[\Phi](t^\ast_2;t^\ast_1) & := \widetilde{\mathrm{E}}^\HH_{m\ellmode}[\Phi](t^\ast_2;t^\ast_1) + \lim_{r\to r_+}\int_{t^\ast_1}^{t^\ast_2} \ellmode^2|\Phi_{m\ellmode}|^2 \,\d t^{\ast},\\
    \overline{\mathrm{E}}^\II_{m\ellmode}[\Phi](t^\ast_2;t^\ast_1) & := \lim_{r\to+\infty}\int_{t^\ast_1}^{t^\ast_2}\le(|\pr_{t^\ast}\Phi_{m\ellmode}|^2 + |\pr_{r^\ast}\Phi_{m\ellmode}|^2 + k^2\ellmode(\ellmode+1)|\Phi_{m\ellmode}|^2\ri)\,\d t^\ast.
  \end{aligned}
\end{align}
where $\pr_{t^\ast},\pr_{r^\ast}$ are the coordinate vectorfields of the $(t^\ast,r^\ast,\varth,\varphi)$ coordinate system and where $\Phi$ is any spin-$\pm2$-weighted complex function.
We recall from Sections \ref{sec:normsintro} and \ref{sec:RWbound} our consistent conventions
\begin{align} \label{mgc}
  \mathrm{E}[\Phi] = \sum_{\ellmode\geq 2}\sum_{|m|\leq\ellmode} \mathrm{E}_{m\ellmode}[\Phi] , \ \ \ 
\mathrm{E}_{m \ell} [\Phi] = \mathrm{E}[\Phi_{m \ell} e^{\pm i m \varphi} S_{m \ell}(\vartheta)] ,
\end{align}
which allow one to define the ``total energy" from the individual mode energy and \emph{vice versa}. In particular, the above defines the energies (\ref{eq:defenergynorms}) without the subscripts. Conversely, (\ref{mgc}) defines from the Teukolsky and Regge-Wheeler energies $\mathrm{E}^{\mathfrak{T}}[\alt]$ and $  \mathrm{E}^{\mathfrak{R}}[\Psi]$ introduced in (\ref{Tenergy}) and (\ref{Renergy}) the energies $\mathrm{E}_{m \ell}^{\mathfrak{T}}[\alt]$ and $  \mathrm{E}_{m \ell}^{\mathfrak{R}}[\Psi]$. We also recall that for any energy, we have its higher order version, obtained by commuting with $\LL$, $\pr_t$ and $w^{-1}\Lb$, see Section~\ref{sec:normsintro}.
\begin{remark}
From the definitions of the Sobolev norms in Section~\ref{sec:normsintro}, we have for all $t^\ast\in\RRR$ the equivalences
  \begin{align*}
    \overline{\mathrm{E}}[\Phi](t^\ast) & \simeq_{M,k} \norm{\pr_{t^\ast}\Phi}^2_{L^2(\Si_{t^\ast})} + \norm{\Phi}^2_{H^1(\Si_{t^\ast})}, \\
    \mathrm{E}^{\mathfrak{R}}[\Psi](t^\ast)    &\simeq_{M,k} \overline{\mathrm{E}}[\Psi](t^\ast) + \overset{{\scriptscriptstyle\infty}}{\mathrm{E}}[\Psi](t^\ast) ,
\\
    \mathrm{E}^{\mathfrak{T}}[\alt](t^\ast) &\simeq_{M,k} \overline{\mathrm{E}}[\alt^{[+2]}](t^\ast) + \overline{\mathrm{E}}[w^{-2}\alt^{[-2]}](t^\ast).  
\end{align*}
\end{remark}
\bibliographystyle{graf_GR_alpha}
\bibliography{graf_GR}
\end{document}